%% file: example_paper.tex
\documentclass{article}

\usepackage{microtype}
\usepackage{graphicx}
\usepackage{subcaption}
\usepackage{booktabs}
\usepackage{multirow}
\usepackage{amsthm}
\usepackage[table]{xcolor}
\usepackage{placeins}
\usepackage{seqsplit}
\usepackage[percent]{overpic}
\usepackage{xcolor} 
\usepackage{tikz}
\usetikzlibrary{calc}

\usepackage{tabularx,booktabs}
\newcolumntype{Y}{>{\centering\arraybackslash}X} 
\usepackage{booktabs} 
\usepackage{hyperref}
\usepackage[preprint]{icml2026}

\usepackage{amsmath}
\usepackage{amssymb}
\usepackage{mathtools}
\usepackage{amsthm}

\usepackage[capitalize,noabbrev]{cleveref}

\theoremstyle{plain}
\newtheorem{theorem}{Theorem}[section]

\newtheorem{lemma}[theorem]{Lemma}
\newtheorem{corollary}[theorem]{Corollary}
\theoremstyle{definition}
\newtheorem{definition}[theorem]{Definition}

\theoremstyle{remark}
\newtheorem{remark}[theorem]{Remark}

\usepackage[textsize=tiny]{todonotes}

\usepackage{multirow}
\usepackage{pifont}
\usepackage{makecell}
\usepackage{caption}
\usepackage{subcaption}
\usepackage[T1]{fontenc}
\usepackage{amssymb}
\usepackage{rotating}
\usepackage{tabularx}
\usepackage{enumitem}
\usepackage[skins,breakable]{tcolorbox}
\usepackage[table]{xcolor}

\usepackage{lipsum} 
\tcbset{
    userstyle/.style={
        enhanced,
        breakable,           
        before upper=\setlength{\parindent}{0pt}, 
        colback=white,
        colframe=black,
        colbacktitle=gray!20,
        coltitle=black,
        rounded corners,
        sharp corners=north,
        boxrule=0.5pt,
        drop shadow=black!50!white,
        attach boxed title to top left={
            xshift=0mm,
            yshift=-0.5mm
        },
        boxed title style={
            rounded corners,
            size=small,
            colback=gray!20
        }
    },
    userstyle_formal/.style={
        enhanced,
        breakable,           
        before upper=\setlength{\parindent}{0pt}, 
        colback=white,
        colframe=black,
        colbacktitle=gray!20,
        coltitle=black,
        rounded corners,
        sharp corners=north,
        boxrule=0.5pt,
        drop shadow=black!50!white,
        attach boxed title to top left={
            xshift=-2mm,
            yshift=-2mm
        },
        boxed title style={
            rounded corners,
            size=small,
            colback=gray!20
        }
    },
    replystyler/.style={
        enhanced,
        breakable,           
        before upper=\setlength{\parindent}{0pt}, 
        colback=red!15,
        colframe=black,
        colbacktitle=red!40,
        coltitle=black,
        boxrule=0.5pt,
        drop shadow=black!50!white,
        rounded corners,
        sharp corners=north,
        attach boxed title to top right={
            xshift=0mm,
            yshift=-0.5mm
        },
        boxed title style={
            rounded corners,
            size=small,
            colback=red!40
        }
    },
    replystyleg/.style={
        enhanced,
        breakable,           
        before upper=\setlength{\parindent}{0pt}, 
        colback=green!15,
        colframe=black,
        colbacktitle=green!30,
        coltitle=black,
        boxrule=0.5pt,
        drop shadow=black!50!white,
        rounded corners,
        sharp corners=north,
        attach boxed title to top right={
            xshift=0mm,
            yshift=-0.5mm
        },
        boxed title style={
            rounded corners,
            size=small,
            colback=green!40
        }
    },
}
\newtcolorbox{prompt_helpfulness}[1][]{
    userstyle,
    title=Prompt for Helpfulness Reward Model,
    #1
}
\newtcolorbox{prompt_overrefusal}[1][]{
    userstyle,
    title=Prompt for Over-Refusal Judgement,
    #1
}
\newtcolorbox{prompt_safety_reasoning}[1][]{
    userstyle,
    title=Prompt for Safety Reasoning Trace,
    #1
}
\newtcolorbox{prompt_math_500}[1][]{
    userstyle,
    title=MATH500,
    #1
}
\newtcolorbox{prompt_gpqa_mmlu}[1][]{
    userstyle,
    title=GPQA and MMLU,
    #1
}
\newtcolorbox{prompt_human_eval}[1][]{
    userstyle,
    title=HumanEval,
    #1
}
\newtcolorbox{prompt_mbpp}[1][]{
    userstyle,
    title=MBPP,
    #1
}
\newtcolorbox{prompt_jailbreak_pair}[1][]{
    userstyle,
    title=Jailbreak Prompt,
    #1
}
\newtcolorbox{response_original}[1][]{
    replystyler,
    title=Original Response,
    #1
}
\newtcolorbox{response_ps_ant}[1][]{
    replystyleg,
    title=PS-ANT Response,
    #1
}
\newtcolorbox{prompt_repeat}[1][]{
    userstyle,
    title=Prompt for Last Round Thinking Content,
    #1
}
\newtcolorbox{response_repeat}[1][]{
    replystyleg,
    title=PS-ANT's Last Round Thinking Content,
    #1
}
\icmltitlerunning{SemBind: Binding Diffusion Watermarks to Semantics Against Black-Box Forgery Attacks}

\begin{document}

\twocolumn[
  \icmltitle{SemBind: Binding Diffusion Watermarks to Semantics \\ Against Black-Box Forgery Attacks}



  \icmlsetsymbol{equal}{*}

  \begin{icmlauthorlist}
    \icmlauthor{Xin Zhang}{equal,ustc,akl}
    \icmlauthor{Zijin Yang}{equal,ustc,akl}
    \icmlauthor{Kejiang Chen}{ustc,akl}
    \icmlauthor{Linfeng Ma}{ustc,akl}
    \icmlauthor{Weiming Zhang}{ustc,akl}
    \icmlauthor{Nenghai Yu}{ustc,akl}

  \end{icmlauthorlist}

  \icmlaffiliation{ustc}{School of Cyber Science and Technology, University of Science and Technology of China, Anhui, China}
  \icmlaffiliation{akl}{Anhui Province Key Laboratory of Digital Security, Anhui, China}
  
  \icmlcorrespondingauthor{Xin Zhang}{XinZhang1999@mail.ustc.edu.cn}
  \icmlcorrespondingauthor{Zijin Yang}{bsmhmmlf@mail.ustc.edu.cn}
  \icmlcorrespondingauthor{Kejiang Chen}{chenkj@ustc.edu.cn}
  \icmlcorrespondingauthor{Linfeng Ma}{linfengma@mail.ustc.edu.cn}
  \icmlcorrespondingauthor{Weiming Zhang}{zhangwm@ustc.edu.cn}
  \icmlcorrespondingauthor{Nenghai Yu}{ynh@ustc.edu.cn}

  \icmlkeywords{Machine Learning, ICML}
  \vskip 0.3in
]



\printAffiliationsAndNotice{\icmlEqualContribution}

\input{section/0_abstract}
\input{section/1_introduction}
\input{section/2_background}
\input{section/3_method}

\input{section/4_experiment}

\input{section/5_discussion}

\input{section/6_conclusion}


\bibliography{example_paper}
\bibliographystyle{icml2026}

\input{section/7_appendix}

\end{document}

%% file: section/0_abstract.tex
\begin{abstract}
Latent-based watermarks, integrated into the generation process of latent diffusion models (LDMs), simplify detection and attribution of generated images. However, recent black-box forgery attacks, where an attacker needs at least one watermarked image and black-box access to the provider’s model, can embed the provider’s watermark into images not produced by the provider, posing outsized risk to provenance and trust. We propose SemBind, the first defense framework for latent-based watermarks that resists black-box forgery by binding latent signals to image semantics via a learned semantic masker. Trained with contrastive learning, the masker yields near-invariant codes for the same prompt and near-orthogonal codes across prompts; these codes are reshaped and permuted to modulate the target latent before any standard latent-based watermark. SemBind is generally compatible with existing latent-based watermarking schemes and keeps image quality essentially unchanged, while a simple mask-ratio parameter offers a tunable trade-off between anti-forgery strength and robustness. Across four mainstream latent-based watermark methods, our SemBind-enabled anti-forgery variants markedly reduce false acceptance under black-box forgery while providing a controllable robustness--security balance.

\end{abstract}

%% file: section/1_introduction.tex
\section{Introduction}
\label{sec:intro}
Latent diffusion models (LDMs)~\cite{ho2020denoising,sohl2015deep,song2019generative} now generate images that are virtually indistinguishable from real photographs, enabling a wide range of creative and assistive applications. At the same time, such realism raises acute concerns about misleading content and deepfakes~\cite{europol2022facing,goldstein2021disinformation}, which can be used to deceive individuals, sway public opinion, and facilitate fraud.

\begin{figure}[t]
  \centering
  \includegraphics[width=0.98\linewidth]{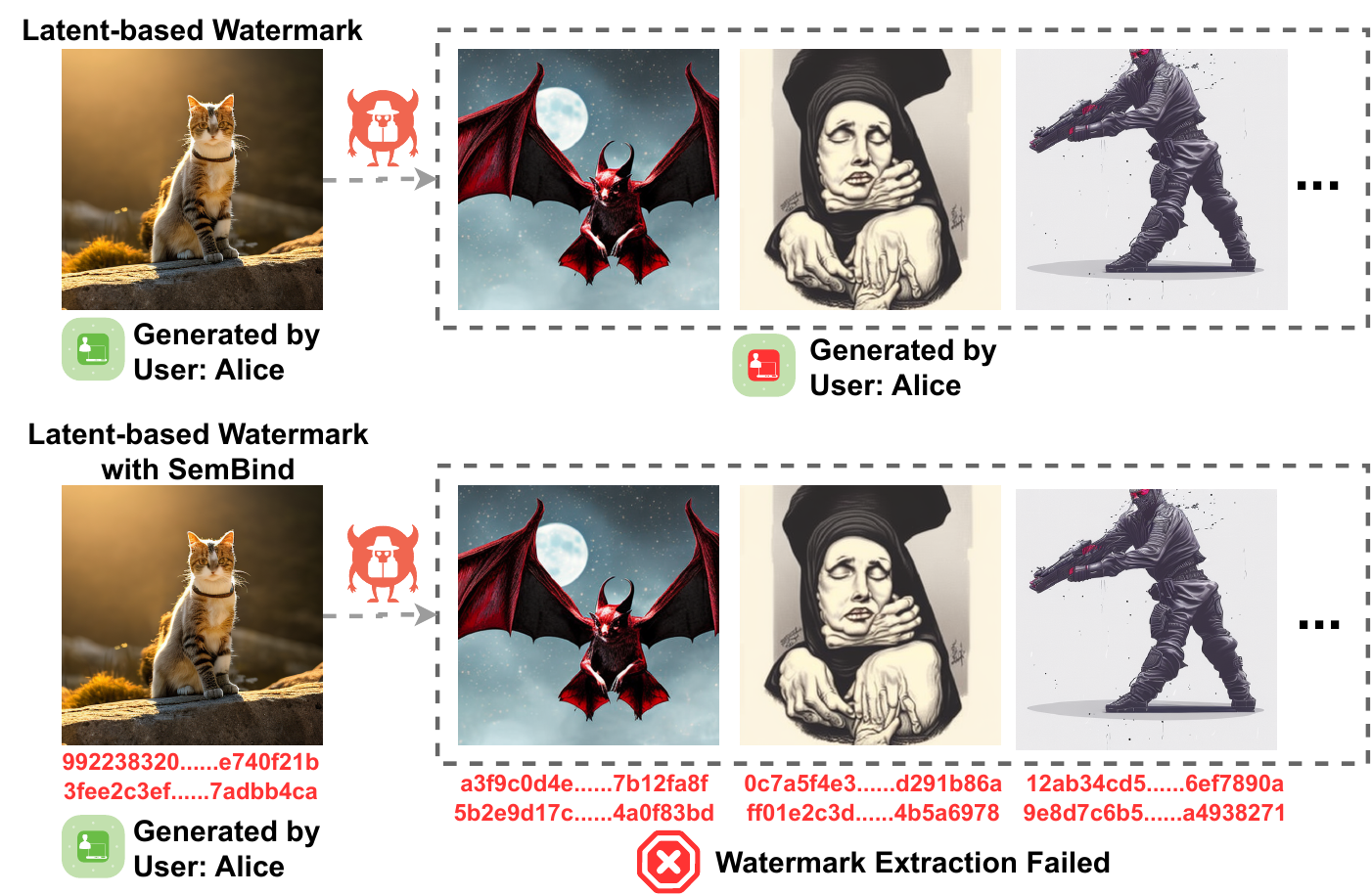}
\caption{Black-box forgery attack and SemBind overview.
Latent-based watermarking embeds a pattern in the initial latent noise, which a black-box attacker can transfer to forged images from at least one watermarked example.
SemBind additionally binds the latent watermark to a semantic bitstring, causing verification to fail when the forged semantics deviate from the original.}
\vspace{-14pt}
  \label{fig:blackbox_forgery_and_sembind}
\end{figure}

Watermarking aims to mitigate these risks by embedding information into generated images for later copyright authentication and provenance tracking. It is already being piloted by governments~\cite{biden2023executive,act2024regulation} and major AI providers~\cite{bartz2023openai,clegg2024labeling} as a key mechanism for responsible deployment.

Watermarking for diffusion models falls into three broad categories. Post-processing methods~\cite{cox2008digital,zhang2019robust} modify robust image features after generation, while fine-tuning methods~\cite{cui2023diffusionshield,fernandez2023stable,xiong2023flexible,zhao2023recipe} amalgamate the watermark embedding process with the image generation process. Latent-based schemes~\cite{wen2023tree,yang2024gaussian,gunnundetectable,yang2025gaussian} instead encode a pattern in the initial noise and recover it by inverting the denoising process. This design keeps the diffusion model unchanged, is typically more robust to image transformations, and, more importantly, has provable undetectability~\cite{christ2024undetectable}. In particular, undetectability implies that the watermarked outputs are statistically indistinguishable from non-watermarked ones, and thus the watermark does not introduce a systematic degradation in generation quality.

However, latent-based watermarks are highly susceptible to black-box forgery attacks~\cite{muller2025black,jain2025forging}. As illustrated in \cref{fig:blackbox_forgery_and_sembind}, an attacker with only black-box access and at least one single watermarked image can transfer the watermark to large volumes of illicit content, undermining both watermark owners and AI service providers.

Within this latent-based paradigm, there is no effective defense against black-box forgery. Naïve countermeasures such as tightening detection thresholds~\cite{muller2025black} are largely ineffective: since the watermark lives purely in the initial latent, both natural perturbations and forgery operations only manifest as modifications to this latent, making them hard to distinguish at verification time. As robustness is improved to tolerate more natural perturbations, the scheme simultaneously becomes more tolerant to forged latents and thus less resistant to black-box forgery. \emph{A key motivation of this work is to strengthen latent-based watermarks against black-box forgery without sacrificing their key advantage: for schemes that admit provable undetectability, we preserve the same guarantee under the same setting.}

In this work, we propose \emph{SemBind}, the first defense framework for latent-based watermarks that resists black-box forgery by binding latent watermark signals to image semantics via a learned semantic masker. The masker is trained contrastively to produce near-invariant codes for images from the same prompt and near-orthogonal codes across prompts. During watermarking, SemBind generates an auxiliary clean image for the target prompt, extracts a semantic code, expands and permutes it under a secret key, and uses the resulting mask to multiplicatively modulate the watermarked latent produced by any standard scheme. This design preserves image quality, while a single \emph{mask-ratio} parameter controls the trade-off between anti-forgery strength and robustness to natural distortions.

We validate SemBind on four representative latent-based schemes—Tree-Ring~\cite{wen2023tree}, Gaussian Shading~\cite{yang2024gaussian}, PRC~\cite{gunnundetectable}, and Gaussian Shading++~\cite{yang2025gaussian}, by instantiating SemBind-enabled variants for each. Our evaluation covers robustness to common perturbations, resistance to imprinting and reprompting attacks, and image quality and semantic alignment measured by FID~\cite{heusel2017gans} and CLIP scores~\cite{radford2021learning}. Across all four schemes, SemBind substantially reduces false acceptance under black-box forgery while preserving watermark robustness and keeping FID and CLIP essentially unchanged, yielding a controllable robustness--security trade-off via the mask ratio.
Moreover, for base schemes that admit provable undetectability, we theoretically prove that SemBind preserves the same undetectability guarantee under the same setting.

In summary, we make the following contributions:
\begin{itemize}
  \item We propose SemBind, the first defense for latent-based diffusion watermarks against black-box forgery. By learning a semantic masker via contrastive learning and introducing a mask-ratio parameter, SemBind binds latent signals to image semantics, providing strong resistance to forgery while enabling a controllable trade-off with watermark robustness.

  \item We instantiate SemBind on four representative latent-based schemes.
For schemes that admit provable undetectability, We theoretically prove that the SemBind-enabled variants preserve the same undetectability guarantee under the same setting, and empirically confirm that FID and CLIP remain on par with the original baselines.

  \item We evaluate SemBind under common perturbations and two canonical black-box forgery strategies (imprinting and reprompting), showing substantially reduced false acceptance while preserving robustness and enabling a tunable robustness–security trade-off.
\end{itemize}

%% file: section/2_background.tex
\section{Related Work}
\label{sec:Related Work}
\subsection{Diffusion Models and Inverse DDIM}
\label{subsec:Diffusion Models and Inverse DDIM}

Diffusion models synthesize images by iteratively denoising a latent variable that is initially drawn from a Gaussian prior. In latent diffusion models (LDMs)~\cite{rombach2022high}, the diffusion process operates in a latent space $\mathcal{Z}$. An encoder $\mathcal{E}$ maps an image $x \in \mathbb{R}^{H\times W\times 3}$ to its latent representation $z_0 = \mathcal{E}(x) \in \mathbb{R}^{h\times w\times c}$, and a decoder $\mathcal{D}$ reconstructs the image as $x = \mathcal{D}(z_0)$. A pretrained LDM therefore consists of the tuple $\Theta = (\mathcal{E}, u, \mathcal{D})$, where $u$ denotes the noise-prediction network (UNet).

Starting from an initial latent $z_T \sim \mathcal{N}(0, I)$, DDIM sampling~\cite{ho2020denoising} runs a deterministic denoising trajectory: at each step $t$ it uses the trained noise predictor $u(z_t, t, C)$ and the noise schedule $\{\alpha_t\}$ to update $z_t$ until reaching a clean latent $z_0$. We denote the full forward denoising process that maps $z_T$ to $z_0$ by $z_0 = \mathcal{G}_{T\to 0}(z_T; u)$.

Conversely, inverse DDIM~\cite{mokady2023null} approximately retraces this trajectory in reverse: given an image latent $z_0$, it iteratively adds noise using the same predictor $u$ and schedules to obtain an estimate $\hat{z}_T$ of the initial noise, which we write compactly as $\hat{z}_T = \mathcal{I}_{0\to T}(z_0; u)$.

\begin{figure}[t]
  \centering
  \includegraphics[width=\linewidth]{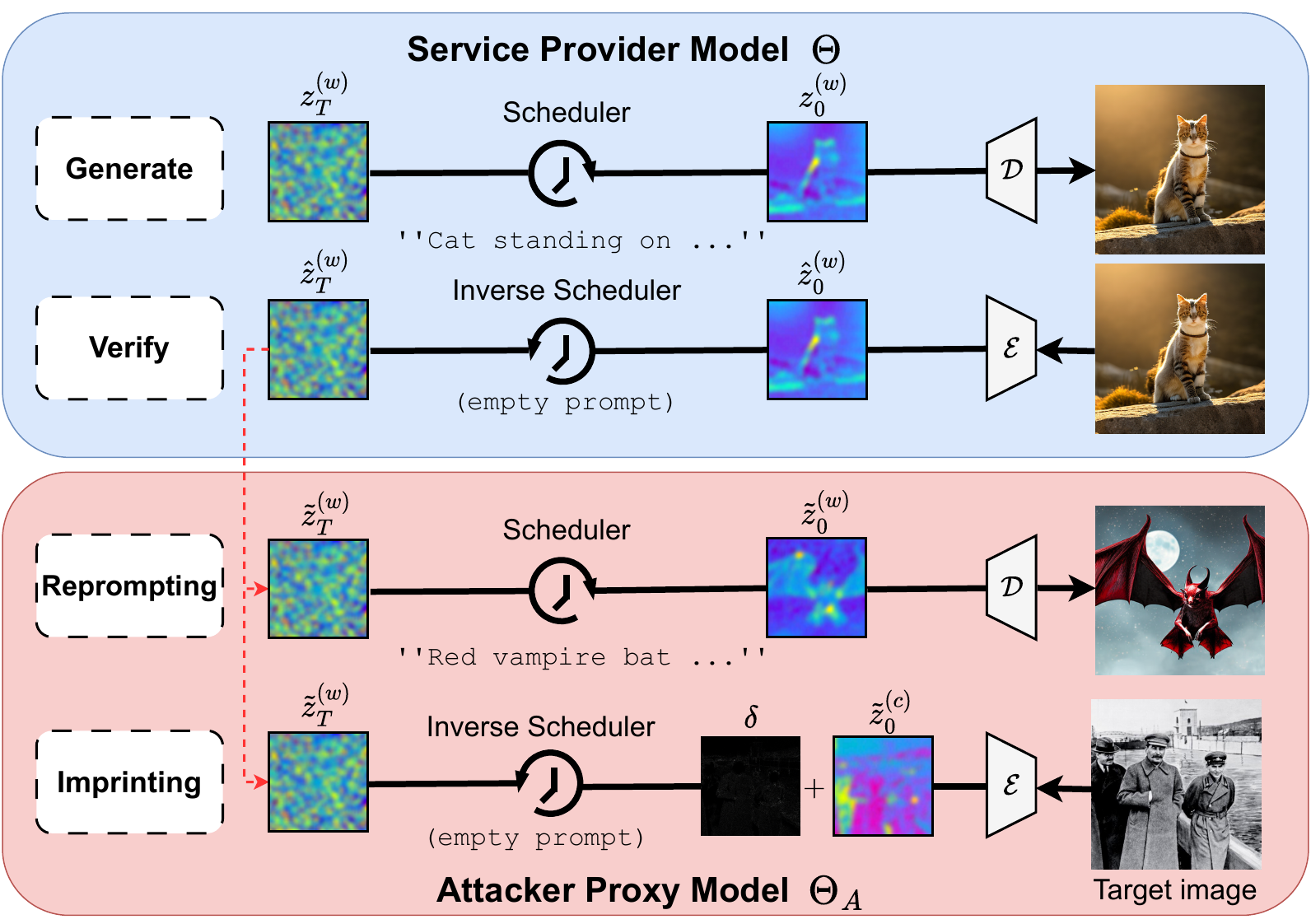}
   \caption{Latent-based watermarking and black-box forgery attacks.}
\vspace{-6pt}
  \label{fig:in processing watermark and forgery attack}
\end{figure}

\subsection{Latent-based Watermark}
\label{subsec:Latent-based Watermark}
In this work we focus on \emph{latent-based} watermarking schemes for diffusion models~\cite{wen2023tree,yang2024gaussian,gunnundetectable,yang2025gaussian,christ2024undetectable}. As illustrated in \cref{fig:in processing watermark and forgery attack}, these schemes encode a message into the initial noise $z_T^{(w)}$ during \emph{Generate}, then approximately recover it by running an inversion scheduler from the generated image during \emph{Verify}. This latent-space design keeps the diffusion backbone unchanged and is compatible with provably undetectable constructions.

We focus on four representative latent-based schemes: Tree-Ring~\cite{wen2023tree}, Gaussian Shading~\cite{yang2024gaussian}, PRC~\cite{gunnundetectable}, and Gaussian Shading++~\cite{yang2025gaussian}.
Tree-Ring embeds a \emph{zero-bit} watermark by imposing a characteristic pattern in the initial latent, and does not provide a provable undetectability guarantee.
Gaussian Shading embeds \emph{multi-bit} watermarks by constraining the initial latent with a secret key and verifying whether the inverted latent remains consistent with this constraint; it admits provable undetectability in the \emph{single-sample} setting.
PRC watermarking encodes messages into key-dependent pseudorandom code patterns in the latent and decodes them after inversion, providing \emph{multi-bit} watermarking with provable undetectability in the \emph{multi-sample} setting.
Gaussian Shading++ combines PRC-protected seed with GS-style payload, providing \emph{multi-bit} watermarking with provable undetectability in the \emph{multi-sample} setting.

\subsection{Black-Box Forgery Attack and Threat Model}
\label{subsec:Watermark Forgery Attack}
Black-box forgery attacks~\cite{muller2025black} are substantially more powerful than earlier ``average'' template attacks~\cite{yang2024can}, which estimate a fixed watermark pattern by aggregating many watermarked samples.
In practice, average attacks are largely ineffective against most latent-based watermarking schemes, whose watermark signals are instance-dependent and cannot be reliably recovered by simple averaging.

Jain \emph{et al.}~\cite{jain2025forging} further improve the computational efficiency of this attack by simplifying the optimization procedure.
However, this efficiency gain comes at the cost of reduced image quality and lower forgery success rates, making it less effective as a universal threat model.
Relatively speaking, the imprinting and reprompting attacks of Müller \emph{et al.} are more powerful and more general; therefore, we adopt them as representative black-box forgery attacks in this work.

\textbf{Threat model.}
In the black-box forgery attack setting, the attacker targets the service provider's watermarked model $\Theta$ to generate watermarked images under unauthorized (potentially malicious) semantics.

In this setting, the service provider holds a private watermarked image generation model $\Theta$, which internally uses a \textit{Generate} procedure to produce watermarked images and a \textit{Verify} procedure to extract the watermark.

The attacker (i) has black-box query access to $\Theta$, i.e., it can submit arbitrary prompts and obtain the generated watermarked images, but has no access to model parameters, gradients, or any intermediate latents during watermark embedding;
(ii) is given at least one watermarked image generated by $\Theta$;
(iii) knows the watermarking algorithm and hyperparameters (Kerckhoffs' principle), but does \emph{not} know the secret keys held by the provider;

The attacker can also holds a proxy diffusion model $\Theta_A$, which in practice is instantiated either as the same backbone as $\Theta$ (the ``match'' case) or as a slightly weaker, publicly available model (the ``mismatch'' case).

The detailed attack procedure is shown in~\cref{fig:in processing watermark and forgery attack}.
The attacker first uses the proxy model $\Theta_A$ to invert the given watermarked image and obtain an estimate $\tilde z_T^{(w)}$ of the provider’s watermarked initial latent $z_T^{(w)}$. Two canonical strategies then arise. In the \emph{imprinting} attack, the attacker additionally supplies a target cover image, typically a semantically unrelated natural image, with latent $\tilde z_0^{(c)}$, and optimizes a small perturbation $\delta$ so that the adversarial latent $\tilde z_0^{(c)} + \delta$ inverts (under $\Theta_A$) to the stolen watermark latent, i.e., $\mathcal{I}_{0\to T}(\tilde z_0^{(c)} + \delta; u_A) \approx \tilde z_T^{(w)}$. In the \emph{reprompting} attack, the attacker simply reuses $\tilde z_T^{(w)}$ as the initial noise and runs the proxy model forward with a different, potentially malicious prompt, generating new images that still carry the provider’s watermark.

%% file: section/3_method.tex
\vspace{-6pt}
\section{Method}
\label{sec:Method}

\begin{figure*}[t]
  \centering
  \includegraphics[width=0.8\linewidth]{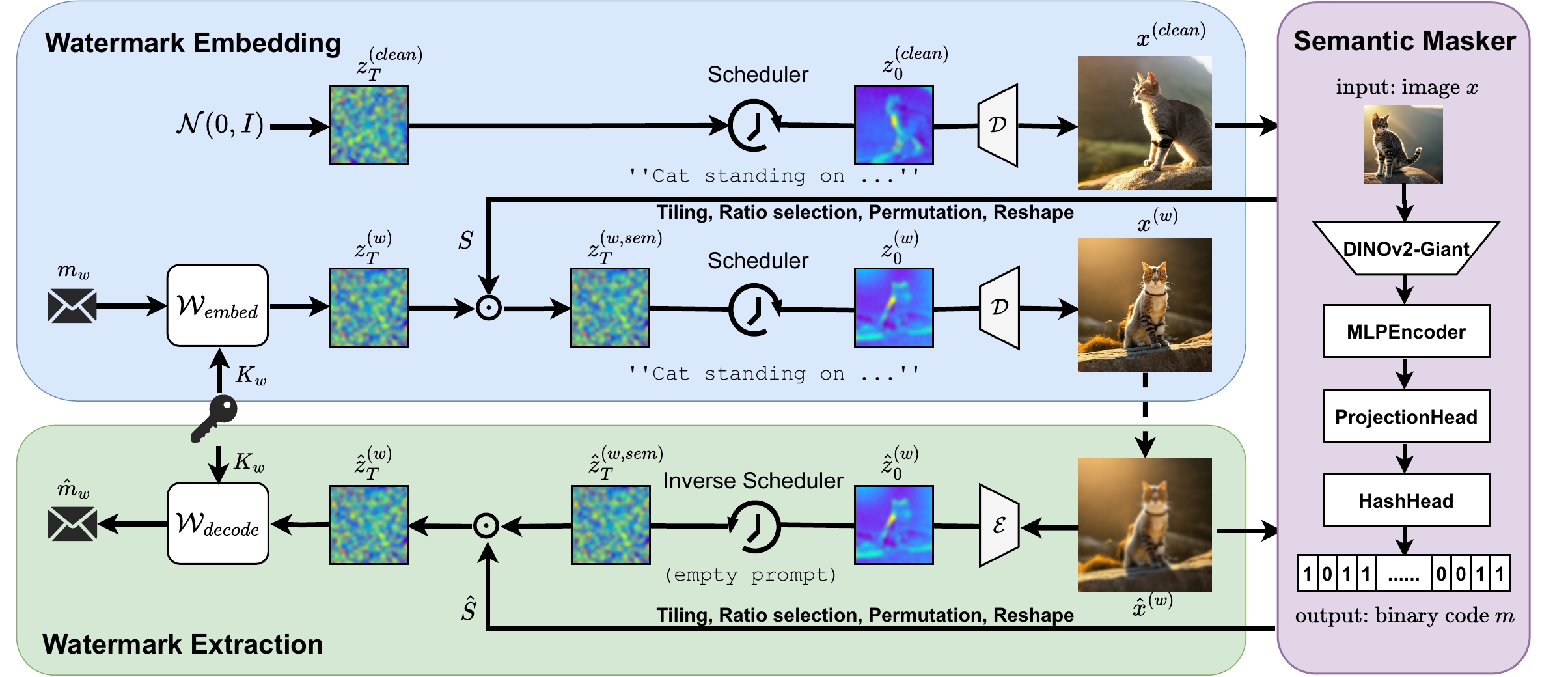}
  \caption{The framework of SemBind, including three
components: semantic masker, embedding procedure, and extraction procedure.}
\vspace{-4pt}
  \label{fig:framework-sembind}
\end{figure*}

\cref{fig:framework-sembind} gives an overview of SemBind, consisting of three components: the semantic masker, the watermark embedding procedure, and the extraction procedure, which we describe in detail below.

Our central idea is to bind latent-based watermarks to image semantics. During embedding, we map an image to a compact binary semantic code, expand it into a latent mask, and use this mask to modulate the watermarked initial latent. During verification, we recompute the code from the watermarked image and rebuild the mask: if the semantics are preserved, then the masks align, and watermark extraction is only slightly perturbed; if the semantics change, the mismatch induces a strong perturbation that makes forged images fail the watermark check.

\vspace{-2pt}
\subsection{Semantic Masker}
\label{sec:semantic-masker}
We introduce a semantic masker $f_\theta$ that maps an image $x$ to a binary code: $m = f_\theta(x) \in \{0,1\}^B$. This semantic masker is trained and kept private by the service provider.

Two properties are satisfied: (i) images generated from the \emph{same} prompt yield near-invariant codes (small Hamming distance), and (ii) images from \emph{different} prompts yield codes that are approximately orthogonal in expectation (Hamming distance $\approx B/2$).  

\vspace{-8pt}
\paragraph{Architecture.}
The masker $f_\theta$ couples a frozen semantic image encoder with a lightweight MLP-based hashing network.
Given an image $x$, a pretrained vision encoder produces a global embedding $e \in \mathbb{R}^{d_e}$ (e.g., the CLS token of a ViT-style backbone).  
This embedding is then processed by three MLP modules: (i) an encoder $\mathrm{Enc}:\mathbb{R}^{d_e}\!\to\!\mathbb{R}^{H}$ composed of residual MLP blocks with batch normalization and GELU; (ii) a projection head $\mathrm{Proj}:\mathbb{R}^{H}\!\to\!\mathbb{R}^{D}$ that outputs an $\ell_2$-normalized representation for contrastive learning; and (iii) a hash head $\mathrm{Hash}:\mathbb{R}^{D}\!\to\!\mathbb{R}^{B}$ implemented as residual fully connected blocks followed by a linear layer that outputs $B$ logits.

The semantic masker $f_\theta$ operates in two modes: a \emph{logit mode} used during training and a \emph{binary mode} used at inference time.
In \emph{logit mode}, given a global embedding $e$, the network outputs a hash logit:
 $\ell = \mathrm{Hash}(\mathrm{Proj}(\mathrm{Enc}(e))) \in \mathbb{R}^B.$

In \emph{binary mode}, we further map $\ell$ to a soft binary code: $b = \tanh(s\,\ell) \in [-1,1]^B$,
and then binarize it into a binary code: $m = \frac{\mathrm{sign}(b)+1}{2} \in \{0,1\}^B$,
where $s>0$ controls the sharpness of the soft sign.

\vspace{-8pt}
\paragraph{Training.}
\label{train details}
We train $f_\theta$ on a large prompt-conditioned corpus of semantic image embeddings.
During training, $f_\theta$ operates in \emph{logit mode}, outputting hash logits for each embedding.
Inspired by the cluster-then-quantize paradigm widely adopted in image hashing~\cite{wang2023deep,shen2024contrastive,wei2024exploring},
we also follow a two-stage routine: first learning a compact and clusterable representation space, and then quantizing it into binary codes.

In \textbf{stage-1}, we optimize $\mathrm{Enc}$ and $\mathrm{Proj}$ using a supervised contrastive loss on the normalized features $z_i \in \mathbb{R}^D$.
The motivation of this stage is to cluster the raw embedding $e$ produced by a pretrained vision encoder, which often exhibits insufficient compact regularization for downstream binary coding.
Although these features are high-dimensional, their \emph{intrinsic} information dimension can be much lower, and the representation space is not explicitly constrained to be compact, leaving substantial redundancy~\cite{zhang2025both}.
Concretely, $\mathrm{Enc}$ and $\mathrm{Proj}$ map $e_i$ into a compact hyperspherical space and encourage prompt-level clustering.

Given a mini-batch $\{(x_i,y_i)\}_{i=1}^N$ with prompt labels $y_i$ and features
$z_i = \mathrm{Proj}(\mathrm{Enc}(e_i))$, we define for each anchor $i$ the set of positives
$P(i) = \{\,p \neq i \mid y_p = y_i\,\}$.  
The \textbf{stage-1} objective is a standard supervised contrastive loss~\cite{khosla2020supervised}:
\begin{equation}
\label{eq:supervised contrastive loss}
  \mathcal{L}_{\mathrm{sup}}
  = - \frac{1}{N} \sum_{i=1}^N\frac{1}{|P(i)|}
    \sum_{p \in P(i)}
    \log
    \frac{\exp\bigl(z_i^\top z_p / \tau\bigr)}
         {\sum\limits_{a \neq i} \exp\bigl(z_i^\top z_a / \tau\bigr)},
\end{equation}
where $\tau>0$ is a temperature.  
This loss encourages features from the same prompt to cluster on the unit sphere while separating different prompts.

In \textbf{stage-2}, we freeze the semantic image encoder and $\mathrm{Enc}$+$\mathrm{Proj}$, and train the hash head $\mathrm{Hash}$
to \emph{quantize} the learned spherical features into binary codes with desirable bit-space properties.
For each sample we compute logits $\ell_i$, codes $b_i=\tanh(s\,\ell_i)$, and use a supervised contrastive loss in the code space,
\begin{equation}
\label{eq:supervised contrastive loss binary}
  \mathcal{L}_{\mathrm{hash}}
  = - \frac{1}{N} \sum_{i=1}^N\frac{1}{|P(i)|}
    \sum_{p \in P(i)}
    \log
    \frac{\exp\bigl( (b_i^\top b_p / B)/\tau_h \bigr)}
         {\sum\limits_{a \neq i} \exp\bigl( (b_i^\top b_a / B)/\tau_h \bigr)},
\end{equation}
where $\tau_h>0$ is a temperature in the hash space.

In addition, we add three regularizers on $\{b_i\}$:
(i) a \emph{quantization} term $\mathcal{L}_{\mathrm{q}} = \mathbb{E}[\,1 - |b_i|\,]$ encouraging $|b_{i,k}| \rightarrow 1$,
(ii) a \emph{bit-balance} term
$\mathcal{L}_{\mathrm{bal}} = \frac{1}{B} \sum_{k=1}^B (\frac{1}{N}\sum_i b_{i,k})^2$ pushing each bit to have zero mean across the batch, and
(iii) a \emph{decorrelation} term
$\mathcal{L}_{\mathrm{dcr}} = \lVert C - I \rVert_F^2$, where $C$ is the sample covariance of $\{b_i\}$.
To further stabilize the codes, we generate two jittered views of each feature and penalize their discrepancy via a bit-consistency loss
$\mathcal{L}_{\mathrm{cons}} = \mathbb{E}[\lVert b_i^{(1)} - b_i^{(2)} \rVert_1]$.

The overall \textbf{stage-2} objective is as follows:
\begin{equation}
  \begin{aligned}
    \mathcal{L}
      = \mathcal{L}_{\mathrm{hash}}
       + \lambda_{\mathrm{q}} \mathcal{L}_{\mathrm{q}}
       + \lambda_{\mathrm{bal}} \mathcal{L}_{\mathrm{bal}} 
       + \lambda_{\mathrm{dcr}} \mathcal{L}_{\mathrm{dcr}}
       + \lambda_{\mathrm{cons}} \mathcal{L}_{\mathrm{cons}}.
  \end{aligned}
\end{equation}

After training, we fix $f_\theta$ and use it as a semantic hashing module within the SemBind framework.

\subsection{Watermark Embedding}
\label{sec:Watermark Embedding}
\begin{figure}[t]
  \centering
  {\captionsetup[sub]{skip=1pt}
  \setlength{\tabcolsep}{2pt}
  \renewcommand{\arraystretch}{1.0}
  \begin{tabular}{c c c c c}
    \multirow{2}{*}{%
      \subcaptionbox{\label{fig:orig}}{%
        \includegraphics[width=0.18\columnwidth]{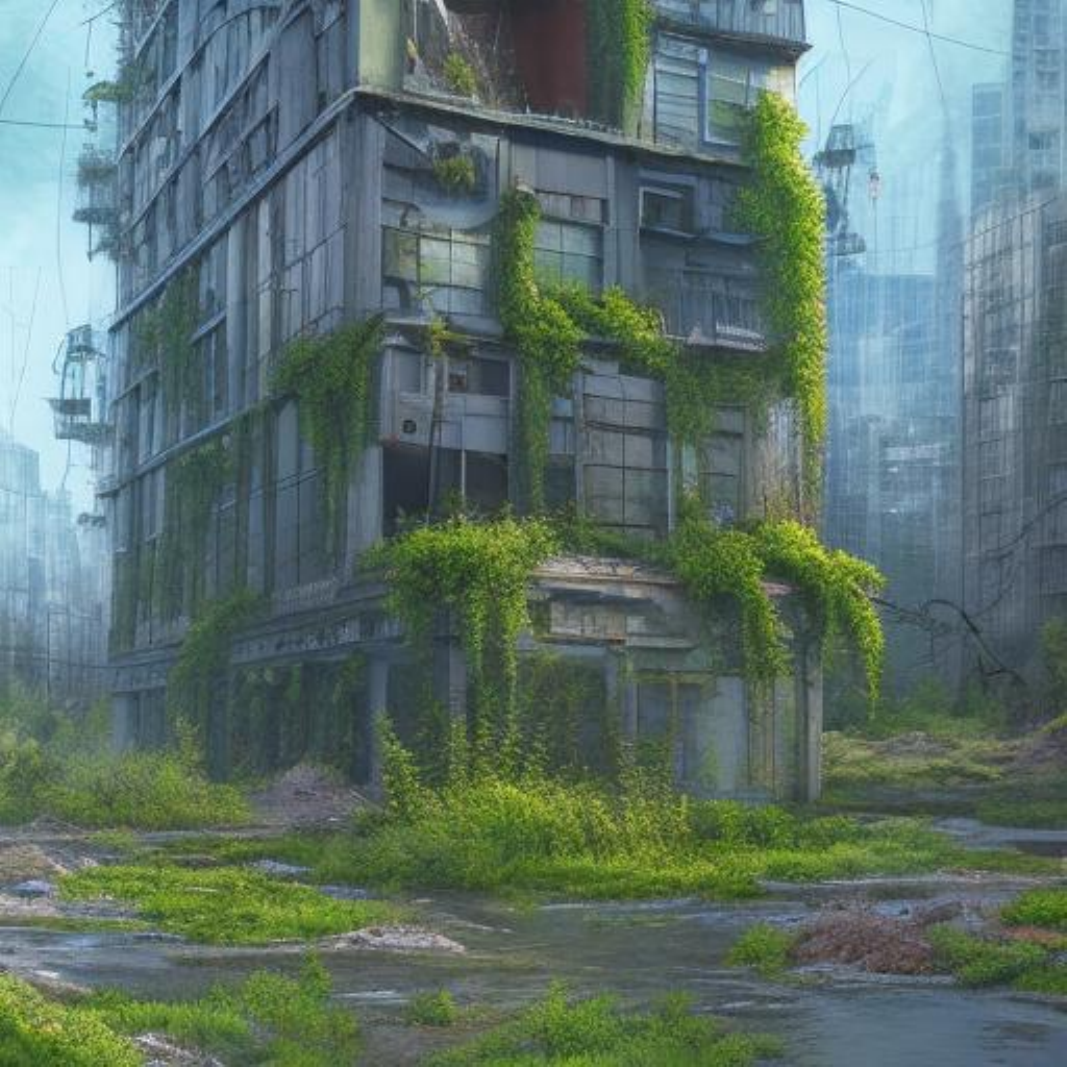}}%
    }
    &
    \subcaptionbox{\label{fig:tr-orig}}{%
      \includegraphics[width=0.18\columnwidth]{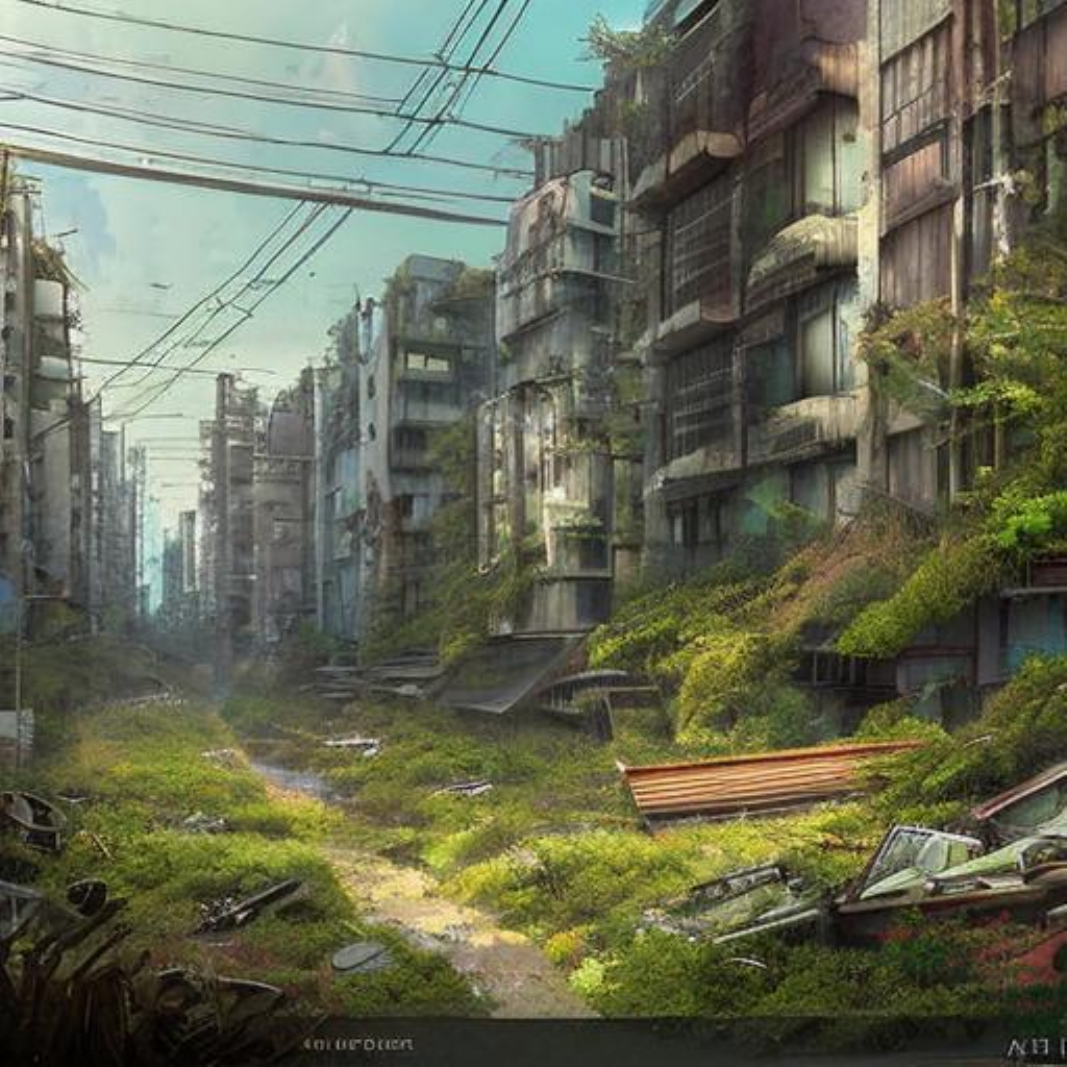}}%
    &
    \subcaptionbox{\label{fig:gs-orig}}{%
      \includegraphics[width=0.18\columnwidth]{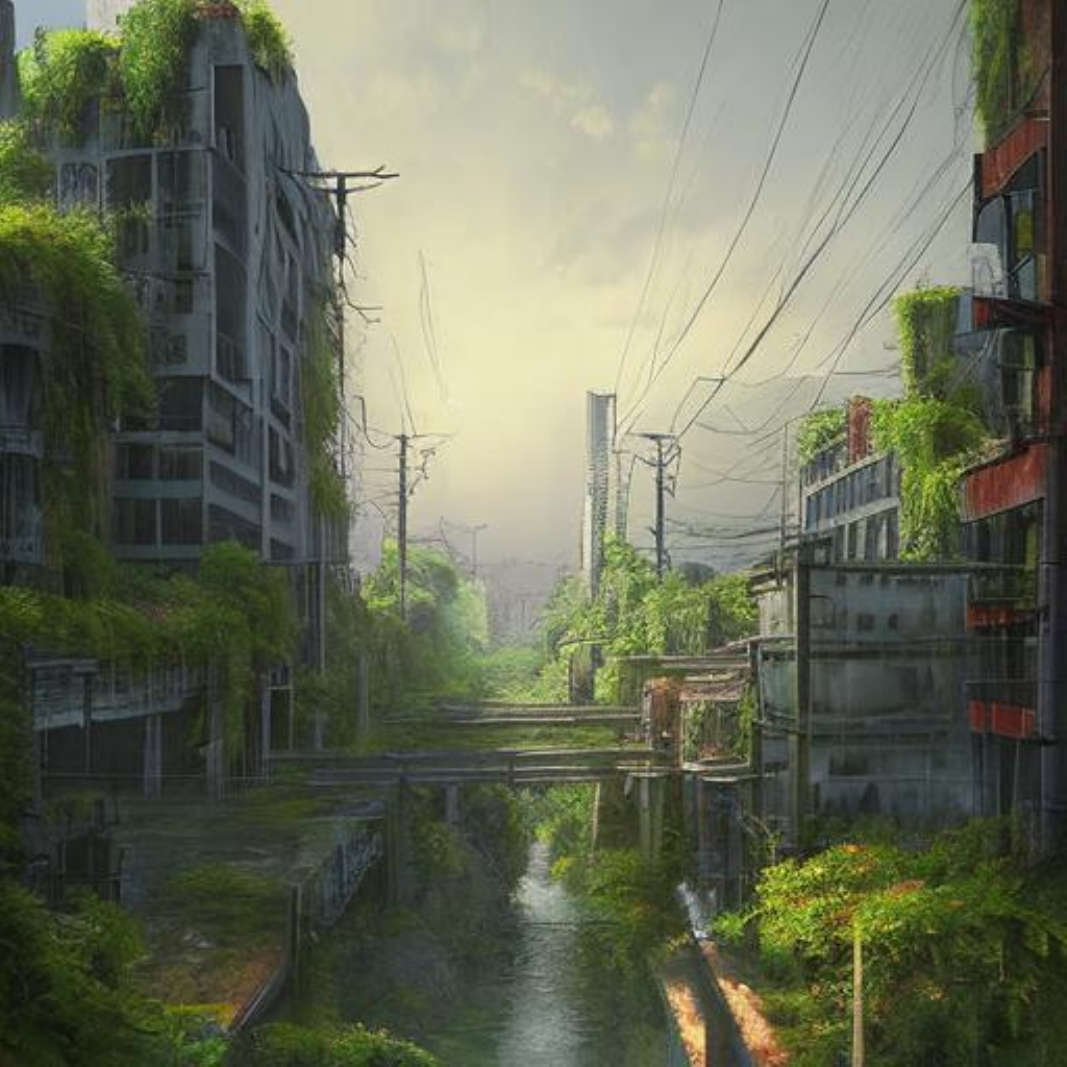}}%
    &
    \subcaptionbox{\label{fig:prc-orig}}{%
      \includegraphics[width=0.18\columnwidth]{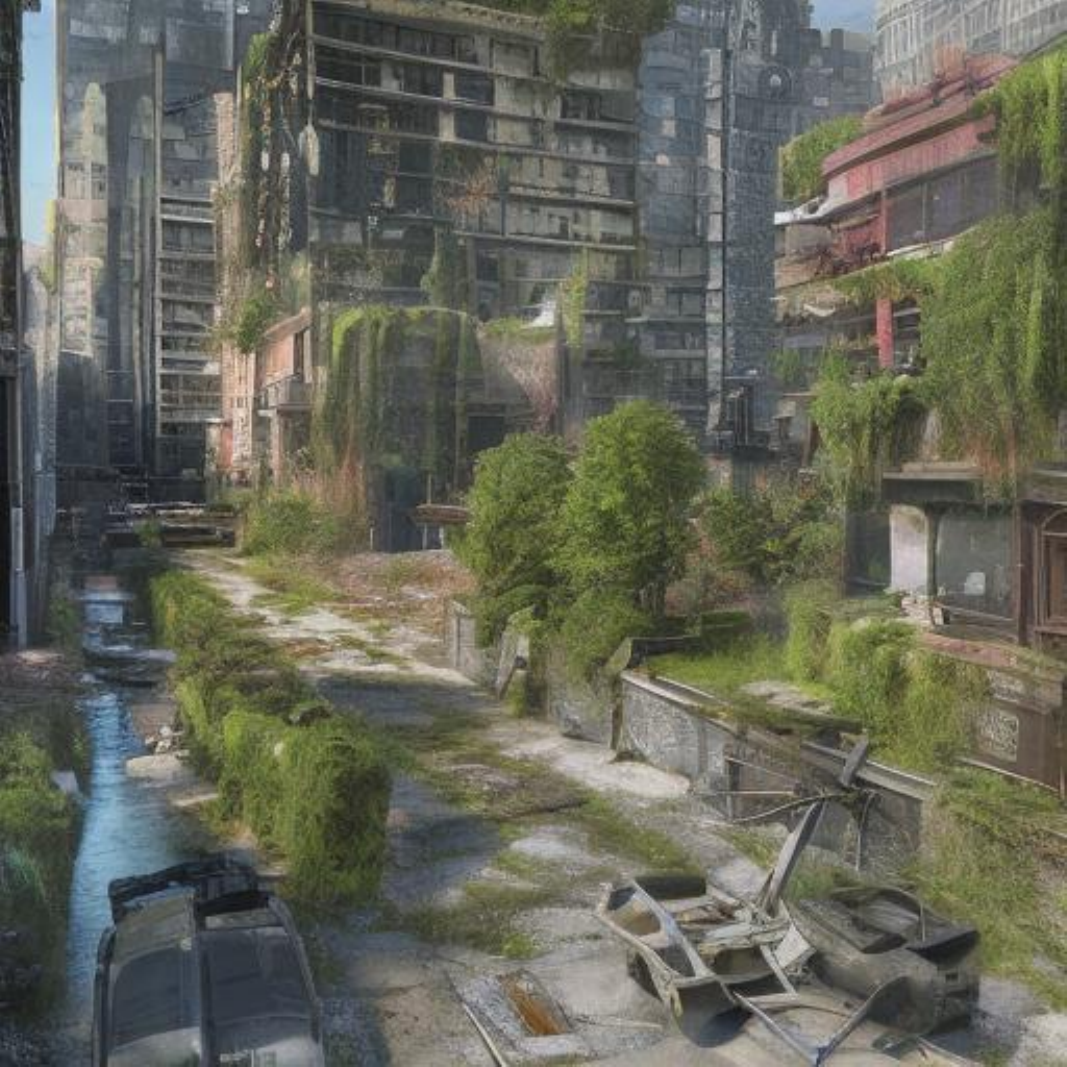}}%
    &
    \subcaptionbox{\label{fig:gspp-orig}}{%
      \includegraphics[width=0.18\columnwidth]{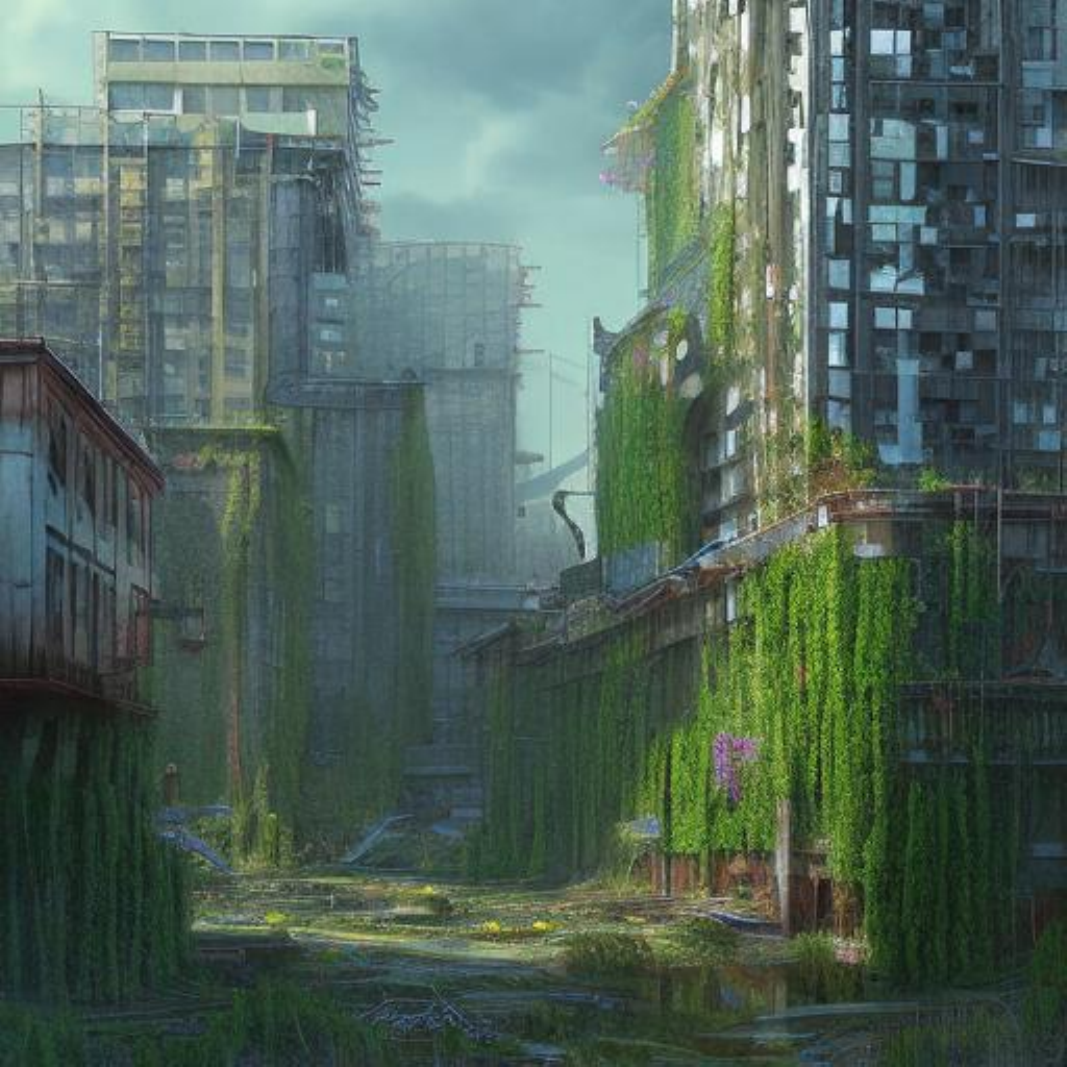}}%
    \\[2pt] 
    &
    \subcaptionbox{\label{fig:tr-sem}}{%
      \includegraphics[width=0.18\columnwidth]{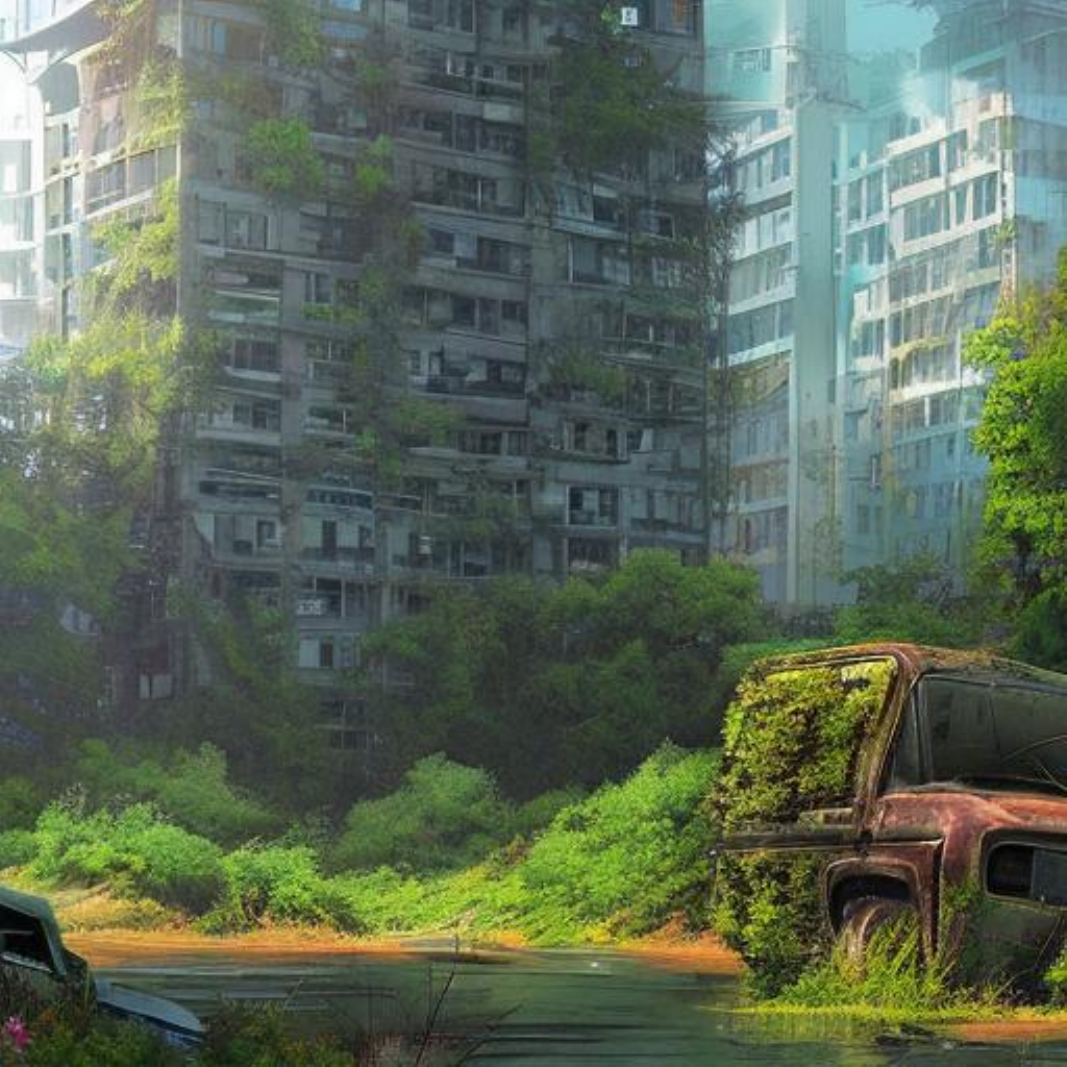}}%
    &
    \subcaptionbox{\label{fig:gs-sem}}{%
      \includegraphics[width=0.18\columnwidth]{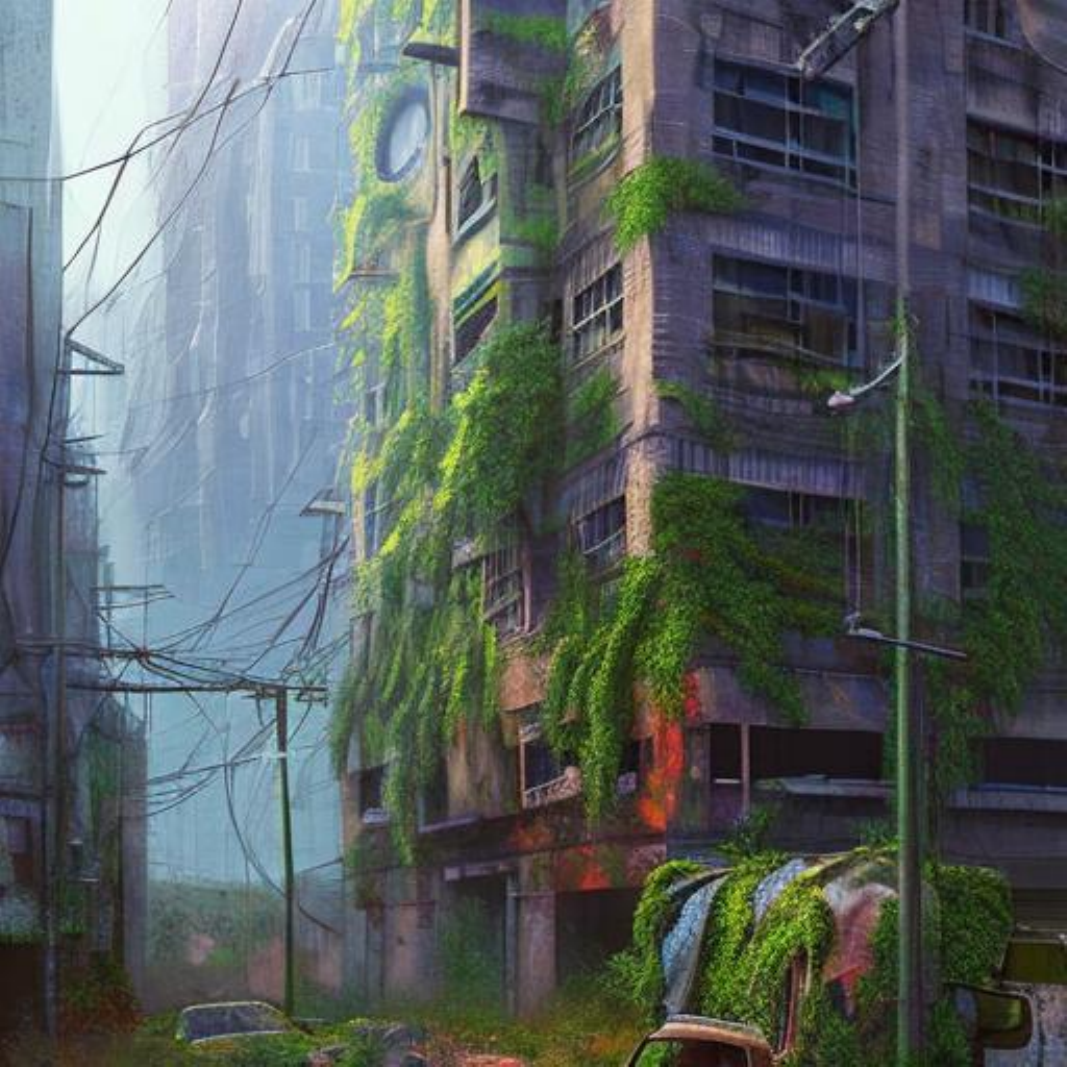}}%
    &
    \subcaptionbox{\label{fig:prc-sem}}{%
      \includegraphics[width=0.18\columnwidth]{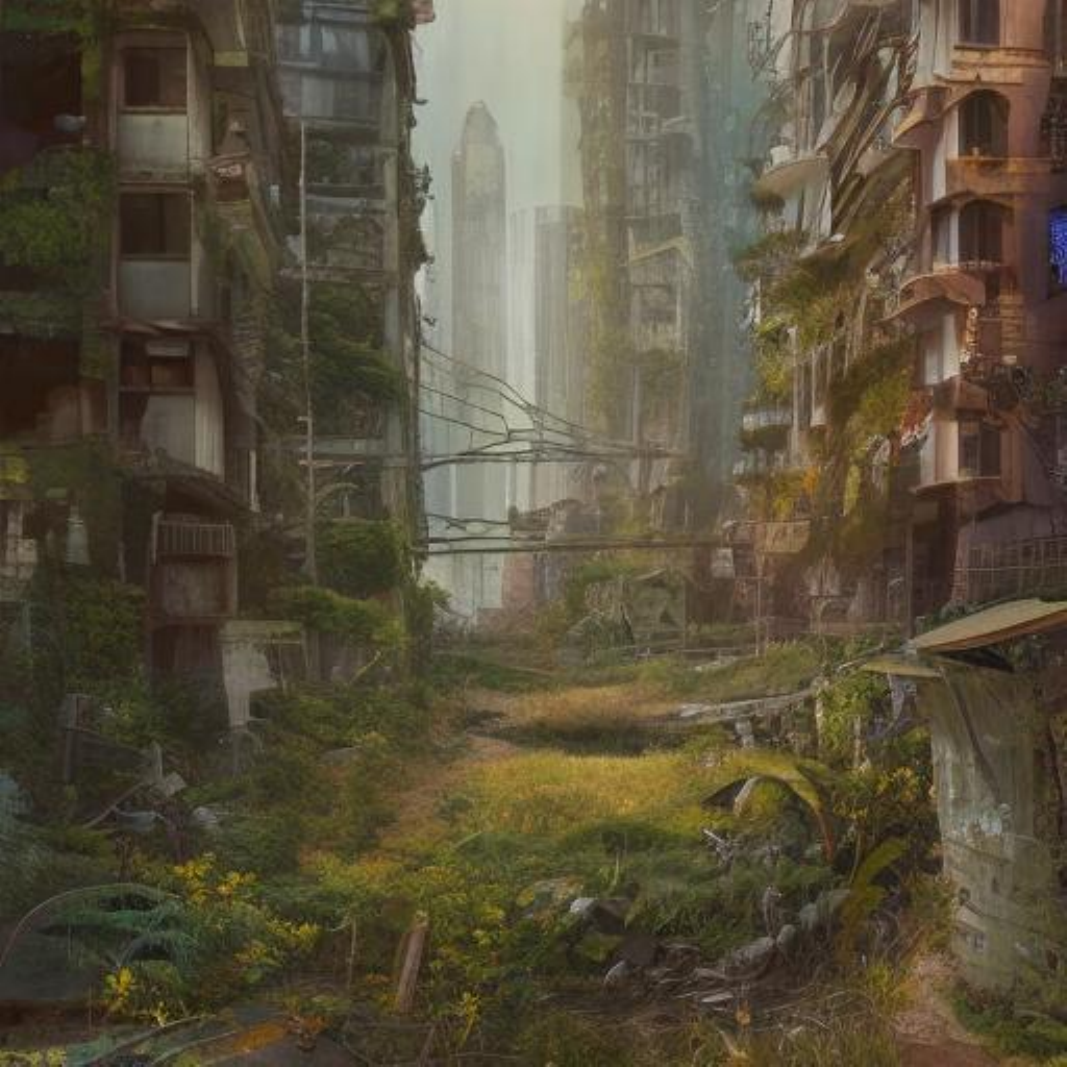}}%
    &
    \subcaptionbox{\label{fig:gspp-sem}}{%
      \includegraphics[width=0.18\columnwidth]{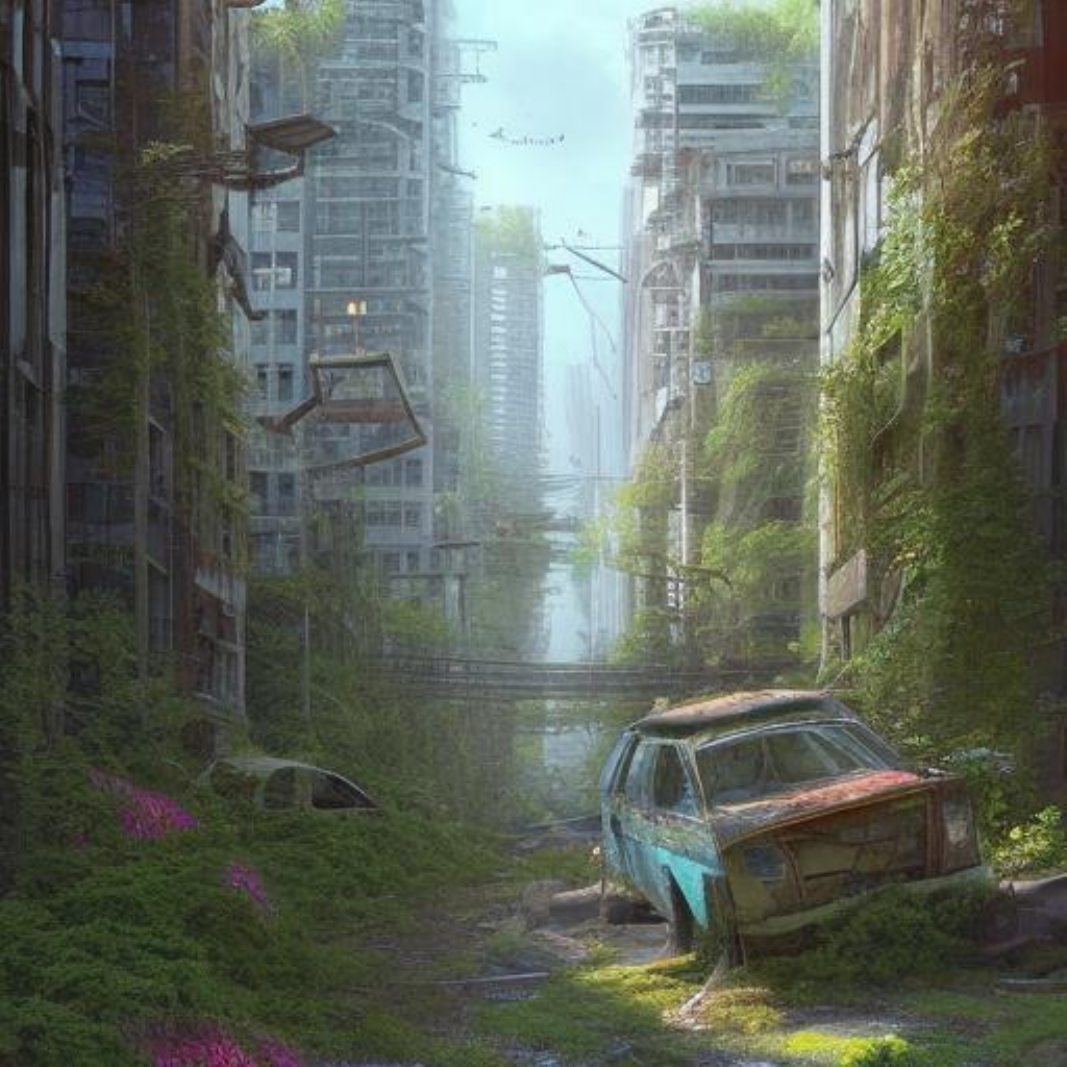}}%
  \end{tabular}
  } 

  \caption{
  Visual comparison of different latent-based watermarking methods and their SemBind-enhanced variants.
  (a) original unwatermarked image;
  (b) Tree-Ring;
  (c) Gaussian Shading;
  (d) PRC;
  (e) Gaussian Shading++;
  (f) Tree-Ring (SemBind);
  (g) Gaussian Shading (SemBind);
  (h) PRC (SemBind);
  (i) Gaussian Shading++ (SemBind).
  All images are generated from the prompt:
  \emph{``Post apocalyptic city overgrown abandoned city, highly detailed, art by Range Murata, highly detailed, 3d, octane render, bright colors, digital painting, trending on artstation, sharp focus.''}
  }
  \vspace{-6pt}
  \label{fig:wm-visual}
\end{figure}

\paragraph{Mask expansion and permutation.}
To bind a watermarked initial latent to the semantic binary code
\(m = f_\theta(x) \in \{0,1\}^B\),
we expand \(m\) into a spatial mask aligned with the shape as the diffusion initial latent via tiling.
A mask ratio \(\sigma\) controls the fraction of latent coordinates being modulated:
larger \(\sigma\) strengthens semantic binding but may reduce robustness to benign perturbations.
Since tiling introduces periodic structure, we further shuffle the mask with a secret permutation key so that the modulation is spatially dispersed and less structured.

Specifically, let the initial latent of the diffusion model be 
\(z_T \in \mathbb{R}^{C\times H\times W}\) with \(L = CHW\) coordinates.
Given a semantic code \(m = f_\theta(x) \in \{0,1\}^B\), a mask ratio
\(\sigma \in [0,1]\), and a \emph{secret} permutation key
\(K_{\text{perm}}\) shared between embedder and verifier, we construct
a spatial binary mask by:
(i) Tiling: form \(\tilde m \in \{0,1\}^L\) by repeating the bits of \(m\) and truncating to length \(L\);
(ii) Ratio selection: let \(L_\sigma=\lfloor \sigma L \rfloor\) and obtain \(\tilde m_\sigma\) by keeping the first \(L_\sigma\) entries of \(\tilde m\) and setting the rest to \(0\);
(iii) Permutation: use \(K_{\text{perm}}\) to define a fixed permutation \(\pi\) over \(\{1,\dots,L\}\) and set \(\tilde m'_\sigma[i]=\tilde m_\sigma[\pi(i)]\) for \(i=1,\dots,L\);
(iv) Reshape: reshape \(\tilde m'_\sigma\) back to \(\{0,1\}^{C\times H\times W}\).

We then map this binary mask to a bipolar \emph{sign mask}: $S(x;\sigma,K_{\text{perm}}) = 1 - 2\,\tilde m'_\sigma \in \{-1,+1\}^{C\times H\times W}$.

\paragraph{SemBind embedding.}
Let \(\mathcal{W}_{\text{embed}}\) denote a latent-based watermarking
scheme that, given a message \(m_w\) and watermark key \(K_w\), produces
a watermarked initial latent $z_T^{(w)} = \mathcal{W}_{\text{embed}}(m_w, K_w)$.

For a given prompt \(p\), the
service provider first randomly generates an auxiliary clean image \(x\) using the same diffusion model \(\Theta\).  
The semantic masker in \emph{binary mode} produces a code \(m = f_\theta(x)\), which is
expanded into a sign mask \(S(x;\sigma,K_{\text{perm}})\).  
We then define the \emph{semantically bound} initial latent by
element-wise multiplication: $z_T^{(w,\text{sem})} = S(x;\sigma,K_{\text{perm}}) \odot z_T^{(w)}$.

The final watermarked image is obtained by running the denoising scheduler from
\(z_T^{(w,\text{sem})}\) down to \(z_0^{(w,\text{sem})}\) and decoding
\(z_0^{(w,\text{sem})}\) with the VAE decoder \(\mathcal{D}\).

\subsection{Watermark Extraction}
\label{sec:Watermark Extraction}

Given a watermarked image $\hat x$, SemBind performs watermark extraction in three steps: we first run the semantic masker to obtain a binary code and its latent mask $\hat S = S(\hat x;\sigma,K_{\text{perm}}) \in \{-1,+1\}^{C\times H\times W}$, where $K_{\text{perm}}$ is the secret permutation key shared with the embedding process; then the image $\hat x$ is passed through the VAE encoder $\mathcal{E}$ to obtain
a latent representation $\hat z_0 = \mathcal{E}(\hat x)$, and inverted with the inversion scheduler to estimate the (semantically bound)
initial noise $\hat z_T^{(w,\mathrm{sem})} = \mathcal{I}_{0\to T}(\hat z_0;u)$; finally, we unbind the semantic modulation by computing $\hat z_T^{(w)} = \hat S \odot \hat z_T^{(w,\mathrm{sem})}$ and feed $\hat z_T^{(w)}$ into the decoder of the underlying watermark scheme to obtain the extracted watermark $\hat m_w = \mathcal{W}_{\text{decode}}\!\left(\hat z_T^{(w)}\right)$.

%% file: section/4_experiment.tex
\vspace{-6pt}
\section{Experiments}
\label{sec:Experiments}
\subsection{Experimental Setup}
\label{sec:Experimental Setup}

In our primary experiments, we focus on text-to-image latent diffusion models and adopt \emph{Stable Diffusion v2.1}\footnote{\url{https://huggingface.co/stabilityai/stable-diffusion-2-1-base}}. All images are generated at a resolution of \(512 \times 512\) with a latent space of \(4 \times 64 \times 64\). During sampling, we use classifier-free guidance with a scale of \(7.5\) and run \(50\) denoising steps using DPMSolver~\cite{lu2022dpm}. For watermark extraction and forgery attack, we perform diffusion inversion scheduler using the exact inversion method of Hong \textit{et al.}~\cite{hong2024exact} to obtain the latent \(z_T\). Unless otherwise specified, we use \(50\) inversion steps and an inverse order of \(0\). 

\vspace{-10pt}
\paragraph{Training details.}
We instantiate the semantic image encoder backbone as DINOv2-Giant~\cite{oquab2023dinov2} and train the Semantic Masker $f_\theta$ on our large prompt-conditioned DINO embedding corpus (SemCon-3M, in Appendix~\ref{sec:dataset_semcon3m}). Full training details are provided in Appendix~\ref{sec:model_arch_train_details}.

\vspace{-10pt}
\paragraph{Watermarking methods and datasets for evaluation.}
We evaluate SemBind on four representative latent-based watermarking schemes:
Tree-Ring~\cite{wen2023tree}, Gaussian Shading~\cite{yang2024gaussian},
PRC watermark~\cite{gunnundetectable}, and Gaussian Shading++~\cite{yang2025gaussian}.
For each scheme, we compare the original baseline with its SemBind-enhanced variant (denoted ``-S'').
Unless stated otherwise, we follow the default settings of the respective papers for both embedding and detection.
For detection, we set the decision threshold to achieve a (theoretical) false positive rate (FPR) of $10^{-6}$ and report the corresponding true positive rate (TPR).
For SemBind, we set the mask ratio to $\sigma=1$ for Tree-Ring and Gaussian Shading, and $\sigma=0.5$ for PRC and Gaussian Shading++.

All the evaluations are performed on two prompt sets: the MS-COCO~\cite{lin2014microsoft} and the Stable
Diffusion Prompts (SDP) dataset\footnote{\url{https://huggingface.co/datasets/Gustavosta/Stable-Diffusion-Prompts}}.  
All experiments are conducted on NVIDIA RTX A6000 GPUs.

\begin{table}[t]
  \centering
  \footnotesize            
  \setlength{\tabcolsep}{3pt} 
  \caption{Quality of watermarked images.}
  \label{tab:fid_clip_main}
  \begin{tabular}{lcccc}
    \toprule
    \multirow{2}{*}{Method} 
      & \multicolumn{2}{c}{FID$\downarrow$} 
      & \multicolumn{2}{c}{CLIP score$\uparrow$} \\
    \cmidrule(lr){2-3}\cmidrule(lr){4-5}
      & value & $t$-value & value & $t$-value \\
    \midrule
    SD 2.1           & $25.23\pm.18$ &   --    & $0.3629\pm.0006$ &   --    \\
    \midrule
    TR       & $25.43\pm.13$ & 2.581   & $0.3632\pm.0006$ & 0.8278  \\
    \rowcolor{gray!30}
    TR-S     & $25.27\pm.20$ & 0.5229  & $0.3633\pm.0010$ & 0.9526  \\
    GS       & $25.20\pm.22$ & 0.3567  & $0.3631\pm.0005$ & 0.6870  \\
    \rowcolor{gray!30}
    GS-S     & $25.28\pm.20$ & 0.5714  & $0.3634\pm.0011$ & 1.055   \\
    PRC      & $25.23\pm.14$ & 0.1772  & $0.3624\pm.0006$ & 0.8176  \\
    \rowcolor{gray!30}
    PRC-S    & $25.24\pm.10$ & 0.2412  & $0.3625\pm.0006$ & 1.136   \\
    GS++     & $25.22\pm.10$ & 0.1597  & $0.3626\pm.0011$ & 0.6228  \\
    \rowcolor{gray!30}
    GS++-S   & $25.23\pm.12$ & 0.3855  & $0.3627\pm.0008$ & 0.8646  \\
    \bottomrule
  \end{tabular}
\end{table}

\vspace{-4pt}
\subsection{Quality Comparison}
\label{sec:Quality}
 
We assess watermarked image quality using FID~\cite{heusel2017gans} and CLIP score~\cite{radford2021learning}.
Following Tree-Ring~\cite{wen2023tree}, we compute FID using $5{,}000$ MS-COCO prompts and paired real images as ground truth.
For CLIP, we use $1{,}000$ prompts from the SDP dataset and their generated images.
For each method, we run 10 trials and report the mean and variance.

To further verify that our method maintains the provable undetectability property of the underlying method, we also perform a $t$-test using the ten sets of data from each method. Specifically, the goal is to confirm that the metrics of the watermarking method and those of SD 2.1 are statistically indistinguishable, where a smaller $t$-value indicates better results.

The experimental results are shown in~\cref{tab:fid_clip_main}. For the non-provable undetectability method Tree-Ring~\cite{wen2023tree}, the addition of SemBind actually \emph{improves} FID performance. This is because the semantic mask is generated independently of the Tree-Ring perturbation, so their multiplicative interaction acts as a mild, semantics–conditioned randomization of the initial latent, enriching its distribution and slightly improving sample quality.

For provable undetectability methods~\cite{yang2024gaussian,gunnundetectable,yang2025gaussian}, the small fluctuations in the $t$-values indicate that SemBind does not introduce any measurable degradation in FID or CLIP, and are fully consistent with the fact that these schemes remain provable undetectability under semantic masking.

We theoretically prove that SemBind preserves the \emph{provable undetectability} of the underlying latent-based watermark, as formalized in the following theorem.

\begin{theorem}[\textbf{Semantic masking preserves provable undetectability (\textit{informal})}]
For latent-based watermarking scheme $\mathcal{W}$ that is provably undetectable in the single-sample (resp.\ multi-sample) setting, its SemBind variant $\mathcal{W}^{\mathrm{sem}}$ remains provably undetectable in the same setting. 
\end{theorem}

Intuitively, these schemes produce an initial watermarked latent that is (computationally) indistinguishable from a standard Gaussian, and SemBind only multiplies this latent by an independent $\{\pm1\}$ sign mask, which leaves the Gaussian distribution invariant.
A formal statement of the security game and the full proof are provided in Appendix~\ref{sec:proof_theorem1}.

\begin{table*}[t]
\centering
\setlength{\tabcolsep}{4pt}
\caption{Imprint forgery attack on the SDP dataset for four latent-based watermarking schemes and their SemBind-enhanced variants.}
\label{tab:imprint_sdp}
\resizebox{\textwidth}{!}{%
\begin{tabular}{l c ccc ccc}
\toprule
 & & \multicolumn{3}{c}{Attacker Model: SD~2.1} & \multicolumn{3}{c}{Attacker Model: SD~1.5} \\
\cmidrule(lr){3-5} \cmidrule(lr){6-8}
Method & Step & Det.$\downarrow$ & Bit Acc.$\downarrow$ & PSNR & Det.$\downarrow$ & Bit Acc.$\downarrow$ & PSNR \\
\midrule
TR      & 50/100/150 & 1.00/1.00/1.00 & ---                  & 23.32/22.12/21.34 & 1.00/1.00/1.00 & ---                  & 22.88/21.85/21.17 \\
\rowcolor{gray!30}
TR-S    & 50/100/150 & 0.01/0.02/0.02 & ---                  & 23.33/22.14/21.38 & 0.01/0.01/0.01 & ---                  & 22.88/21.83/21.16 \\
GS      & 50/100/150 & 1.00/1.00/1.00 & 0.9998/1.0000/1.0000 & 23.36/22.16/21.40 & 1.00/1.00/1.00 & 0.9993/0.9998/0.9998 & 22.89/21.86/21.19 \\
\rowcolor{gray!30}
GS-S    & 50/100/150 & 0.04/0.06/0.07 & 0.4818/0.4843/0.4925 & 23.34/22.14/21.37 & 0.03/0.04/0.05 & 0.4893/0.4918/0.4858 & 22.90/21.85/21.17 \\
PRC     & 50/100/150 & 1.00/1.00/1.00 & 1.0000/1.0000/1.0000 & 23.37/22.17/21.40 & 0.99/1.00/1.00 & 0.9949/1.0000/1.0000 & 22.93/21.89/21.22 \\
\rowcolor{gray!30}
PRC-S   & 50/100/150 & 0.06/0.18/0.25 & 0.5266/0.5559/0.5064 & 23.39/22.19/21.42 & 0.00/0.01/0.02 & 0.5015/0.5067/0.5067 & 22.93/21.90/21.21 \\
GS++    & 50/100/150 & 1.00/1.00/1.00 & 0.9995/0.9999/1.0000 & 23.37/22.18/21.42 & 0.98/0.99/1.00 & 0.9803/0.9881/0.9944 & 22.91/21.88/21.19 \\
\rowcolor{gray!30}
GS++-S  & 50/100/150 & 0.15/0.36/0.49 & 0.5699/0.6720/0.7291 & 23.37/22.17/21.41 & 0.02/0.05/0.09 & 0.5063/0.5222/0.5383 & 22.91/21.88/21.20 \\
\bottomrule
\end{tabular}%
}
\end{table*}

\begin{table*}[t]
\centering
\caption{Imprint forgery attack on the COCO dataset for four latent-based watermarking schemes and their SemBind-enhanced variants.}
\label{tab:imprint_coco}
\resizebox{\textwidth}{!}{%
\begin{tabular}{l c ccc ccc}
\toprule
 & & \multicolumn{3}{c}{Attacker Model: SD~2.1} & \multicolumn{3}{c}{Attacker Model: SD~1.5} \\
\cmidrule(lr){3-5} \cmidrule(lr){6-8}
Method & Step & Det.$\downarrow$ & Bit Acc.$\downarrow$ & PSNR & Det.$\downarrow$ & Bit Acc.$\downarrow$ & PSNR \\
\midrule
TR      & 50/100/150 & 0.99/1.00/1.00 & ---                  & 23.34/22.15/21.38 & 0.99/1.00/1.00 & ---                  & 22.88/21.84/21.15 \\
\rowcolor{gray!30}
TR-S    & 50/100/150 & 0.08/0.08/0.08 & ---                  & 23.33/22.13/21.36 & 0.07/0.07/0.07 & ---                  & 22.86/21.82/21.14 \\
GS      & 50/100/150 & 1.00/1.00/1.00 & 0.9998/0.9999/1.0000 & 23.34/22.15/21.37 & 1.00/1.00/1.00 & 0.9960/1.0000/1.0000 & 22.87/21.83/21.15 \\
\rowcolor{gray!30}
GS-S    & 50/100/150 & 0.09/0.10/0.10 & 0.5084/0.5060/0.5077 & 23.32/22.11/21.34 & 0.08/0.09/0.10 & 0.5078/0.5122/0.5122 & 22.87/21.82/21.13 \\
PRC     & 50/100/150 & 1.00/1.00/1.00 & 1.0000/1.0000/1.0000 & 23.37/22.15/21.38 & 0.99/1.00/1.00 & 0.9941/1.0000/1.0000 & 22.92/21.87/21.18 \\
\rowcolor{gray!30}
PRC-S   & 50/100/150 & 0.08/0.14/0.27 & 0.5444/0.5547/0.6211 & 23.37/22.17/21.40 & 0.05/0.05/0.06 & 0.5241/0.5241/0.5295 & 22.92/21.89/21.19 \\
GS++    & 50/100/150 & 1.00/1.00/1.00 & 0.9994/0.9999/0.9999 & 23.36/22.15/21.37 & 0.99/1.00/1.00 & 0.9845/0.9951/0.9964 & 22.92/21.87/21.17 \\
\rowcolor{gray!30}
GS++-S  & 50/100/150 & 0.15/0.43/0.59 & 0.5050/0.7096/0.7769 & 23.36/22.16/21.39 & 0.05/0.07/0.13 & 0.5166/0.5323/0.5586 & 22.92/21.88/21.19 \\
\bottomrule
\end{tabular}%
}
\end{table*}

\begin{table}[t]
  \centering
  \setlength{\tabcolsep}{3pt}
  \caption{Reprompt forgery attack for four latent-based watermarking schemes and their SemBind-enhanced variants.}
  \label{tab:reprompt}
  \scriptsize
  \resizebox{1.0\linewidth}{!}{%
  \begin{tabular}{lcccccccc}
    \toprule
    & \multicolumn{4}{c}{SD 2.1 attacker} & \multicolumn{4}{c}{SD 1.5 attacker} \\
    \cmidrule(lr){2-5} \cmidrule(lr){6-9}
    & \multicolumn{2}{c}{SDP} & \multicolumn{2}{c}{COCO} 
    & \multicolumn{2}{c}{SDP} & \multicolumn{2}{c}{COCO} \\
    Method 
      & Det.$\downarrow$ & Bit Acc.$\downarrow$
      & Det.$\downarrow$ & Bit Acc.$\downarrow$
      & Det.$\downarrow$ & Bit Acc.$\downarrow$
      & Det.$\downarrow$ & Bit Acc.$\downarrow$ \\
    \midrule
    TR     & 0.99 & ---    & 1.00 & ---    & 1.00 & ---    & 1.00 & ---    \\
    \rowcolor{gray!30}
    TR-S   & 0.54 & ---    & 0.14 & ---    & 0.48 & ---    & 0.18 & ---    \\
    GS     & 1.00 & 0.9887 & 0.99 & 0.9788 & 1.00 & 0.9895 & 1.00 & 0.9823 \\
    \rowcolor{gray!30}
    GS-S   & 0.60 & 0.7054 & 0.06 & 0.4894 & 0.56 & 0.6867 & 0.10 & 0.5025 \\
    PRC    & 0.95 & 0.9745 & 0.93 & 0.9662 & 0.94 & 0.9711 & 0.93 & 0.9656 \\
    \rowcolor{gray!30}
    PRC-S  & 0.51 & 0.7338 & 0.50 & 0.7287 & 0.12 & 0.5622 & 0.11 & 0.5571 \\
    GS++   & 0.91 & 0.9527 & 0.84 & 0.9132 & 0.69 & 0.8377 & 0.57 & 0.7732 \\
    \rowcolor{gray!30}
    GS++-S & 0.43 & 0.7044 & 0.25 & 0.6150 & 0.15 & 0.5679 & 0.03 & 0.5107 \\
    \bottomrule
  \end{tabular}%
  }
\end{table}

\label{sec:ROBUSTNESS}
\begin{table*}[t]
  \centering
  \small
  \setlength{\tabcolsep}{4pt}
  \caption{Robustness test on \textsc{sdp}. ``Average (Distortion)" represents the mean value across all distortion types.}
  \label{tab:robust_sdp}
  \resizebox{\textwidth}{!}{%
  \begin{tabular}{lcccccccc}
    \toprule
    Method & None (Det./Acc.) & JPEG (QF=70) & Brightness (×1.0) & GauBlur ($r$=3) & GauNoise ($\sigma$=0.01) & MedFilter ($k$=7) & Resize (×0.5) & Average (Distortion)\\
    \midrule
    TR     & 1.00 & 0.98 & 1.00 & 1.00 & 1.00 & 1.00 & 1.00 & 0.997\\
    \rowcolor{gray!30}
    TR-S   & 1.00 & 0.98 & 1.00 & 0.99 & 0.99 & 0.99 & 1.00 & 0.992 \\
    GS     & 1.00/1.0000 & 0.99/0.9996 & 1.00/0.9966 & 1.00/0.9966 & 1.00/0.9989 & 1.00/0.9984 & 1.00/0.9999 & 0.998/0.99833\\
    \rowcolor{gray!30}
    GS-S   & 1.00/0.9982 & 0.99/0.9871 & 0.99/0.9888 & 0.99/0.9796 & 0.99/0.9894 & 1.00/0.9884 & 1.00/0.9958 &0.993/0.98818\\
    PRC    & 1.00/1.0000 & 1.00/1.0000 & 1.00/1.0000 & 1.00/1.0000 & 1.00/1.0000 & 0.97/0.9807 & 1.00/1.0000 &0.995/0.99678\\
    \rowcolor{gray!30}
    PRC-S  & 1.00/1.0000 & 0.99/0.9902 & 0.98/0.9850 & 0.86/0.9271 & 0.98/0.9896 & 0.83/0.9250 & 1.00/1.0000 &0.940/0.96948\\
    GS++   & 1.00/1.0000 & 1.00/0.9973 & 0.97/0.9841 & 0.98/0.9765 & 0.93/0.9599 & 0.96/0.9707 & 1.00/0.9988 &0.973/0.98122\\
    \rowcolor{gray!30}
    GS++-S & 1.00/0.9998 & 0.96/0.9752 & 0.96/0.9785 & 0.94/0.9466 & 0.89/0.9402 & 0.85/0.9138 & 0.99/0.9932 &0.932/0.95792\\
    \bottomrule
  \end{tabular}%
  }
\end{table*}

\begin{table*}[t]
  \centering
  \small
  \setlength{\tabcolsep}{4pt}
  \caption{Robustness test on \textsc{coco}. ``Average (Distortion)" represents the mean value across all distortion types.}
  \label{tab:robust_coco}
  \resizebox{\textwidth}{!}{%
  \begin{tabular}{lcccccccc}
    \toprule
    Method & None (Det./Acc.) & JPEG (QF=70) & Brightness (×1.0) & GauBlur ($r$=3) & GauNoise ($\sigma$=0.01) & MedFilter ($k$=7) & Resize (×0.5) & Average (Distortion)\\
    \midrule
    TR     & 1.00 & 1.00 & 1.00 & 1.00 & 1.00 & 1.00 & 1.00 &1.000\\
    \rowcolor{gray!30}
    TR-S   & 1.00 & 1.00 & 0.99 & 1.00 & 1.00 & 1.00 & 1.00 &0.998\\
    GS     & 1.00/1.0000 & 1.00/0.9980 & 1.00/0.9984 & 1.00/0.9967 & 1.00/0.9988 & 1.00/0.9980 & 1.00/1.0000 &1.000/0.99832\\
    \rowcolor{gray!30}
    GS-S   & 1.00/0.9996 & 1.00/0.9934 & 1.00/0.9915 & 1.00/0.9830 & 0.99/0.9896 & 1.00/0.9871 & 1.00/0.9964 &0.998/0.99017\\
    PRC    & 1.00/1.0000 & 0.99/0.9948 & 0.95/0.9751 & 0.96/0.9854 & 1.00/1.0000 & 0.91/0.9643 & 1.00/1.0000 &0.968/0.98660\\
    \rowcolor{gray!30}
    PRC-S  & 1.00/1.0000 & 0.97/0.9855 & 0.91/0.9593 & 0.84/0.9261 & 0.99/0.9945 & 0.80/0.8910 & 1.00/1.0000 &0.918/0.95940\\
    GS++   & 1.00/1.0000 & 0.97/0.9796 & 0.98/0.9877 & 0.97/0.9700 & 0.87/0.9259 & 0.95/0.9598 & 1.00/0.9980 &0.957/0.97017\\
    \rowcolor{gray!30}
    GS++-S & 1.00/0.9998 & 0.96/0.9752 & 0.96/0.9757 & 0.92/0.9388 & 0.84/0.9146 & 0.89/0.9279 & 0.98/0.9864 &0.925/0.95310\\
    \bottomrule
  \end{tabular}%
  }
\end{table*}

\vspace{-8pt}
\subsection{Defense Against Forgery Attack}
\label{sec:Defense Against Forgery Attack}

We evaluate SemBind under the two canonical black-box forgery attacks introduced by M\"uller et al.~\cite{muller2025black}: the \emph{imprinting attack} and the \emph{reprompting attack}.

\vspace{-10pt}
\paragraph{Imprinting attack.}
We consider two attacker models: \emph{Stable Diffusion v2.1} (the ``match'' case) and \emph{Stable Diffusion v1.5} (the ``mismatch'' case). For each watermarking scheme and attacker model, we generate 100 watermarked images on both SDP and COCO prompt validation sets. As target images for the attacker, we randomly sample 100 natural photographs from the COCO natural image validation set. Following~\cite{muller2025black}, the attacker optimizes the target latents for 150 gradient steps with a learning rate of $0.01$, and we probe watermark detection every 50 steps.

\cref{tab:imprint_sdp} and~\cref{tab:imprint_coco} report the imprinting results. When the attacker uses SD~2.1 (the ``match'' case), SemBind already results in substantial gains: for Tree-Ring and Gaussian Shading, it reduces the detection rate (Det.) from essentially $100\%$ to below $10\%$ on both SDP and COCO, while the bit accuracy for GS decreases from nearly perfect ($\approx 1.0$) to around $0.5$, i.e., close to random guessing. For PRC and Gaussian Shading++, SemBind again cuts Det.\ by more than half and lowers Bit Acc.\ by over $20\%$. 

When the attacker instead uses SD~1.5 (the ``mismatch'' case), SemBind becomes even more effective: across all four watermarking schemes, Det.\ is driven to very low values and Bit Acc.\ is pushed much closer to $0.5$, corresponding to near-random decoding. This indicates that SemBind’s defensive advantage increases as the attacker’s model becomes more mismatched to the defended generator, yielding near-perfect protection against imprinting in this setting. This scenario is also closer to practical real-world deployments, where attackers typically have access only to the mismatched diffusion models rather than the exact one.

\vspace{-12pt}
\paragraph{Reprompting attack.}
For each of the four latent-based watermarking schemes and their SemBind-enhanced variants, we use the same SDP and COCO \emph{prompt-evaluation} sets, each containing 100 prompts.
We use mismatched text prompts sampled from the Inappropriate Image Prompts (I2P) dataset\footnote{\url{https://huggingface.co/datasets/AIML-TUDA/i2p}}, with 100 mismatched prompts in total, for attack.
As attacker models, we again consider both \emph{Stable Diffusion v2.1} and \emph{Stable Diffusion v1.5}.

\cref{tab:reprompt} summarizes the reprompting results.
Across all watermarking schemes, SemBind consistently and substantially attenuates forgery: for every setting, the SemBind variants reduce the detection rate (Det.) by more than $40\%$ and lower bit accuracy (Bit Acc.) by at least $25\%$ compared to the original watermark baselines.
The gains are particularly pronounced on COCO, where Det.\ often decreases to $20\%$ or less and Bit Acc.\ approaches $0.5$, corresponding to near-random message recovery.
As in the imprinting case, SemBind is even more effective against the mismatch-case attacker (SD~1.5), where forged images are rarely accepted and recovered payloads are close to random.

\vspace{-6pt}
\subsection{Robustness}
\label{robustness main paper}
\vspace{-2pt}
To evaluate the robustness of the method, we selected several common image manipulations: (a) JPEG Compression, $QF=70$ (JPEG). (b) Brightness, $factor=1$. (c) Gaussian Blur, $radius= 3$ (GauBlur). (d) Gaussian Noise, $\mu = 0$, $\sigma = 0.01$ (GauNoise). (e) Median Filtering, $kernel\_size=7$ (MedFilter). (f) 50\% Resize and restore (Resize). We conducted experiments on the MS-COCO and SDP datasets respectively, testing 100 watermarked images under each type of distortion. For each method, we report the metrics under each distortion type and compute the average across all distortions. The experimental results are presented in \cref{tab:robust_sdp} and \cref{tab:robust_coco}.

Overall, SemBind causes only a mild robustness drop across schemes. The impact on Tree-Ringand Gaussian Shading is particularly small, and their detection performance remains close to the original baselines under all tested distortions. For PRC and Gaussian Shading++, we observe a moderate but still acceptable degradation: the true positive rate decreases by about 3--5\% on average and bit accuracy drops by about 2--3\%, with the most visible gaps under  Gaussian blur and median filtering. This trend is consistent with the underlying robustness of the base watermarks: schemes that are robust to common perturbations remain robust after adding SemBind, while schemes with weaker baseline robustness exhibit larger drops in their SemBind-enabled variants.




%% file: section/5_discussion.tex
\vspace{-6pt}
\section{Discussion}
\label{sec:Limitations}
SemBind introduces additional training cost to learn the semantic masker.
In typical platform settings, however, this is a one-time investment: since it interfaces with the system only through a semantic code and a masking operation in the initial latent, upgrading the underlying latent watermark scheme does not require retraining the masker.

SemBind requires generating an auxiliary clean image at deployment time to extract the semantic code used for masking. But this additional generation step is necessary to preserve provable undetectability: maintaining undetectability typically requires embedding to occur entirely in the initial latent space, yet at that point no image is available to infer semantics from. A possible solution is to start from the prompt, but this faces the challenge of aligning the prompt with the generated image, and currently there is no satisfactory method that achieves this reliably.

Moreover, we defer further discussion of adaptive attackers to Appendix~\ref{app:adaptive_attack}, and additional generalization experiments to Appendix~\ref{app:generalization}.
Our experiments suggest that SemBind generalizes well across different models and prompt distributions, and that adaptive attackers cannot easily train a surrogate semantic masker or reliably spoof it in the pixel domain.

%% file: section/6_conclusion.tex
\vspace{-4pt}
\section{Conclusion}
\label{sec:Conclusion and Future Work}
We propose \emph{SemBind}, the first defense for latent-based diffusion watermarks against black-box forgery by binding latent watermark signals to image semantics via a learned semantic masker. SemBind applies broadly to existing latent-based schemes. It reduces false acceptance under imprinting and reprompting, and preserves provable undetectability whenever the underlying scheme has it.

\section*{Impact Statement}
This paper studies defenses for latent-based watermarking in diffusion models, with the goal of improving provenance tracking and copyright authentication in the presence of black-box forgery attacks. If deployed responsibly, our method can help AI service providers and content platforms better identify and deter the unauthorized generation of misleading or harmful content by preventing attackers from transferring a valid watermark to illicit semantics. We believe the primary impact of this work is positive: it advances the reliability of watermark-based provenance for generative models while making explicit the assumptions and limitations under which such defenses are effective.

%% file: section/7_appendix.tex
\newpage
\appendix
\onecolumn

\clearpage
\setcounter{page}{1}

\section{Dataset Construction (SemCon-3M)}
\label{sec:dataset_semcon3m}
We next describe how \textbf{SemCon-3M} is built. Constrained by computational resources, we first mine and cluster prompts from the MS-COCO-2017~\cite{lin2014microsoft} and the Stable
Diffusion Prompts (SDP) dataset\footnote{\url{https://huggingface.co/datasets/Gustavosta/Stable-Diffusion-Prompts}} to form a compact and representative prompt set. To train a semantic masker that yields similar outputs for images generated from the same prompt (and their semantics-preserving transforms), we then generate multiple images per prompt with \emph{Stable Diffusion v2.1}\footnote{\url{https://huggingface.co/stabilityai/stable-diffusion-2-1-base}} and apply semantics-preserving augmentations. Finally, we summarize the resulting dataset composition and discuss how the COCO/SDP imbalance influences downstream results.

\paragraph{Prompt Mining and Clustering.}
We collect caption prompts from the training split of MS-COCO-2017~\cite{lin2014microsoft} and from the \texttt{train.parquet} file of the Stable Diffusion Prompts (SDP) dataset\footnote{\url{https://huggingface.co/datasets/Gustavosta/Stable-Diffusion-Prompts}}, normalize whitespace, and deduplicate at the string level. Each prompt is embedded with a sentence-transformer \textit{BGE-large-en-v1.5\footnote{\url{https://huggingface.co/BAAI/bge-large-en-v1.5}}
} and \(\ell_2\)-normalized so that Euclidean distance rankings coincide with cosine similarity, making \(k\)-means effectively operate in the spherical setting. To scale to millions of prompts, we run Mini-Batch \(k\)-means~\cite{sculley2010web} with \(K=64{,}000\) clusters. Within each cluster we rank candidates by proximity to the centroid and select \(M_i\!\in\!\{1,2\}\) prompts in proportion to the cluster’s size, enforcing intra-cluster diversity by rejecting any candidate whose cosine similarity to already selected ones exceeds \(0.8\). We do this to preferentially select prompts that are semantically typical. This produces \(\approx 96{,}041\) representative prompts that cover the prompt distribution while avoiding near duplicates.

\paragraph{Data Generation and Augmentation.}
For each selected prompt, we synthesize \textbf{16} independent images with \emph{Stable Diffusion v2.1} (guidance scale \(=7.5\), \(50\) denoising steps with DPMSolver~\cite{lu2022dpm}, output \(512{\times}512\), latent \(4{\times}64{\times}64\)). We then create \textbf{16} additional, semantics–preserving views by applying one augmentation per original, sampled from a fixed operator pool and executed in a stable order (composition \(\rightarrow\) geometry \(\rightarrow\) resolution/compression \(\rightarrow\) color/blur/noise \(\rightarrow\) flip), which avoids black borders and preserves content semantics. The pool includes:

\begin{itemize}
  \item \textbf{RandomResizedCrop} (scale \(0.4\!-\!1.0\), aspect \(3{:}4\!-\!4{:}3\), antialias).
  \item \textbf{RandomPerspective} (distortion\_scale \(=0.25\)).
  \item \textbf{RandomAffine} (degrees \(\pm 15^\circ\), translate \(\le 10\%\), scale \(0.9\!-\!1.1\), shear \(\pm(5^\circ,3^\circ)\), bicubic).
  \item \textbf{ResizeDownUp} (downscale ratio \(0.5\!-\!0.85\), then upscale back).
  \item \textbf{JPEG Compression} (quality \(55\!-\!85\)).
  \item \textbf{ColorJitter} (brightness/contrast/saturation \(0.3\)).
  \item \textbf{RandomGrayscale}.
  \item \textbf{Gaussian Blur} (Gaussian radius \(0.8\!-\!1.5\)).
  \item \textbf{Additive Gaussian noise} (std \(0.005\!-\!0.03\)).
  \item \textbf{Light salt\&pepper noise} (prob \(0.002\!-\!0.01\)).
  \item \textbf{Deterministic horizontal flip}.
\end{itemize}

For each generated image, we randomly select one or two operators from the above pool and apply them to that image to obtain one augmented view. This yields \(32\) views per prompt (16 generations + 16 augmented).

\paragraph{Final Dataset Composition.}
From the curated pool of $96{,}041$ prompts (selected from COCO and SDP), we generate $32$ views per prompt (16 independent generations and 16 semantics-preserving augmentations).
We then extract DINOv2-giant~\cite{oquab2023dinov2} CLS embeddings (dimension $1536$) for all images, yielding \textbf{SemCon-3M}: a corpus of $96{,}041 \times 32 = 3{,}073{,}312$ embeddings.
By source, COCO contributes $84{,}675$ prompts ($88.2\%$; $2{,}709{,}600$ embeddings) and SDP contributes $11{,}366$ prompts ($11.8\%$; $363{,}712$ embeddings), as summarized in Table~\ref{tab:semcon3m_composition}.
Because COCO dominates the training distribution, we consistently observe stronger performance on COCO than on SDP, which we attribute to this dataset imbalance rather than a method-specific bias.

\begin{table}[t]
  \centering
  \caption{SemCon-3M composition by source. Each prompt contributes 32 views (images/embeddings).}
  \label{tab:semcon3m_composition}
  \begin{tabular}{lrr}
    \toprule
    Source &  \# Embeddings & \% Embeddings \\
    \midrule
    COCO &  \textbf{2{,}709{,}600} & \textbf{88.2}\% \\
    SDP 
         &  \textbf{363{,}712} & \textbf{11.8}\% \\
    \midrule
    Total & \textbf{3{,}073{,}312} & 100\% \\
    \bottomrule
  \end{tabular}
\end{table}

Because COCO comprises the vast majority of \textbf{SemCon-3M}, the semantic masker is trained more fully on COCO-style prompts and imagery. As a result, we consistently observe stronger anti-forgery resistance (lower false-accept rates under imprinting/reprompting) and higher robustness to common perturbations on COCO than on SDP. We attribute this gap to data imbalance rather than a method-specific bias. The experimental result is analyzed in ~\cref{sec:addition_exp}. 

We will publicly release \textbf{SemCon-3M} to the research community in the future.

\section{Model Architecture \& Training Details}
\label{sec:model_arch_train_details}

\paragraph{Backbone \& Embeddings.} We adopt a frozen DINOv2-giant~\cite{oquab2023dinov2} vision encoder to obtain a single global CLS embedding $e\!\in\!\mathbb{R}^{1536}$ per image. All images are preprocessed using the standard DINOv2 pipeline (resize/crop and normalization to the encoder’s default statistics). Unless otherwise stated, CLS embeddings are $\ell_2$–normalized before being fed to our hashing network.

\paragraph{Network Architecture Details.} The semantic masker $f_\theta$ is a lightweight MLP-style hashing network with three modules: \begin{itemize} \item \textbf{Encoder} $\mathrm{Enc}:\mathbb{R}^{1536}\!\to\!\mathbb{R}^{H}$: a stack of residual fully-connected blocks (depth $2$ by default) with BatchNorm and GELU; we set hidden width $H\!=\!2048$. \item \textbf{Projection head} $\mathrm{Proj}:\mathbb{R}^{H}\!\to\!\mathbb{R}^{D}$: a small MLP that maps to $D$-dimensional features, followed by $\ell_2$ normalization; we use $D=8192$ in experiments. \item \textbf{Hash head} $\mathrm{Hash}:\mathbb{R}^{D}\!\to\!\mathbb{R}^{B}$: residual fully-connected blocks ending in a linear layer that outputs $B$ hash logits; we use $B\!=\!1024$. \end{itemize} At inference, given CLS embedding $e$, the network first produces a \emph{hash logit} $\ell=\mathrm{Hash}(\mathrm{Proj}(\mathrm{Enc}(e)))\in\mathbb{R}^B$, then a \emph{soft binary code} $b=\tanh(s\,\ell)\in[-1,1]^B$, and finally a \emph{binary code} $m=\bigl(\mathrm{sign}(b)+1\bigr)/2\in\{0,1\}^B$, where $s\!>\!0$ controls the soft sign’s sharpness.

\paragraph{Training Schedule: losses, temperatures, and hardness.} We train $f_\theta$ in two stages, operating in \emph{logit mode} (i.e., losses consume logits or their $\tanh$ transforms).

\textbf{Stage-1 (feature contrast).} We optimize $\mathrm{Enc}$+$\mathrm{Proj}$ with supervised contrastive loss on normalized features $z_i\!\in\!\mathbb{R}^D$: \[ \mathcal{L}_{\mathrm{sup}} = -\frac{1}{N}\sum_{i=1}^N \frac{1}{|P(i)|} \sum_{p\in P(i)} \log \frac{\exp(z_i^\top z_p/\tau)}{\sum_{a\neq i}\exp(z_i^\top z_a/\tau)}. \] We use Adam with $\text{lr}\!=\!1\!\times\!10^{-4}$, gradient clipping $1.0$, temperature $\tau\!=\!0.07$ for the first half of training and $0.10$ for the second half. Default epochs: $\texttt{epoch1}=180$. \textbf{Stage-2 (hash contrast + regularizers).} We freeze the backbone and Stage-1 modules, and train $\mathrm{Hash}$ to produce near-binary, balanced, decorrelated codes. Let $\ell_i$ be logits and $b_i=\tanh(s\,\ell_i)$ the soft binary codes. We apply a supervised contrastive loss in code space with temperature $\tau_h$, \[ \mathcal{L}_{\mathrm{hash}} = -\frac{1}{N}\sum_{i=1}^N \frac{1}{|P(i)|} \sum_{p\in P(i)} \log \frac{\exp\bigl((b_i^\top b_p/B)/\tau_h\bigr)}{\sum_{a\neq i}\exp\bigl((b_i^\top b_a/B)/\tau_h\bigr)}, \] and add three standard regularizers on $\{b_i\}$: quantization $\mathcal{L}_{\mathrm{q}}=\mathbb{E}[\,1-|b_i|\,]$, bit-balance $\mathcal{L}_{\mathrm{bal}}=\frac{1}{B}\sum_{k=1}^B \bigl(\frac{1}{N}\sum_i b_{i,k}\bigr)^2$, decorrelation $\mathcal{L}_{\mathrm{dcr}}=\lVert C-I\rVert_F^2$ (where $C$ is the batch covariance). To stabilize codes, we jitter features twice and penalize their discrepancy $\mathcal{L}_{\mathrm{cons}}=\mathbb{E}[\lVert b_i^{(1)}-b_i^{(2)}\rVert_1]$. The stage objective is \[ \mathcal{L} = \mathcal{L}_{\mathrm{hash}} + \lambda_{\mathrm{q}}\mathcal{L}_{\mathrm{q}} + \lambda_{\mathrm{bal}}\mathcal{L}_{\mathrm{bal}} + \lambda_{\mathrm{dcr}}\mathcal{L}_{\mathrm{dcr}} + \lambda_{\mathrm{cons}}\mathcal{L}_{\mathrm{cons}}. \] We use Adam with $\text{lr}\!=\!1\!\times\!10^{-4}$, jitter $\sigma\!=\!0.01$ (no feature dropout), and a \emph{hardness} schedule for $s$: $s\_{\text{init}}\!=\!1.0$ and at epochs $\{20,40,70\}$ we multiply $s$ by $\gamma\!=\!2.5$. Default epochs: $\texttt{epoch2}=160$. (At evaluation, we use a fixed inference sharpness, e.g., $s\!=\!12$ to binarize.)

\emph{Note on ablations.} The effect of \texttt{epoch1}/\texttt{epoch2} is analyzed in Sec.~\ref{sec:addition_exp}; here we keep the schedule fixed and report the default settings above.

Tables~\ref{tab:train_hparams} and \ref{tab:masker_arch} summarize the model architecture and training hyperparameters; we will publicly release our code to the research community.

\begin{table}[t]
  \centering
  \caption{Training objectives and hyperparameters.}
  \label{tab:train_hparams}
  \footnotesize
  \setlength{\tabcolsep}{4.5pt}
  \renewcommand{\arraystretch}{1.15}
  \begin{tabular}{ll}
    \toprule
    \textbf{Setting} & \textbf{Value} \\
    \midrule
    Widths $(H,D,B)$ & $2048,\;8192,\;1024$ \\
    Norm/Act & BN (eps $10^{-5}$, mom $0.1$), GELU \\
    \midrule
    Stage-1 loss & SupCon; $\tau: 0.07 \!\to\! 0.10$ (mid-training) \\
    Stage-1 opt & Adam, lr $1{\times}10^{-4}$, grad-clip $1.0$ \\
    Stage-1 epochs & $180$ \\
    \midrule
    Stage-2 loss & Hash SupCon ($\tau_h{=}0.10$) \\
    Regularizers & $\lambda_q{=}0.10,\;\lambda_{\text{bal}}{=}0.10,$ \\
                 & $\lambda_{\text{dcr}}{=}0.10,\;\lambda_{\text{cons}}{=}0.20$ \\
    Hardness $s$ & $1.0 \xrightarrow{\times\,2.5}$ at epochs $20/40/70$ \\
    Stage-2 opt & Adam, lr $1{\times}10^{-4}$; jitter $\sigma{=}0.01$ \\
    Stage-2 epochs & $160$ \\
    \midrule
    PK sampling & $P{=}64,\;K{=}16$ (batch $=PK$) \\
    Inference $s$ & $12$ (for $\tanh$ at test) \\
    \bottomrule
  \end{tabular}
\end{table}

\begin{table}[t]
  \centering
  \caption{Semantic masker architecture. DINOv2-giant is frozen and we use its CLS embedding as input.}
  \label{tab:masker_arch}
  \small
  \setlength{\tabcolsep}{4pt}
  \renewcommand{\arraystretch}{1.12}
  \begin{tabularx}{\columnwidth}{@{} l X c @{}}
    \toprule
    \textbf{Module} & \textbf{Layer (settings)} & \textbf{Output} \\
    \midrule
    Backbone & DINOv2-giant (ViT-G/14), CLS only;  & $\mathbb{R}^{1536}$ \\
    \midrule
    Enc  & Linear $1536\!\to\!2048$ + BN + GELU & $\mathbb{R}^{2048}$ \\
         & Residual FC: [Linear $2048\!\to\!2048$ + BN + GELU] $\times 1$ & $\mathbb{R}^{2048}$ \\
    \midrule
    Proj & Linear $2048\!\to\!8192$ + BN + GELU & $\mathbb{R}^{8192}$ \\
         & $\ell_2$-normalize feature $z$ & $\mathbb{R}^{8192}$ \\
    \midrule
    Hash & Residual FC: [Linear $D\!\to\!D$ + BN + GELU] $\times 2$ & $\mathbb{R}^{D}$ \\
         & Linear $D\!\to\!B$ (hash logits $\ell$) & $\mathbb{R}^{B}$ \\
    \bottomrule
  \end{tabularx}
\end{table}

\paragraph{Batch Construction via PK Sampling.}
Our dataset contains $\sim\!10^5$ prompt classes; uniform random batching would yield almost-all-negative batches, weakening the supervised contrast.
We therefore adopt a PK-sampling strategy per mini-batch: sample $P$ distinct prompts and $K$ views per prompt to form a batch of size $P\!\times\!K$ with many in-class positives.
Concretely, we maintain per-class index lists, reshuffle them each epoch, and draw contiguous $K$-sized chunks per selected label; if any class runs out of $K$ samples, we finish the epoch to avoid label imbalance in incomplete batches.
Unless specified, we use $P\!=\!64$, $K\!=\!16$ (batch size $1024$), which balanced GPU throughput and contrastive signal quality in our setup.

\begin{table*}[t]
  \centering
  \caption{Semantic masker invariance/separability on COCO vs.\ SDP (SD-2.1, Guidance Scale 7.5, \(B{=}1024\); averages over \(1{,}000\) prompts). }
  \label{tab:sem_masker_coco_sdp}
  \setlength{\tabcolsep}{5pt}
  \renewcommand{\arraystretch}{1.1}
  \resizebox{\linewidth}{!}{
  \begin{tabular}{lccccccccc}
    \toprule
    Dataset & Intra-Orig\(\downarrow\) & Ref-vs-Dist\(\downarrow\) & All-20\(\downarrow\) & Bit Ent. & Max dist\(\uparrow\) & Min dist\(\uparrow\) & Mean dist\(\uparrow\) & Entropy (diff-prompt)\(\uparrow\) \\
    \midrule
    COCO & 73.47 & 22.99 & 61.50 & 0.131 & 573.02 & 319.07 & 508.77 & 0.995 \\
    SDP & 102.77 & 35.92 & 86.59 & 0.203 & 552.84 & 237.82 & 428.68 & 0.810 \\
    \bottomrule
  \end{tabular}}
\end{table*}

\begin{table*}[t]
  \centering
  \caption{Epoch ablation. Results averaged over 500 prompts. Lower is better for the first three columns.}
  \label{tab:epoch_ablation}
  \footnotesize
  \renewcommand{\arraystretch}{1.08}
  \setlength{\tabcolsep}{6pt}
  \begin{tabular}{@{} c c c c c c @{}}
    \toprule
    Ep1 & Ep2 & Intra-Orig$\downarrow$ & Ref-vs-Dist$\downarrow$ & All-20$\downarrow$ & Bit Ent. \\
    \midrule
    160 & 160 & 71.98 & 23.23 & 60.43 & 0.129 \\
    160 & 180 & 71.03 & 24.79 & 61.05 & 0.130 \\
    160 & 200 & 71.66 & 22.12 & 60.00 & 0.128 \\
    170 & 160 & 73.54 & 22.60 & 61.53 & 0.131 \\
    170 & 190 & 72.22 & 22.70 & 59.67 & 0.128 \\
    \textbf{180} & \textbf{160} & \textbf{66.92} & \textbf{21.83} & \textbf{56.61} & \textbf{0.121} \\
    180 & 180 & 76.64 & 24.34 & 63.92 & 0.136 \\
    180 & 200 & 69.72 & 25.20 & 60.92 & 0.130 \\
    200 & 250 & 69.13 & 22.95 & 59.07 & 0.126 \\
    \bottomrule
  \end{tabular}
\end{table*}

\section{Experiments on Semantic Masker and Ablation Study}
\label{sec:addition_exp}

\paragraph{Semantic masker invariance and separability.}
We assess (i) within-prompt tightness and (ii) across-prompt separability.  
For each of \(1{,}000\) random prompts (COCO or SDP), we sample \(N_o{=}10\) originals with SD-2.1 (512\(\times\)512, 50 steps, guidance 7.5), create \(N_d{=}10\) semantics-preserving distortions of the first original image, and compute \(B{=}1024\)-bit codes (inference \(s{=}12\)). Let \(\mathrm{Ham}(\mathbf{u},\mathbf{v})=\lVert \mathbf{u}\oplus\mathbf{v}\rVert_0\). With originals \(\{\mathbf{m}^{\mathrm{orig}}_i\}_{i=1}^{N_o}\subset\{0,1\}^B\) and distortions \(\{\mathbf{m}^{\mathrm{dist}}_j\}_{j=1}^{N_d}\), we report:
\begin{equation}
\label{eq:intra-ref}
\begin{aligned}
\mathrm{Intra\text{-}Orig}
&= \frac{2}{N_o(N_o\!-\!1)}
   \sum_{1\le i<j\le N_o}
   \mathrm{Ham}\!\bigl(\mathbf{m}^{\mathrm{orig}}_i,\mathbf{m}^{\mathrm{orig}}_j\bigr),\\[2pt]
\mathrm{Ref\text{-}vs\text{-}Dist}
&= \frac{1}{N_d}\sum_{j=1}^{N_d}
   \mathrm{Ham}\!\bigl(\mathbf{m}^{\mathrm{orig}}_1,\mathbf{m}^{\mathrm{dist}}_j\bigr).
\end{aligned}
\end{equation}
Let \(\mathcal{M}=\{\mathbf{m}^{\mathrm{orig}}_i\}_{i=1}^{N_o}\cup\{\mathbf{m}^{\mathrm{dist}}_j\}_{j=1}^{N_d}\) and enumerate \(\mathcal{M}=\{\mathbf{m}_a\}_{a=1}^{N}\) with \(N{=}N_o\!+\!N_d{=}20\). Then
\begin{equation}
\label{eq:all20}
\mathrm{All\text{-}20}
= \frac{2}{N(N\!-\!1)}
  \sum_{1\le a<b\le N}
  \mathrm{Ham}\!\bigl(\mathbf{m}_a,\mathbf{m}_b\bigr).
\end{equation}
Bit entropy uses per-bit frequency \(p_k=\tfrac{1}{N}\sum_{a=1}^{N} m_{a,k}\):
\begin{equation}
\label{eq:entropy}
\mathrm{Entropy}
= \frac{1}{B}\sum_{k=1}^{B}
  \Bigl[-p_k\log_2 p_k - (1-p_k)\log_2(1-p_k)\Bigr].
\end{equation}
For cross-prompt statistics, draw codes \(\{\mathbf{c}_t\}_{t=1}^{M}\) from different prompts and set
\(D_{it}=\mathrm{Ham}(\mathbf{m}^{\mathrm{orig}}_i,\mathbf{c}_t)\); we report
\(\min_{i,t} D_{it}\), \(\frac{1}{N_o M}\sum_{i,t} D_{it}\), and \(\max_{i,t} D_{it}\).
Table~\ref{tab:sem_masker_coco_sdp} shows strong within-prompt invariance and near-random across-prompt behavior on COCO (mean cross-prompt \(\approx 509\)), while SDP is weaker (mean \(\approx 409\)), matching our COCO-skewed training mix.

\paragraph{Ablation on two-stage epochs.}
We vary the number of epochs in Stage~1 (\texttt{epoch1}) and Stage~2 (\texttt{epoch2}) while keeping all other settings fixed. Metrics are averaged over 500 prompts; lower is better for \emph{Intra-Orig}, \emph{Ref-vs-Dist}, and \emph{All-20}. The result is shown in ~\cref{tab:epoch_ablation}.

The best setting is \(\mathbf{epoch1{=}180,\,epoch2{=}160}\), which achieves the lowest values across all three distance metrics and the lowest bit entropy. Increasing Stage~2 beyond \(\sim\!160\) epochs tends to degrade both same-prompt consistency and distortion robustness (e.g., \(180/180\), \(180/200\)). 


\section{Adaptive Attack}
\label{app:adaptive_attack}
\subsection{Training a surrogate semantic masker.}
SemBind assumes the semantic masker is kept private by the service provider.
Nevertheless, an adaptive adversary may attempt to approximate this component by collecting multiple watermarked images generated under the same prompt and training a surrogate masker.

\textbf{Threat model.}
We consider a deliberately \emph{stronger-than-realistic} adaptive attacker.
The adversary is assumed to possess completely \textbf{the same} training dataset(Appendix~\ref{sec:dataset_semcon3m}), network architecture, loss functions, and training hyperparameters (Sec.~\ref{sec:semantic-masker} and Appendix~\ref{sec:model_arch_train_details}) as the defender.
The only difference is the randomness in training (e.g., random seed used in training).
The adversary's goal is to train a surrogate semantic masker whose outputs match the provider's codes for images generated from the same prompt.

\textbf{Experimental setup.}
Following the same training recipe as in the main paper, we independently train five semantic maskers (M1--M5) using different random seeds.
Each masker outputs a $B{=}1024$-bit semantic code.
We evaluate code consistency on both COCO and SDP by randomly sampling $1000$ prompts, generating one image per prompt, and extracting semantic codes using each trained masker.
For each pair of maskers, we compute the Hamming distance between their codes over the same set of generated images.

\textbf{Results and discussion.} Tables~\ref{tab:adaptive_coco} and~\ref{tab:adaptive_sdp} report the resulting pairwise distance matrices. We can see that, even under this stringent setting, independently trained maskers exhibit near-random bitwise mismatch: pairwise mean Hamming distances concentrate around $\approx 512$ bits (i.e., $\approx 50\%$ of $1024$ bits) on both COCO and SDP.
This suggests that training a surrogate masker does not reliably reproduce the provider's bit-level code mapping, making it difficult to directly clone the semantics-to-mask binding used by SemBind.
A key reason is that our objective is primarily relational---it enforces that same-prompt samples are close and different-prompt samples are separated---and is non-identifiable up to symmetry transformations (e.g., bit permutations/flips or rotations of intermediate features), under which the loss remains (nearly) unchanged.
Consequently, different random seeds can converge to different but equally valid solutions that preserve relative structure while inducing different bitwise representations.

\subsection{Pixel-level spoofing of the semantic masker}
\label{subsec:pixel_spoofing_masker}

Beyond training a surrogate masker, an adaptive adversary may attempt to \emph{spoof} the semantic masker directly in the pixel domain.
For example, given a benign \emph{source} image (e.g., a cat photo), the adversary may take an arbitrary \emph{target} image (e.g., containing policy-violating or otherwise illicit content) and overlay the source image onto the target at some scale.
This raises a natural question: does the resulting semantic code become entirely different, partially overlap with the source code, or approach the source code as the overlay grows?
Equivalently, how much pixel-level overlap is required to cause the semantic masking process to be misled?

\textbf{Threat model.}
We consider a spoofing attacker who modifies images in the pixel space.
In this setting, the adversary is given a benign reference image and attempts to overlay it onto a chosen target image with a controllable size ratio, with the goal that the composite image yields a semantic code close to that of the benign image.

\textbf{Experimental setup.}
We evaluate this spoofing attack on \emph{both} COCO and SDP prompt sets, sampling $100$ prompts from each dataset (one generated image per prompt).
For the benign source image, we use a fixed ``cat'' image generated by Stable Diffusion.
For target images, we randomly sample prompts from the COCO and SDP prompt-evaluation sets and generate one image per prompt, then extract their semantic codes using the semantic masker.
We then perform pixel-space overlay by resizing the cat image to a scale ratio $r\in\{0.0,0.1,\ldots,1.0\}$ relative to the target image resolution and pasting it onto the target under two placements: (i) a bottom-right overlay, and (ii) a centered overlay.
For each $r$, we compute the semantic code of the resulting composite image and measure its Hamming distance to the reference code at $r=1.0$ (i.e., when the image is fully replaced by the cat image), reporting distances for $r=0.0,\ldots,0.9$. 

\textbf{Results and discussion.}
Results are summarized in Figs.~\ref{fig:spoofing_curve_center} and~\ref{fig:spoofing_curve_bottomright}, with an illustrative example shown in Fig.~\ref{fig:spoofing_grid}.
Overall, increasing the cat overlay ratio $r$ monotonically reduces the Hamming distance to the source code (at $r=1.0$), confirming that large pixel-level overlap can pull the semantic code toward the benign source.
However, the effect is limited for moderate overlay sizes: across both COCO and SDP, the semantic code remains far from the source for $r\le 0.8$, and a sharp transition only occurs when the overlay becomes extremely large (around $r=0.9$, close to fully replacing the target image).
This indicates that SemBind largely preserves its anti-forgery effectiveness against pixel-space spoofing unless the adversary is willing to overwrite most of the target content with benign pixels.
In practical forgery scenarios, such a large benign overlay would substantially compromise the attacker’s intended (illicit) semantics and thus significantly reduce the utility of the forged image.
We therefore conclude that pixel-level spoofing is an ineffective strategy to neutralize SemBind in practice, as successfully misleading the semantic masker requires overwhelming the target image with benign content.

\begin{table}[t]
\centering
\begin{minipage}{0.49\columnwidth}
  \centering
  \caption{COCO: mean Hamming distance (1024-bit) between independently trained maskers (1000 prompts).}
  \label{tab:adaptive_coco}
  \footnotesize                
  \setlength{\tabcolsep}{3pt} 
  \renewcommand{\arraystretch}{1.05}
  \begin{tabular}{lccccc}
    \toprule
          & M1 & M2 & M3 & M4 & M5 \\
    \midrule
    M1 & 0.0   & 508.0 & 492.8 & 505.2 & 511.6 \\
    M2 & 508.0 & 0.0   & 514.3 & 506.7 & 520.1 \\
    M3 & 492.8 & 514.3 & 0.0   & 509.5 & 503.8 \\
    M4 & 505.2 & 506.7 & 509.5 & 0.0   & 515.4 \\
    M5 & 511.6 & 520.1 & 503.8 & 515.4 & 0.0   \\
    \bottomrule
  \end{tabular}
\end{minipage}\hfill
\begin{minipage}{0.49\columnwidth}
  \centering
  \caption{SDP: mean Hamming distance (1024-bit) between independently trained maskers (1000 prompts).}
  \label{tab:adaptive_sdp}
  \footnotesize
  \setlength{\tabcolsep}{3pt}
  \renewcommand{\arraystretch}{1.05}
  \begin{tabular}{lccccc}
    \toprule
          & M1 & M2 & M3 & M4 & M5 \\
    \midrule
    M1 & 0.0   & 510.2 & 498.7 & 507.9 & 513.4 \\
    M2 & 510.2 & 0.0   & 516.8 & 509.1 & 518.6 \\
    M3 & 498.7 & 516.8 & 0.0   & 511.0 & 505.6 \\
    M4 & 507.9 & 509.1 & 511.0 & 0.0   & 516.2 \\
    M5 & 513.4 & 518.6 & 505.6 & 516.2 & 0.0   \\
    \bottomrule
  \end{tabular}
\end{minipage}
\end{table}

\begin{figure*}[t]
\centering
\setlength{\tabcolsep}{0pt} 
\renewcommand{\arraystretch}{1.0}
\newcommand{\imgcell}[2]{%
  \begin{minipage}[t]{0.085\textwidth}\centering
    \includegraphics[width=\linewidth]{#1}\\[-0.55em]
    {\scriptsize #2}
  \end{minipage}%
}

\begin{tabular}{@{}*{11}{c}@{}}
\imgcell{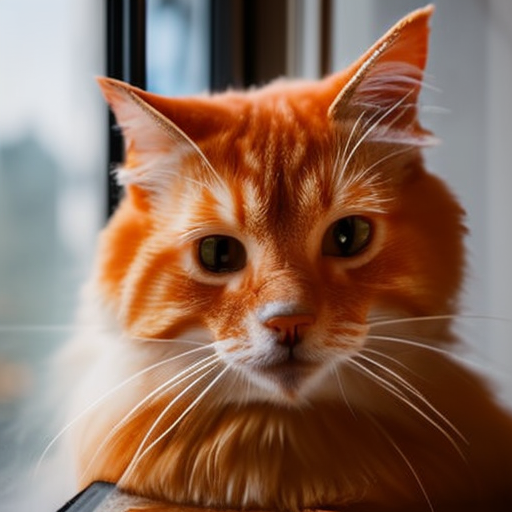}{source} &
\imgcell{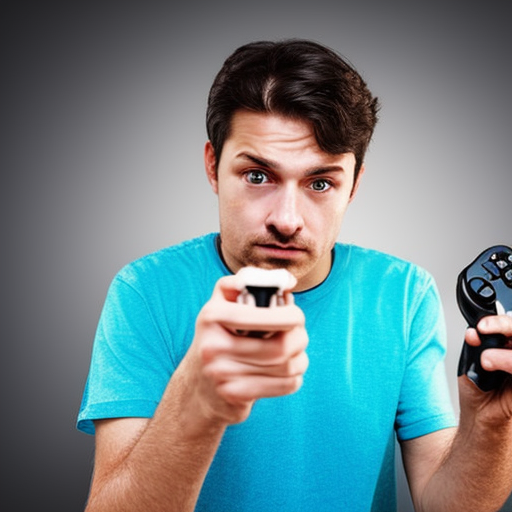}{target} &
\imgcell{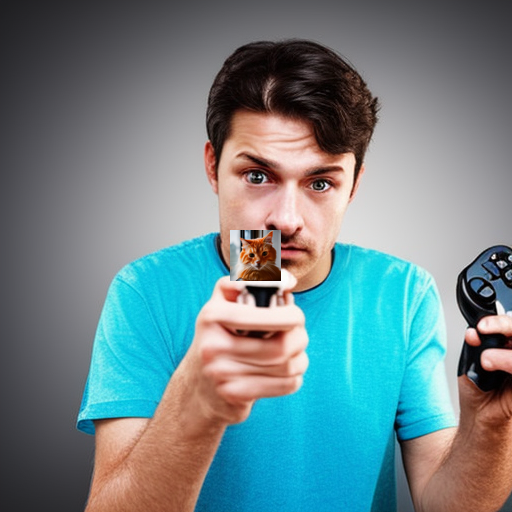}{$d{=}532$} &
\imgcell{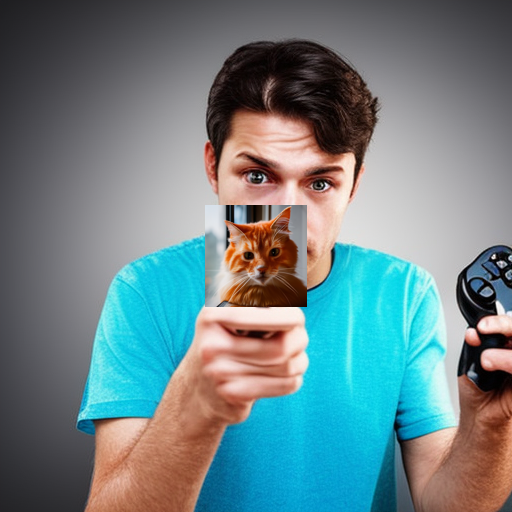}{$d{=}284$} &
\imgcell{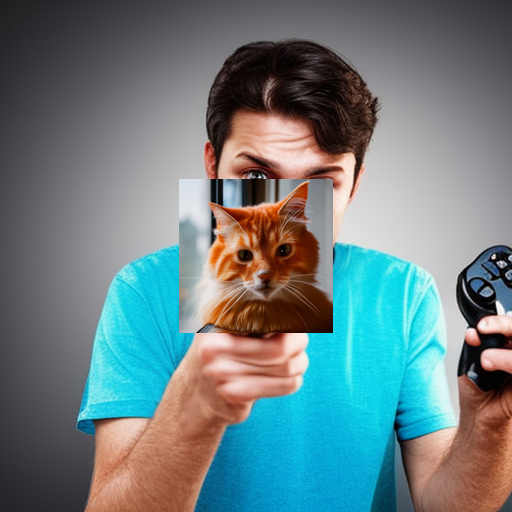}{$d{=}287$} &
\imgcell{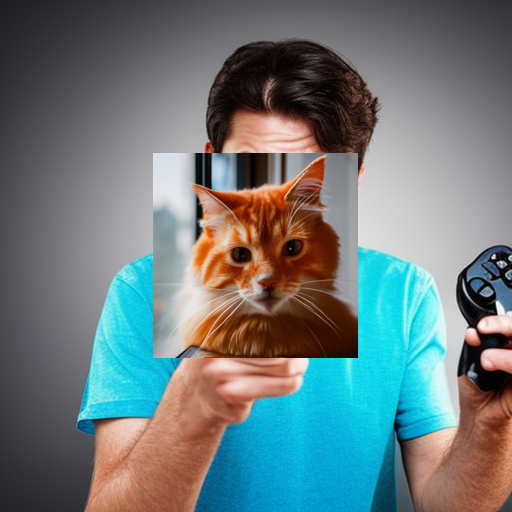}{$d{=}293$} &
\imgcell{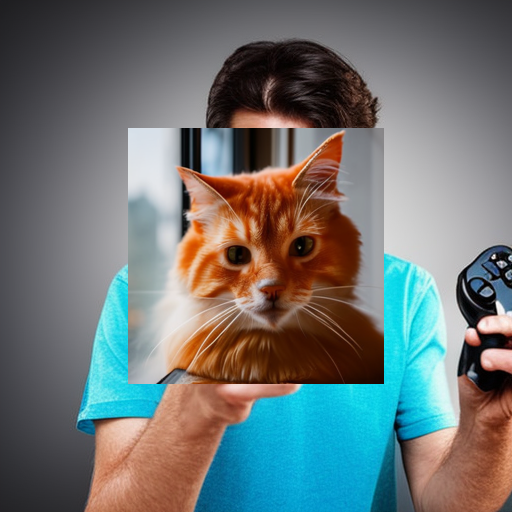}{$d{=}278$} &
\imgcell{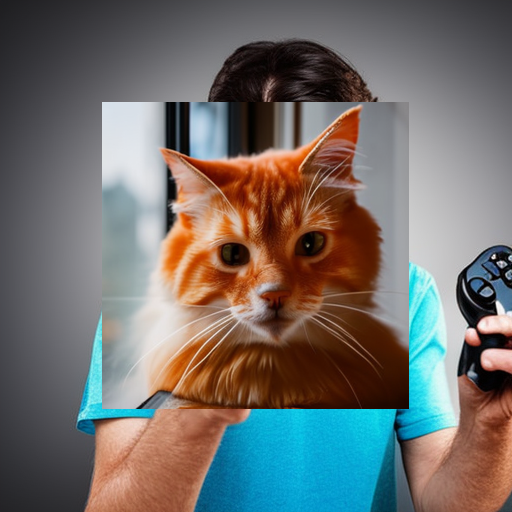}{$d{=}287$} &
\imgcell{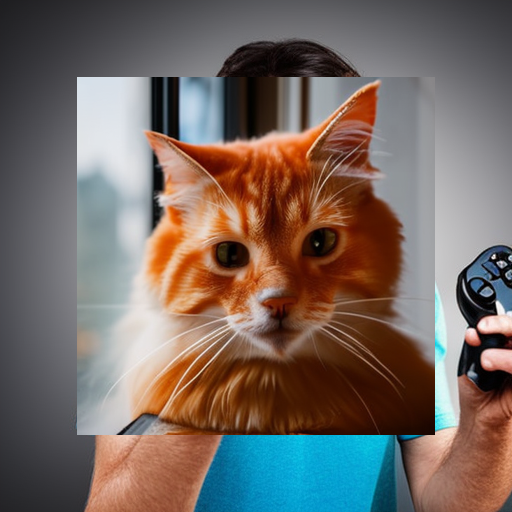}{$d{=}283$} &
\imgcell{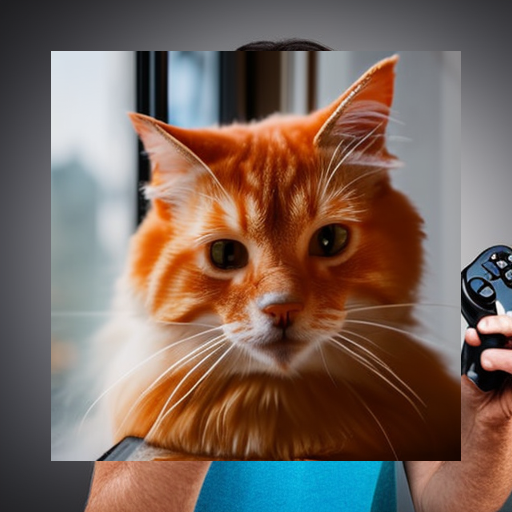}{$d{=}272$} &
\imgcell{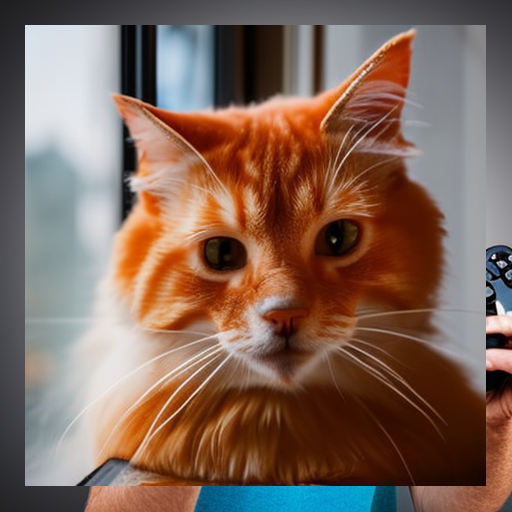}{$d{=}6$} \\
\end{tabular}

\vspace{2pt}

\begin{tabular}{@{}*{11}{c}@{}}
\imgcell{figs/spoofing_attack/cat_source.png}{source} &
\imgcell{figs/spoofing_attack/base_scale0.png}{target} &
\imgcell{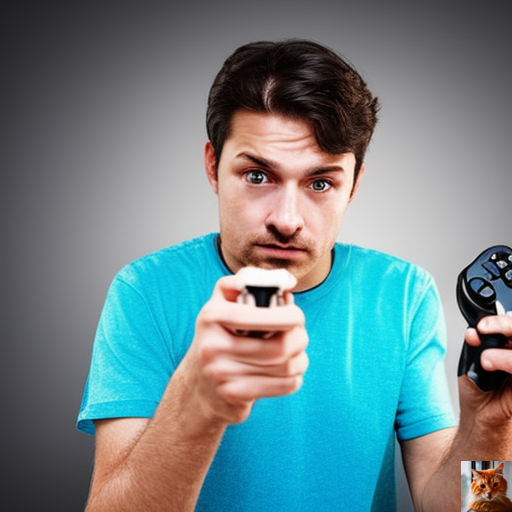}{$d{=}535$} &
\imgcell{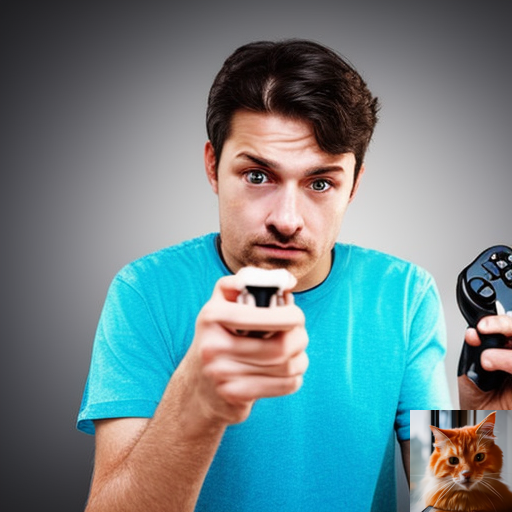}{$d{=}513$} &
\imgcell{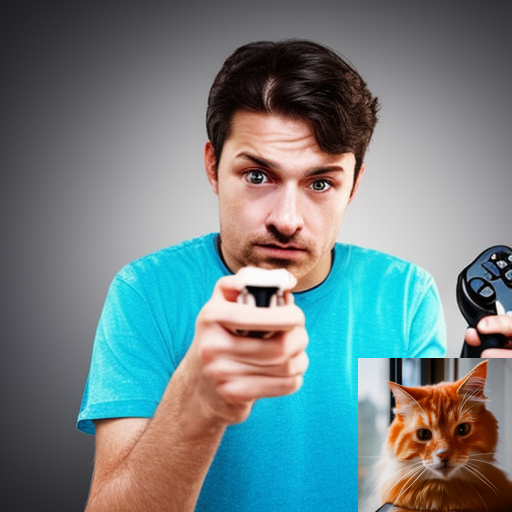}{$d{=}317$} &
\imgcell{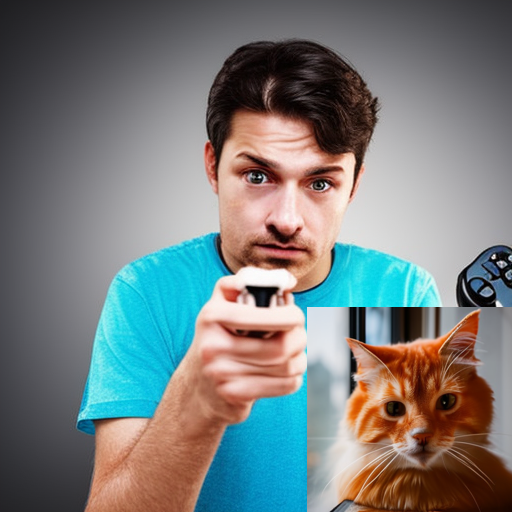}{$d{=}325$} &
\imgcell{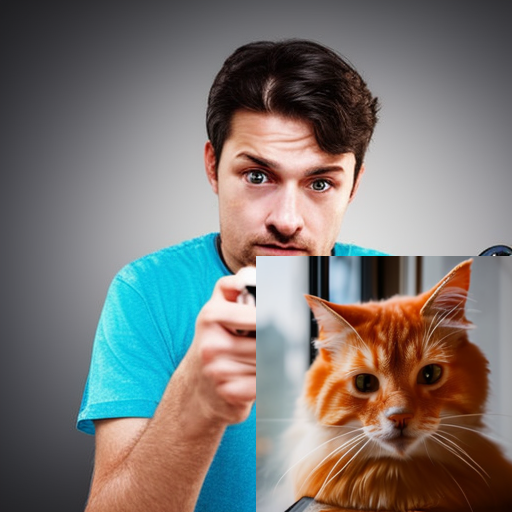}{$d{=}269$} &
\imgcell{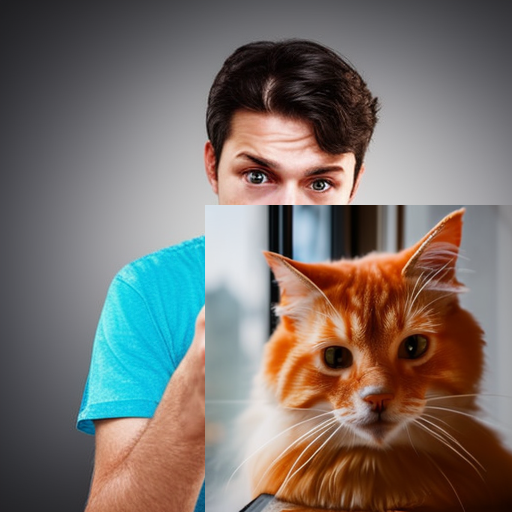}{$d{=}225$} &
\imgcell{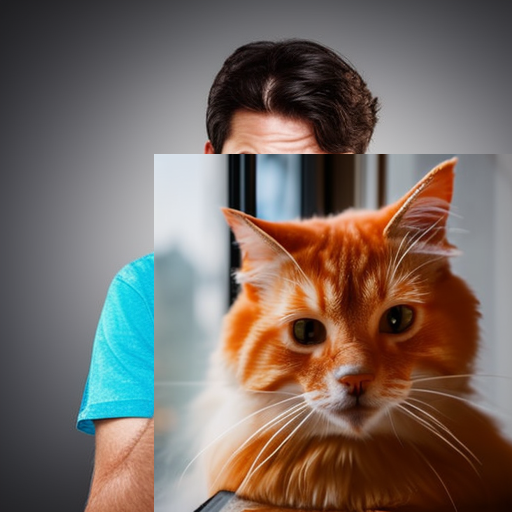}{$d{=}218$} &
\imgcell{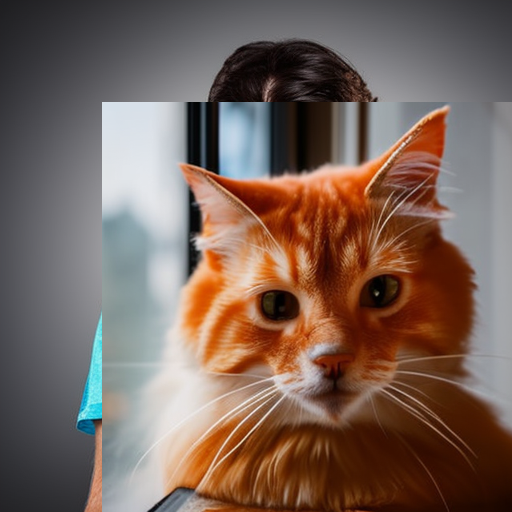}{$d{=}13$} &
\imgcell{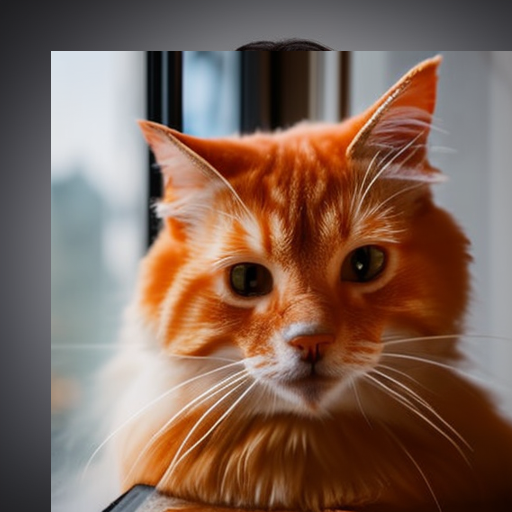}{$d{=}3$} \\
\end{tabular}

\caption{Pixel-level spoofing by overlaying a benign source image (cat) onto a target image at scale ratio $r\in\{0.1,0.2,\ldots,0.9\}$ (with $r=0.0$ the original target and $r=1.0$ the full source image).
Top: centered overlay. Bottom: bottom-right overlay. Labels report the Hamming distance $d$ to the reference semantic code of the source image.}

\label{fig:spoofing_grid}
\end{figure*}

\begin{figure}[t]
  \centering
  \begin{subfigure}[t]{0.49\columnwidth}
    \centering
    \includegraphics[width=\linewidth]{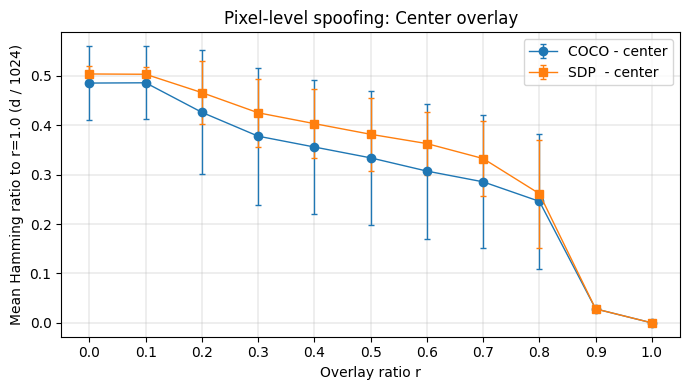}
    \caption{Centered overlay.}
    \label{fig:spoofing_curve_center}
  \end{subfigure}\hfill
  \begin{subfigure}[t]{0.49\columnwidth}
    \centering
    \includegraphics[width=\linewidth]{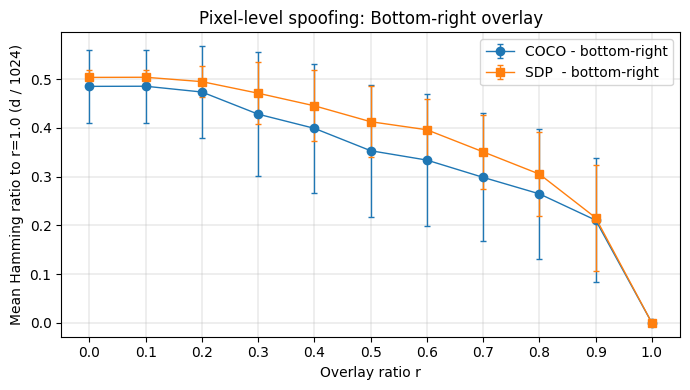}
    \caption{Bottom-right overlay.}
    \label{fig:spoofing_curve_bottomright}
  \end{subfigure}
  \caption{Pixel-level spoofing curves on COCO and SDP (100 prompts each). We report the mean Hamming distance (normalized by $1024$) between the semantic code of the composite image at overlay ratio $r$ and the reference code at $r=1.0$ (full cat image). Error bars indicate $\pm 1$ standard deviation.}
  \label{fig:spoofing_curves}
\end{figure}

\section{Generalization Experiments}
\label{app:generalization}

To further demonstrate the generalization of SemBind, we conduct two additional experiments.
First, we evaluate SemBind on the same image generation model under a more realistic and broader prompt dataset in ~\ref{app:gen_large_prompts}, assessing robustness and defense against forgery attacks beyond our primary evaluation sets.
Second, we study the feasibility of applying SemBind to \textsc{FLUX}~\cite{labs2025flux} in ~\ref{app:FLUX}, testing whether the proposed semantic binding mechanism transfers to a different diffusion-based generator.
We next describe the experimental setups and then report the results with analysis for each experiment.

\subsection{Generalization on a Larger Prompt Dataset}
\label{app:gen_large_prompts}

\paragraph{Experimental setup.}
To evaluate generalization, we reuse the same SemBind-enabled variants and hyperparameter settings as in the main paper (Sec.~\ref{sec:Experimental Setup}), and repeat the forgery-resistance and robustness evaluations on a larger and more realistic prompt corpus,
\texttt{FredZhang7/stable-diffusion-prompts-2.47M}\footnote{\url{https://huggingface.co/datasets/FredZhang7/stable-diffusion-prompts-2.47M}}.
Unless stated otherwise, all watermark backends, detection thresholds (FPR $=10^{-6}$), and mask ratios are identical to the main paper.

\paragraph{Black-box forgery attacks.}
We follow the same black-box forgery experiment setting as Sec.~\ref{sec:Defense Against Forgery Attack}, and consider the two canonical attacks of M\"uller \emph{et al.}~\cite{muller2025black}: \emph{imprinting} and \emph{reprompting}.
As attacker models, we again use Stable Diffusion v2.1 (match) and Stable Diffusion v1.5 (mismatch).

\begin{table*}[t]
\centering
\setlength{\tabcolsep}{4pt}
\caption{Imprint forgery attack on the FredZhang7/stable-diffusion-prompts-2.47M dataset for four latent-based watermarking schemes and their SemBind-enhanced variants.}
\label{tab:imprint_FredZhang7}
\resizebox{\textwidth}{!}{%
\begin{tabular}{l c ccc ccc}
\toprule
 & & \multicolumn{3}{c}{Attacker Model: SD~2.1} & \multicolumn{3}{c}{Attacker Model: SD~1.5} \\
\cmidrule(lr){3-5} \cmidrule(lr){6-8}
Method & Step & Det.$\downarrow$ & Bit Acc.$\downarrow$ & PSNR & Det.$\downarrow$ & Bit Acc.$\downarrow$ & PSNR \\
\midrule
TR      & 50/100/150 & 1.00/1.00/1.00 & ---                  & 23.37/22.21/21.25 & 1.00/1.00/1.00 & ---                  & 22.86/21.83/21.15 \\
\rowcolor{gray!30}
TR-S    & 50/100/150 & 0.02/0.03/0.03 & ---                  & 23.32/22.12/21.35 & 0.01/0.02/0.02 & ---                  & 22.87/21.83/21.15 \\
GS      & 50/100/150 & 1.00/1.00/1.00 & 1.0000/1.0000/1.0000 & 23.35/22.15/21.39 & 1.00/1.00/1.00 & 0.9990/0.9996/0.9998 & 22.88/21.85/21.18 \\
\rowcolor{gray!30}
GS-S    & 50/100/150 & 0.05/0.05/0.07 & 0.4901/0.4926/0.5007 & 23.33/22.13/21.36 & 0.02/0.03/0.04 & 0.4936/0.4951/0.4902 & 22.89/21.84/21.16 \\
PRC     & 50/100/150 & 1.00/1.00/1.00 & 1.0000/1.0000/1.0000 & 23.36/22.16/21.39 & 1.00/1.00/1.00 & 1.0000/1.0000/1.0000 & 22.91/21.88/21.21 \\
\rowcolor{gray!30}
PRC-S   & 50/100/150 & 0.08/0.20/0.24 & 0.5402/0.5631/0.5210 & 23.38/22.18/21.41 & 0.01/0.02/0.02 & 0.5068/0.5112/0.5115 & 22.92/21.89/21.20 \\
GS++    & 50/100/150 & 1.00/1.00/1.00 & 0.9995/0.9999/1.0000 & 23.36/22.17/21.41 & 0.99/0.99/1.00 & 0.9803/0.9881/0.9944 & 22.90/21.87/21.18 \\
\rowcolor{gray!30}
GS++-S  & 50/100/150 & 0.18/0.38/0.51 & 0.5820/0.6845/0.7423 & 23.35/22.20/21.38 & 0.03/0.05/0.09 & 0.5125/0.5298/0.5469 & 22.90/21.87/21.19 \\
\bottomrule
\end{tabular}%
}
\end{table*}

\begin{table}[t]
  \centering
  \caption{Reprompt forgery attack for four latent-based watermarking schemes and their SemBind-enhanced variants (evaluated on FredZhang7/stable-diffusion-prompts-2.47M dataset).}
  \label{tab:reprompt_FredZhang7}
  \scriptsize
  \setlength{\tabcolsep}{6pt}
  \renewcommand{\arraystretch}{1.08}
  \begin{tabular}{lcccc}
    \toprule
    & \multicolumn{2}{c}{SD 2.1 attacker} & \multicolumn{2}{c}{SD 1.5 attacker} \\
    \cmidrule(lr){2-3}\cmidrule(lr){4-5}
    Method & Det.$\downarrow$ & Bit Acc.$\downarrow$ & Det.$\downarrow$ & Bit Acc.$\downarrow$ \\
    \midrule
    TR     & 1.00 & ---    & 0.99 & ---    \\
    \rowcolor{gray!30}
    TR-S   & 0.46 & ---    & 0.44 & ---    \\
    GS     & 0.99 & 0.9876 & 0.99 & 0.9884 \\
    \rowcolor{gray!30}
    GS-S   & 0.56 & 0.6228 & 0.50 & 0.6059 \\
    PRC    & 0.95 & 0.9723 & 0.93 & 0.9692 \\
    \rowcolor{gray!30}
    PRC-S  & 0.53 & 0.7421 & 0.13 & 0.5715 \\
    GS++   & 0.90 & 0.9461 & 0.87 & 0.8824 \\
    \rowcolor{gray!30}
    GS++-S & 0.45 & 0.7132 & 0.16 & 0.5756 \\
    \bottomrule
  \end{tabular}
\end{table}

\begin{table*}[t]
  \centering
  \small
  \setlength{\tabcolsep}{4pt}
  \caption{Robustness under common distortions on \texttt{FredZhang7/stable-diffusion-prompts-2.47M}. ``Average (Distortion)'' is the mean across the six distortion types (excluding ``None'').}
  \label{tab:robust_FredZhang7}
  \resizebox{\textwidth}{!}{%
  \begin{tabular}{lcccccccc}
    \toprule
    Method & None (Det./Acc.) & JPEG (QF=70) & Brightness (×1.0) & GauBlur ($r$=3) & GauNoise ($\sigma$=0.01) & MedFilter ($k$=7) & Resize (×0.5) & Average (Distortion)\\
    \midrule
    TR     & 1.00 & 0.97 & 1.00 & 1.00 & 1.00 & 1.00 & 1.00 & 0.995 \\
    \rowcolor{gray!30}
    TR-S   & 1.00 & 0.90 & 0.97 & 0.96 & 0.93 & 0.93 & 0.96 & 0.942 \\
    GS     & 1.00/1.0000 & 1.00/1.0000 & 1.00/0.9987 & 1.00/0.9961 & 1.00/0.9939 & 1.00/0.9982 & 1.00/0.9997 & 1.000/0.9978 \\
    \rowcolor{gray!30}
    GS-S   & 1.00/0.9873 & 0.98/0.9741 & 0.98/0.9804 & 0.99/0.9547 & 0.99/0.9593 & 0.99/0.9871 & 1.00/0.9803 & 0.988/0.9726 \\
    PRC    & 1.00/1.0000 & 1.00/1.0000 & 1.00/1.0000 & 0.94/0.9746 & 0.99/0.9950 & 0.94/0.9748 & 1.00/1.0000 & 0.978/0.9907 \\
    \rowcolor{gray!30}
    PRC-S  & 1.00/1.0000 & 0.98/0.9854 & 0.95/0.9598 & 0.71/0.8802 & 0.95/0.9748 & 0.75/0.8782 & 0.98/0.9950 & 0.887/0.9456 \\
    GS++   & 1.00/0.9998 & 0.98/0.9913 & 0.96/0.9841 & 0.97/0.9866 & 0.85/0.9259 & 0.94/0.9725 & 1.00/1.0000 & 0.950/0.9767 \\
    \rowcolor{gray!30}
    GS++-S & 1.00/0.9998 & 0.96/0.9822 & 0.94/0.9691 & 0.86/0.9344 & 0.82/0.9075 & 0.80/0.9019 & 0.97/0.9856 & 0.892/0.9468 \\
    \bottomrule
  \end{tabular}%
  }
\end{table*}

\paragraph{Imprinting.}
For each watermarking scheme and attacker model, we sample 100 prompts from the 2.47M prompt corpus and generate 100 corresponding watermarked images.
As target cover images, we randomly sample 100 natural photographs from the MS-COCO validation set.
We then run the imprinting optimization exactly as in~\cite{muller2025black} and Sec.~\ref{sec:Defense Against Forgery Attack} (150 gradient steps with learning rate $0.01$), and probe watermark detection every 50 steps.

Results are reported in Table~\ref{tab:imprint_FredZhang7}. Overall, SemBind continues to provide strong protection against imprinting on this larger and more realistic prompt corpus. In particular, imprinting is largely ineffective against TR-S and GS-S, with detection rates remaining near zero throughout optimization in both the match (SD~2.1) and mismatch (SD~1.5) settings. For PRC and GS++, SemBind also substantially reduces the attack success rate, yielding markedly lower detection/bit-recovery performance compared to the corresponding baselines across all probed steps.

\paragraph{Reprompting.}
For each scheme, we sample 100 prompts from the 2.47M corpus to generate watermarked images, and sample an additional 100 mismatched prompts from the I2P dataset\footnote{\url{https://huggingface.co/datasets/AIML-TUDA/i2p}} for reprompting.
The adversary reuses the estimated watermarked initial latent and runs the proxy model forward under the mismatched prompts, following the same procedure as in Sec.~\ref{sec:Defense Against Forgery Attack}.

Results are summarized in Table~\ref{tab:reprompt_FredZhang7}. Overall, SemBind remains highly effective against reprompting on this larger prompt corpus: for all four backends, the SemBind-enabled variants exhibit substantially reduced forgery success compared to their corresponding baselines, under both the match (SD~2.1) and mismatch (SD~1.5) attacker models.

\paragraph{Robustness.}
We follow the same robustness experiment settings as in Sec.~\ref{robustness main paper}.

The results in Table~\ref{tab:robust_FredZhang7} show that SemBind generalizes well to this larger prompt corpus: although most prompts in \texttt{FredZhang7/stable-diffusion-prompts-2.47M} do not appear in the data used to train the semantic masker, the SemBind-enabled variants retain strong robustness under most distortions. 
As in the main paper, PRC and GS++ exhibit larger robustness drops under the stronger filtering distortions (GauBlur and MedFilter). 
We attribute this primarily to the limited strength/frequency of such filtering-style augmentations during semantic-masker training; in principle, this can be mitigated by adjusting the augmentation recipe or further fine-tuning the masker on augmentations that better match these distortions.

\subsection{Generalization on FLUX}
\label{app:FLUX}
\paragraph{Experimental setup.}
To demonstrate that SemBind can transfer to other image generation models, we further evaluate it on \textsc{FLUX}.
To balance feasibility and computational cost, we focus on Gaussian Shading (GS) and its SemBind-enabled variant (GS-S), and deploy them on \textsc{FLUX}~1.\texttt{dev}.
We reuse the same trained semantic masker (without retraining) and integrate it with GS in the \textsc{FLUX} pipeline.
All experiments are conducted at $512{\times}512$ resolution, where the latent has shape $1{\times}16{\times}64{\times}64$, and we embed a 256-bit payload.
We then evaluate forgery resistance and robustness on both the SDP and COCO prompt sets.

\paragraph{Black-box forgery attacks.}
We follow the same black-box forgery experiment settings as the main paper, and consider the two canonical attacks of M\"uller \emph{et al.}~\cite{muller2025black}: \emph{imprinting} and \emph{reprompting}.

We use Stable Diffusion v2.1 as the attacker model; note that this corresponds to the mismatch case for \textsc{FLUX}.

\begin{table}[t]
  \centering
  \setlength{\tabcolsep}{4pt}
  \caption{\textsc{FLUX} imprinting attack results.}
  \label{tab:flux_imprinting}
  \small
  \begin{tabular}{lcccccccc}
    \toprule
    & \multicolumn{4}{c}{COCO} & \multicolumn{4}{c}{SDP} \\
    \cmidrule(lr){2-5}\cmidrule(lr){6-9}
    Method & Step & Det.$\downarrow$ & Bit Acc.$\downarrow$ & PSNR
           & Step & Det.$\downarrow$ & Bit Acc.$\downarrow$ & PSNR \\
    \midrule
    GS   & 50/100/150 & 0.86/0.91/0.94 & 0.7566/0.8930/0.9254 & 23.4399
         & 50/100/150 & 0.80/0.90/0.92 & 0.7629/0.8730/0.9174 & 23.4796 \\
    \rowcolor{gray!30}
    GS-S & 50/100/150 & 0.00/0.00/0.00 & 0.5020/0.5000/0.5039 & 23.4579
         & 50/100/150 & 0.00/0.00/0.00 & 0.4875/0.4832/0.4934 & 23.4622 \\
    \bottomrule
  \end{tabular}
\end{table}

\begin{table}[t]
  \centering
  \setlength{\tabcolsep}{8pt}
  \caption{\textsc{FLUX} reprompting attack results.}
  \label{tab:flux_reprompting}
  \small
  \begin{tabular}{lccccc}
    \toprule
    & \multicolumn{2}{c}{COCO} & \multicolumn{2}{c}{SDP} \\
    \cmidrule(lr){2-3}\cmidrule(lr){4-5}
    Method & Det.$\downarrow$ & Bit Acc.$\downarrow$ & Det.$\downarrow$ & Bit Acc.$\downarrow$ \\
    \midrule
    GS     & 0.47 & 0.6453 & 0.78 & 0.7191 \\
    \rowcolor{gray!30}
    GS-S   & 0.01 & 0.5039 & 0.00 & 0.5070 \\
    \bottomrule
  \end{tabular}
\end{table}

\begin{table*}[t]
  \centering
  \small
  \setlength{\tabcolsep}{4pt}
  \renewcommand{\arraystretch}{1.08}
  \caption{\textsc{FLUX} robustness under common distortions. Rows 1--2 report COCO results and rows 3--4 report SDP results. ``Average (Distortion)'' is the mean across the six distortion types (excluding ``None'').}
  \label{tab:robust_FLUX}
  \resizebox{\textwidth}{!}{%
  \begin{tabular}{lcccccccc}
    \toprule
    Method & None (Det./Acc.) & JPEG (QF=70) & Brightness ($\times 1.0$) & GauBlur ($r$=3) & GauNoise ($\sigma$=0.01) & MedFilter ($k$=7) & Resize ($\times 0.5$) & Average (Distortion) \\
    \midrule
    GS   & 1.0/1.0000 & 0.99/0.9789 & 0.99/0.9928 & 1.0/0.9911 & 0.99/0.9937 & 1.0/0.9937 & 1.0/0.9988 & 0.995/0.9915 \\
    \rowcolor{gray!30}
    GS-S & 1.0/0.9989 & 1.0/0.9625  & 1.0/0.9939  & 1.0/0.9668 & 1.0/0.9884  & 1.0/0.9816  & 1.0/0.9938  & 1.000/0.9812 \\
    \midrule
    GS   & 1.0/0.9995 & 0.99/0.9933 & 1.0/0.9966 & 1.0/0.9894 & 0.99/0.9954 & 1.0/0.9909 & 1.0/0.9977 & 0.997/0.9939 \\
    \rowcolor{gray!30}
    GS-S & 0.98/0.9872 & 0.98/0.9592 & 0.98/0.9867 & 0.98/0.9272 & 0.97/0.9757 & 0.98/0.9300 & 1.0/0.9864 & 0.982/0.9609 \\
    \bottomrule
  \end{tabular}%
  }
\end{table*}

\paragraph{Imprinting.}
We sample 100 prompts from each of the SDP and COCO prompt sets and generate 100 corresponding watermarked images with \textsc{FLUX}~1.\texttt{dev}.
As target cover images, we randomly sample 100 natural photographs from the MS-COCO validation set.
We then run the imprinting optimization for 150 gradient steps with learning rate $0.01$, and probe watermark detection every 50 steps.

Results are reported in Table~\ref{tab:flux_imprinting}.
Notably, although the attacker uses Stable Diffusion v2.1 as the proxy model (a mismatch setting with a substantial gap to \textsc{FLUX}), the baseline GS still achieves over 80\% forgery success.
In contrast, the SemBind-enabled variant (GS-S) fully suppresses imprinting across all probed steps, demonstrating that SemBind transfers effectively to a different generator backbone and provides strong anti-forgery generalization.

\paragraph{Reprompting.}
We sample 100 prompts from each of the SDP and COCO prompt sets to generate watermarked images, and sample an additional 100 mismatched prompts from the I2P dataset for the reprompting attack.
The adversary reuses the estimated watermarked initial latent and runs the SD~2.1 proxy model forward under the mismatched prompts, following the same procedure as in the main paper.

Results are reported in Table~\ref{tab:flux_reprompting}.
Overall, SemBind remains effective against reprompting on \textsc{FLUX}: compared to the baseline GS, the SemBind-enabled variant (GS-S) substantially reduces forgery success across both COCO and SDP.

\paragraph{Robustness.}
We follow the same robustness protocol as in the Sec.~\ref{robustness main paper}.
We evaluate GS and its SemBind-enabled variant (GS-S) on \textsc{FLUX}~1.\texttt{dev} using 100 prompts from each of the COCO and SDP prompt sets.

Results are reported in Table~\ref{tab:robust_FLUX}.
Overall, SemBind transfers to \textsc{FLUX} without introducing notable robustness issues: GS-S maintains high detection/bit-accuracy under all tested distortions, with only a modest drop relative to the GS baseline.

\section{Proof of Theorem 1}
\label{sec:proof_theorem1}

\paragraph{Notation.}
Let $L=CHW$ be the latent dimension.
Write $\mathcal{N}\equiv\mathcal{N}(0,I_L)$ for the $L$-dimensional standard Gaussian, and $\mathcal{N}^{\otimes t}$ for $t$ i.i.d.\ copies.
For a sign mask $S\in\{\pm1\}^L$ and a permutation $\pi$ of $[L]$, denote by
\[
F_{S,\pi}(u) \;=\; P_\pi\,\mathrm{Diag}(S)\,u\quad\text{for }u\in\mathbb{R}^L,
\]
where $P_\pi$ is the permutation matrix.  We allow $(S,\pi)$ to be jointly sampled (and even shared across $t$ samples), but require $(S,\pi)$ be independent of the challenge latents.

\begin{definition}[Single-/multi-sample undetectability]
\label{def:undetectability}
A watermarking scheme $\mathcal{W}$ (on initial latents) is \textbf{single-sample $(\varepsilon,\mathsf{Adv})$-undetectable} if for every (non-uniform) adversary $D\in\mathsf{Adv}$, satisfying:
\begin{equation}
    \bigl|\Pr[D(Z)=1]-\Pr[D(G)=1]\bigr|\le\varepsilon,
\end{equation}
with $Z\!\leftarrow\!\mathcal{W},\; G\!\leftarrow\!\mathcal{N}$.

It is \textbf{multi-sample $(t,\varepsilon,\mathsf{Adv})$-undetectable} if for every $D\in\mathsf{Adv}$, satisfying:
\begin{equation}
\bigl|\Pr[D(Z_1,\ldots,Z_t)=1]-\Pr[D(G_1,\ldots,G_t)=1]\bigr|\le\varepsilon,
\end{equation}
with $(Z_i)_{i=1}^t\!\xleftarrow{\text{i.i.d.}}\!\mathcal{W}$ and $(G_i)_{i=1}^t\!\xleftarrow{\text{i.i.d.}}\!\mathcal{N}$.
Here $\mathsf{Adv}$ may be any probabilistic polynomial-time (PPT) adversary~\cite{goldwasser1984probabilistic,goldreich2001foundations}.

\end{definition}

\begin{definition}[SemBind post-processing on initial latents]
\label{def:sembind}
Given a base scheme $\mathcal{W}$ that outputs $Z\!\leftarrow\!\mathcal{W}$, its SemBind variant first samples randomness $R$ (independent of $Z$) which determines a sign mask $S\in\{\pm1\}^L$ and a permutation $\pi$ (both possibly correlated across samples and possibly encoding a mask ratio by forcing some entries of $S$ to $+1$), and outputs
\[
Z^{\mathrm{sem}} \;=\; F_{S,\pi}(Z).
\]
\end{definition}

\begin{lemma}[Gaussian invariance under independent sign flips and permutations]
\label{lem:gauss_inv}
Let $G\!\leftarrow\!\mathcal{N}$ and $(S,\pi)\perp G$. Then $F_{S,\pi}(G)\stackrel{d}{=}G$. The same holds for $t$ i.i.d.\ copies jointly transformed by $\{(S_i,\pi_i)\}_{i=1}^t$ independent of $\{G_i\}_{i=1}^t$.
\end{lemma}

\begin{proof}
For fixed $(S,\pi)$, the map $u\mapsto P_\pi\mathrm{Diag}(S)u$ permutes coordinates and flips signs, hence preserves $\mathcal{N}(0,I_L)$. Independence allows averaging over $(S,\pi)$.
\end{proof}

\begin{lemma}[Closure under independent post-processing]
\label{lem:closure}
Let $X,Y$ be random variables on $\mathbb{R}^L$, and $T$ any (possibly randomized) map independent of $(X,Y)$. Then for any distinguisher class $\mathsf{Adv}$ and $\varepsilon\!\ge\!0$,
\begin{equation}
\Delta_{\mathsf{Adv}}\bigl(T(X),T(Y)\bigr)\;\le\;\Delta_{\mathsf{Adv}}(X,Y),
\end{equation}
where $\Delta_{\mathsf{Adv}}(U,V)=\sup_{D\in\mathsf{Adv}}|\Pr[D(U)=1]-\Pr[D(V)=1]|$.
\end{lemma}

\begin{proof}
For fixed $D$, define $D_T(u)=\mathbb{E}[D(T(u))]$ over $T$'s randomness. Then
$|\Pr[D(T(X))=1]-\Pr[D(T(Y))=1]|
=|\Pr[D_T(X)=1]-\Pr[D_T(Y)=1]|
\le\Delta_{\mathsf{Adv}}(X,Y)$.
\end{proof}

\begin{theorem}[Semantic masking preserves undetectability]
\label{thm:sem_preserves}
If $\mathcal{W}$ is single-sample $(\varepsilon,\mathsf{Adv})$-undetectable \emph{(resp.\ multi-sample $(t,\varepsilon,\mathsf{Adv})$-undetectable)}, then its SemBind variant $\mathcal{W}^{\mathrm{sem}}$ per Def.~\ref{def:sembind} is single-sample \emph{(resp.\ multi-sample)} $(\varepsilon,\mathsf{Adv})$-undetectable as well.
\end{theorem}

\begin{proof}
\textbf{Single-sample.} Let $Z\!\leftarrow\!\mathcal{W}$, $G\!\leftarrow\!\mathcal{N}$, and $T(\cdot)=F_{S,\pi}(\cdot)$ with $(S,\pi)\perp(Z,G)$. By Lemma~\ref{lem:closure},
\begin{equation}
\Delta_{\mathsf{Adv}}\!\bigl(F_{S,\pi}(Z),\,F_{S,\pi}(G)\bigr)\;\le\;\Delta_{\mathsf{Adv}}(Z,G)\;\le\;\varepsilon.
\end{equation}
By Lemma~\ref{lem:gauss_inv}, $F_{S,\pi}(G)\stackrel{d}{=}G$, hence $\Delta_{\mathsf{Adv}}\bigl(F_{S,\pi}(Z),G\bigr)\le\varepsilon$.

\noindent\textbf{Multi-sample.} Let $(Z_i)_{i=1}^t\!\xleftarrow{\text{i.i.d.}}\!\mathcal{W}$, $(G_i)_{i=1}^t\!\xleftarrow{\text{i.i.d.}}\!\mathcal{N}$, and $T((u_i)_{i=1}^t)=(F_{S_i,\pi_i}(u_i))_{i=1}^t$ with $\{(S_i,\pi_i)\}_{i=1}^t\perp\{(Z_i,G_i)\}_{i=1}^t$. Then we have:
\begin{equation}
\begin{aligned}
\Delta_{\mathsf{Adv}}\!\bigl((F_{S_i,\pi_i}&(Z_i))_{i=1}^t,\,(F_{S_i,\pi_i}(G_i))_{i=1}^t\bigr) \\
&\le\;
\Delta_{\mathsf{Adv}}\!\bigl((Z_i)_{i=1}^t,(G_i)_{i=1}^t\bigr)\le\;\varepsilon.
\end{aligned}
\end{equation}

and Lemma~\ref{lem:gauss_inv} yields $(F_{S_i,\pi_i}(G_i))_{i=1}^t\stackrel{d}{=}(G_i)_{i=1}^t$, concluding the claim.
\end{proof}

\begin{corollary}[Instantiations]
If Gaussian Shading is single-sample undetectable, then its SemBind variant (GS-S) remains single-sample undetectable with the same bound. If PRC or Gaussian Shading++ is multi-sample undetectable, then PRC-S and GS++-S inherit the same multi-sample bound.
\end{corollary}

\begin{remark}[Scope]
(1) Only independence of $(S,\pi)$ from the challenge samples is required; mask ratio (some coordinates forced to $+1$), shared/correlated masks across samples, and any secret permutation are all permitted. 
(2) No assumptions on bit balance or inter-bit independence are needed for undetectability (they matter for robustness/forgery resistance, not for indistinguishability).
\end{remark}

\section{Additional Visual Results}
\label{sec:additional_visuals}

\newcommand{\imgpath}{figs/sup_figs_sdp} 

\newcommand{\RowLabel}[1]{%
  \noindent\makebox[\linewidth]{\rule{\linewidth}{1pt}}%
  \par\medskip
  {\footnotesize\bfseries #1\par}%
  \noindent\makebox[\linewidth]{\rule{\linewidth}{1pt}}%
}

\begin{figure*}[t]
  \centering
  \begin{minipage}[b]{.1245\textwidth}\centering
    \begin{overpic}[width=\linewidth]{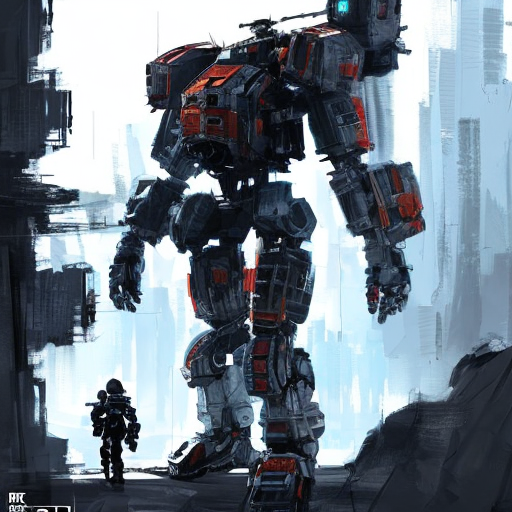}
      \put(3,6){\colorbox{black!60}{\scriptsize\color{white} (a) True}}
    \end{overpic}
  \end{minipage}%
  \begin{minipage}[b]{.1245\textwidth}\centering
    \begin{overpic}[width=\linewidth]{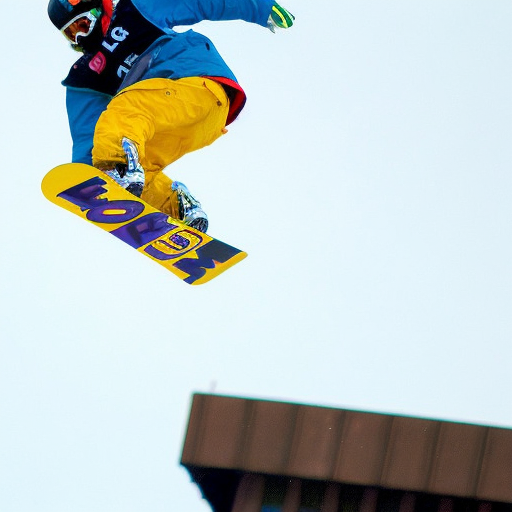}
      \put(3,6){\colorbox{black!60}{\scriptsize\color{white} (b) False}}
    \end{overpic}
  \end{minipage}%
  \begin{minipage}[b]{.1245\textwidth}\centering
    \begin{overpic}[width=\linewidth]{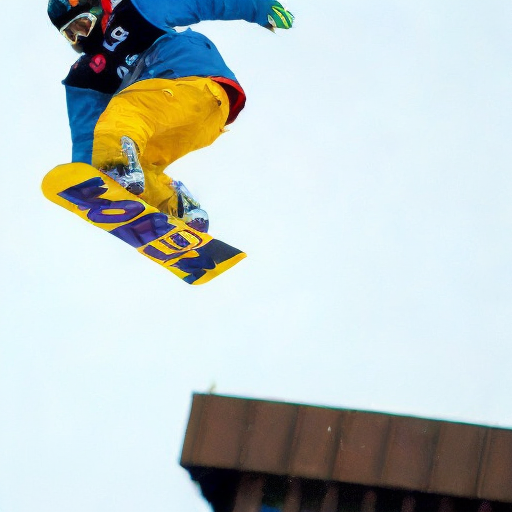}
      \put(3,6){\colorbox{black!60}{\scriptsize\color{white} (c) True}}
    \end{overpic}
  \end{minipage}%
  \begin{minipage}[b]{.1245\textwidth}\centering
    \begin{overpic}[width=\linewidth]{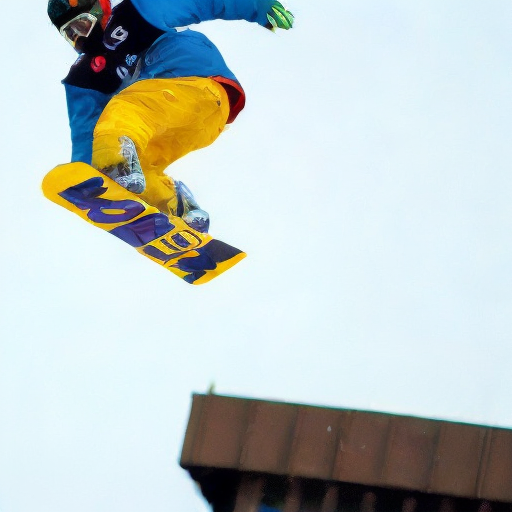}
      \put(3,6){\colorbox{black!60}{\scriptsize\color{white} (d) True}}
    \end{overpic}
  \end{minipage}%
  \begin{minipage}[b]{.1245\textwidth}\centering
    \begin{overpic}[width=\linewidth]{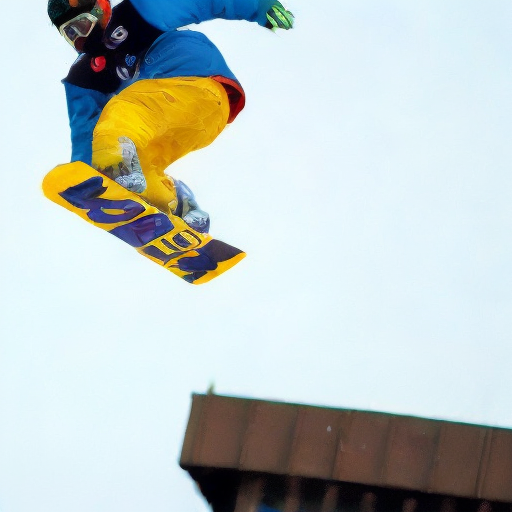}
      \put(3,6){\colorbox{black!60}{\scriptsize\color{white} (e) True}}
    \end{overpic}
  \end{minipage}%
  \begin{minipage}[b]{.1245\textwidth}\centering
    \begin{overpic}[width=\linewidth]{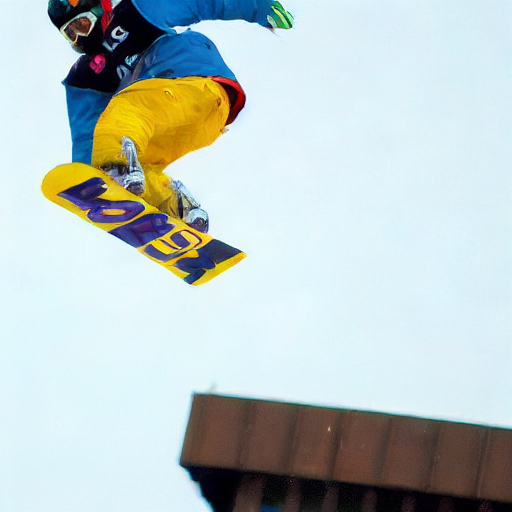}
      \put(3,6){\colorbox{black!60}{\scriptsize\color{white} (f) True}}
    \end{overpic}
  \end{minipage}%
  \begin{minipage}[b]{.1245\textwidth}\centering
    \begin{overpic}[width=\linewidth]{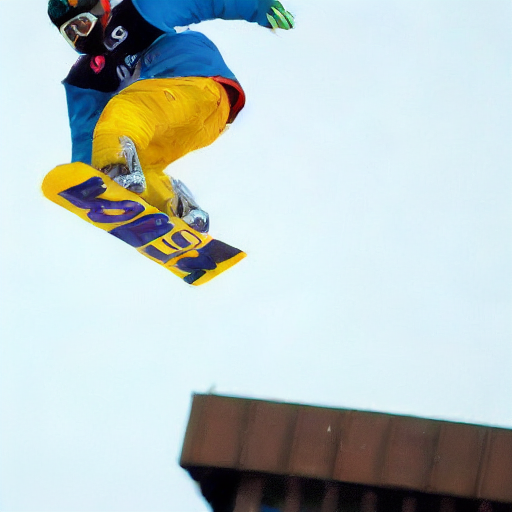}
      \put(3,6){\colorbox{black!60}{\scriptsize\color{white} (g) True}}
    \end{overpic}
  \end{minipage}%
  \begin{minipage}[b]{.1245\textwidth}\centering
    \begin{overpic}[width=\linewidth]{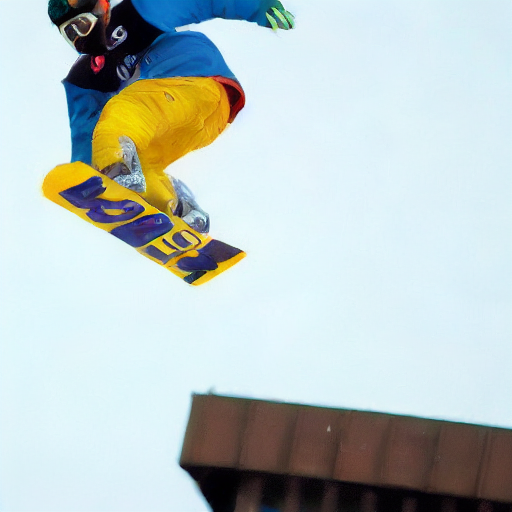}
      \put(3,6){\colorbox{black!60}{\scriptsize\color{white} (h) True}}
    \end{overpic}
  \end{minipage}%

  \par

  \begin{minipage}[b]{.1245\textwidth}\centering
    \begin{overpic}[width=\linewidth]{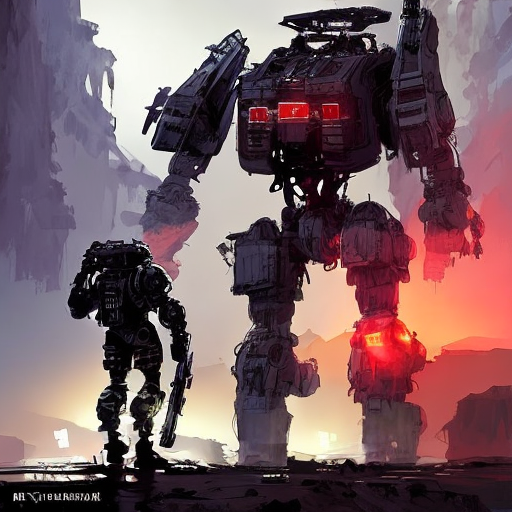}
      \put(3,6){\colorbox{black!60}{\scriptsize\color{white} (a) True}}
    \end{overpic}
  \end{minipage}%
  \begin{minipage}[b]{.1245\textwidth}\centering
    \begin{overpic}[width=\linewidth]{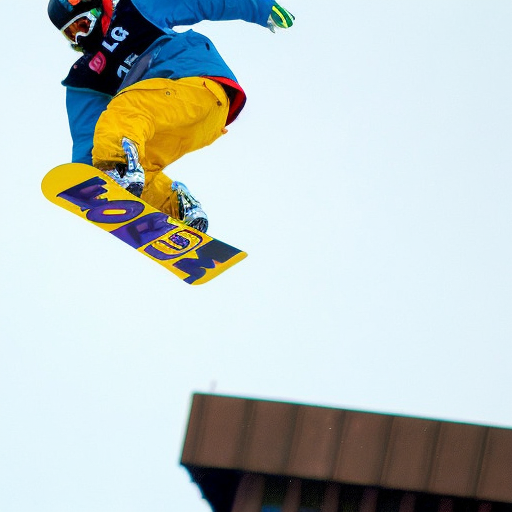}
      \put(3,6){\colorbox{black!60}{\scriptsize\color{white} (b) False}}
    \end{overpic}
  \end{minipage}%
  \begin{minipage}[b]{.1245\textwidth}\centering
    \begin{overpic}[width=\linewidth]{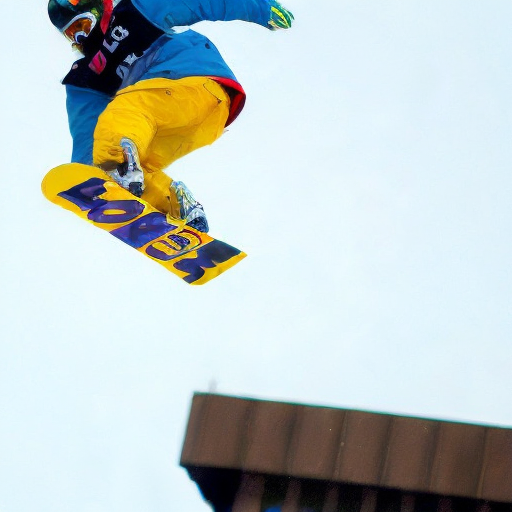}
      \put(3,6){\colorbox{black!60}{\scriptsize\color{white} (c) False}}
    \end{overpic}
  \end{minipage}%
  \begin{minipage}[b]{.1245\textwidth}\centering
    \begin{overpic}[width=\linewidth]{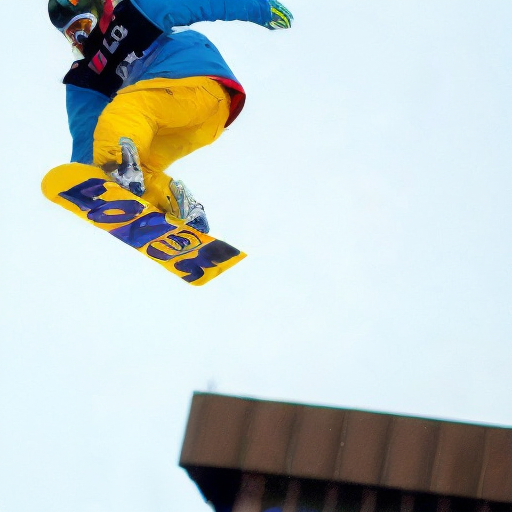}
      \put(3,6){\colorbox{black!60}{\scriptsize\color{white} (d) False}}
    \end{overpic}
  \end{minipage}%
  \begin{minipage}[b]{.1245\textwidth}\centering
    \begin{overpic}[width=\linewidth]{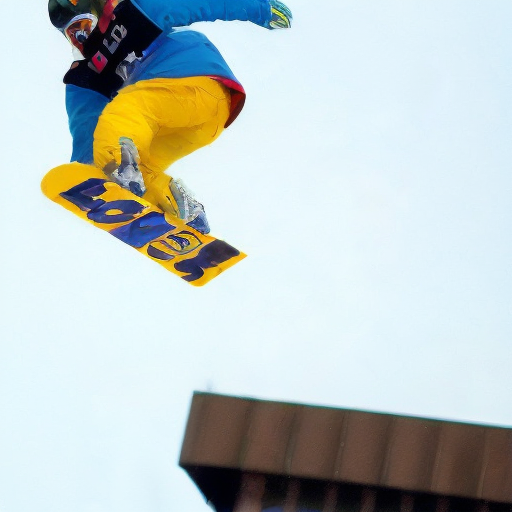}
      \put(3,6){\colorbox{black!60}{\scriptsize\color{white} (e) False}}
    \end{overpic}
  \end{minipage}%
  \begin{minipage}[b]{.1245\textwidth}\centering
    \begin{overpic}[width=\linewidth]{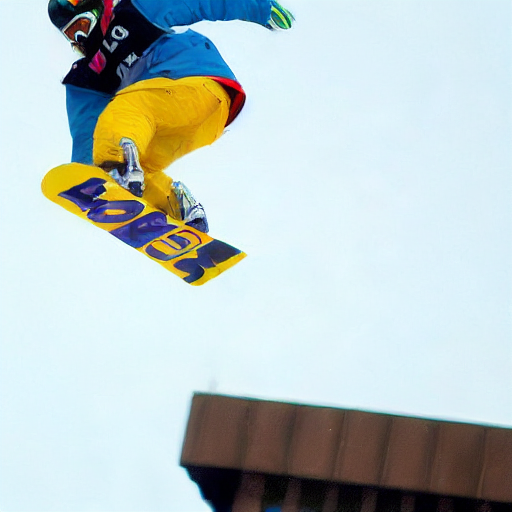}
      \put(3,6){\colorbox{black!60}{\scriptsize\color{white} (f) False}}
    \end{overpic}
  \end{minipage}%
  \begin{minipage}[b]{.1245\textwidth}\centering
    \begin{overpic}[width=\linewidth]{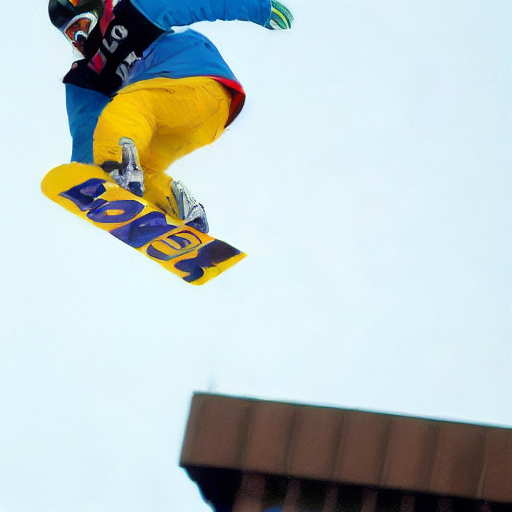}
      \put(3,6){\colorbox{black!60}{\scriptsize\color{white} (g) False}}
    \end{overpic}
  \end{minipage}%
  \begin{minipage}[b]{.1245\textwidth}\centering
    \begin{overpic}[width=\linewidth]{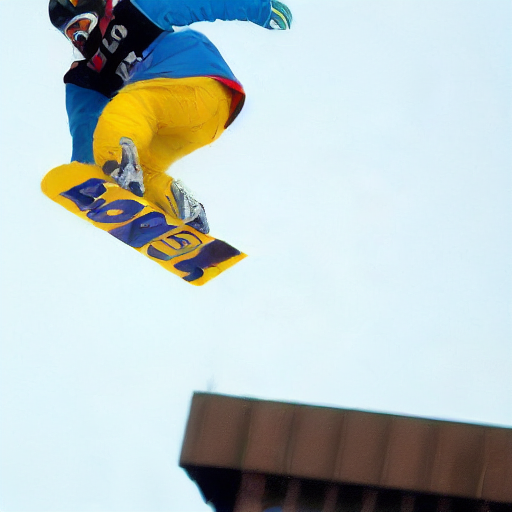}
      \put(3,6){\colorbox{black!60}{\scriptsize\color{white} (h) False}}
    \end{overpic}
  \end{minipage}%

  \par

  \begin{minipage}[b]{.1245\textwidth}\centering
    \begin{overpic}[width=\linewidth]{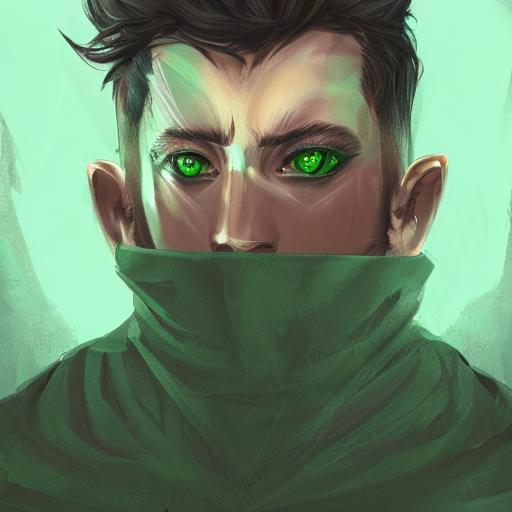}
      \put(3,6){\colorbox{black!60}{\scriptsize\color{white} (a) True/1.0000}}
    \end{overpic}
  \end{minipage}%
  \begin{minipage}[b]{.1245\textwidth}\centering
    \begin{overpic}[width=\linewidth]{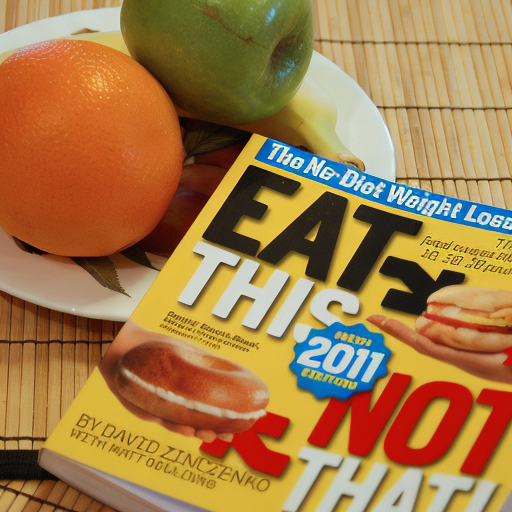}
      \put(3,6){\colorbox{black!60}{\scriptsize\color{white} (b) False/0.5044}}
    \end{overpic}
  \end{minipage}%
  \begin{minipage}[b]{.1245\textwidth}\centering
    \begin{overpic}[width=\linewidth]{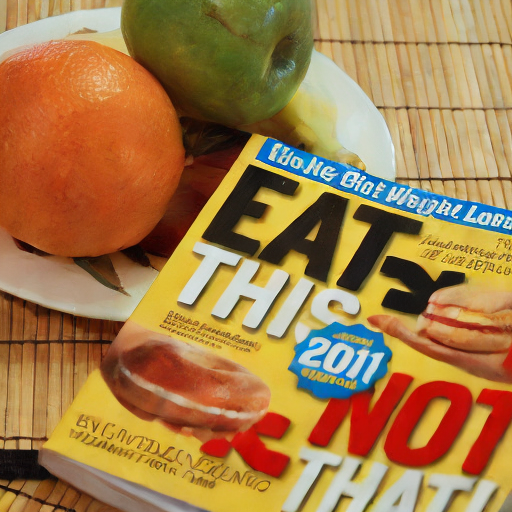}
      \put(3,6){\colorbox{black!60}{\scriptsize\color{white} (c) True/1.0000}}
    \end{overpic}
  \end{minipage}%
  \begin{minipage}[b]{.1245\textwidth}\centering
    \begin{overpic}[width=\linewidth]{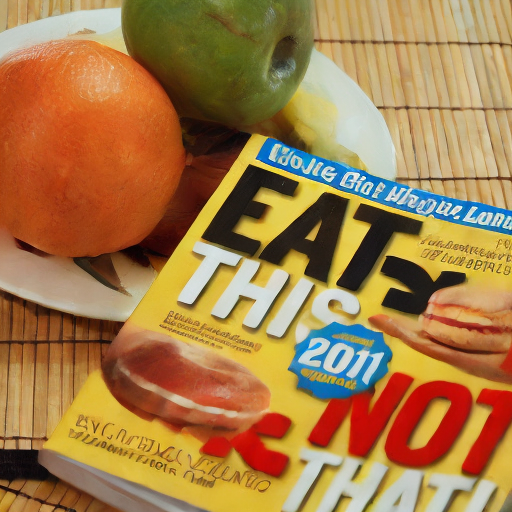}
      \put(3,6){\colorbox{black!60}{\scriptsize\color{white} (d) True/1.0000}}
    \end{overpic}
  \end{minipage}%
  \begin{minipage}[b]{.1245\textwidth}\centering
    \begin{overpic}[width=\linewidth]{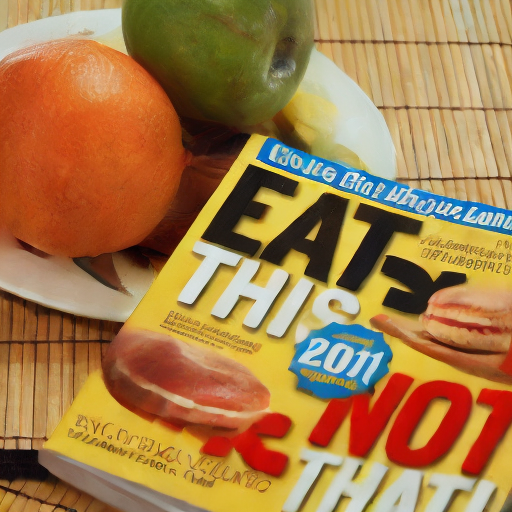}
      \put(3,6){\colorbox{black!60}{\scriptsize\color{white} (e) True/1.0000}}
    \end{overpic}
  \end{minipage}%
  \begin{minipage}[b]{.1245\textwidth}\centering
    \begin{overpic}[width=\linewidth]{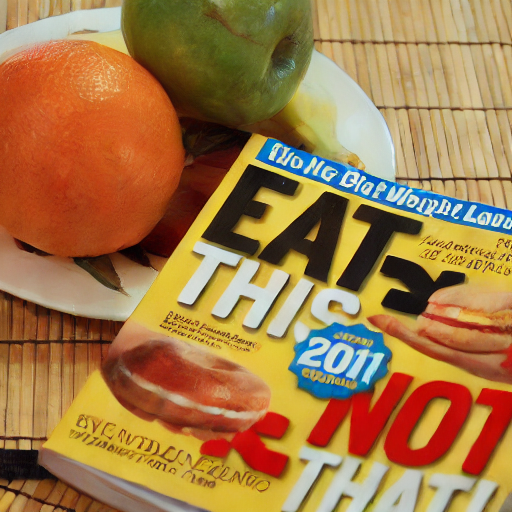}
      \put(3,6){\colorbox{black!60}{\scriptsize\color{white} (f) True/1.0000}}
    \end{overpic}
  \end{minipage}%
  \begin{minipage}[b]{.1245\textwidth}\centering
    \begin{overpic}[width=\linewidth]{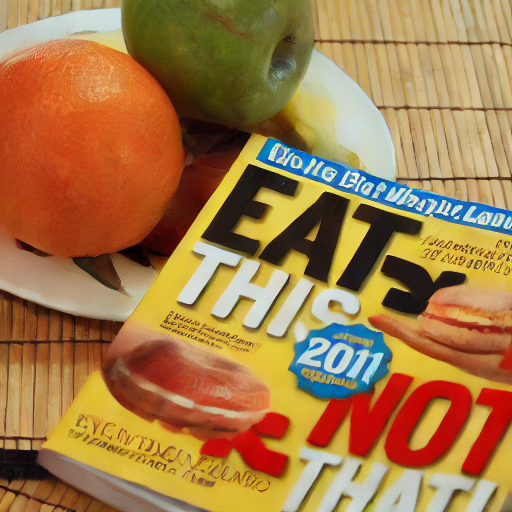}
      \put(3,6){\colorbox{black!60}{\scriptsize\color{white} (g) True/1.0000}}
    \end{overpic}
  \end{minipage}%
  \begin{minipage}[b]{.1245\textwidth}\centering
    \begin{overpic}[width=\linewidth]{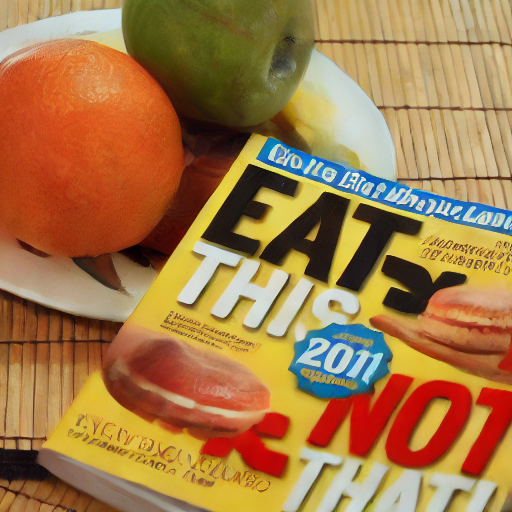}
      \put(3,6){\colorbox{black!60}{\scriptsize\color{white} (h) True/1.0000}}
    \end{overpic}
  \end{minipage}%

  \par

  \begin{minipage}[b]{.1245\textwidth}\centering
    \begin{overpic}[width=\linewidth]{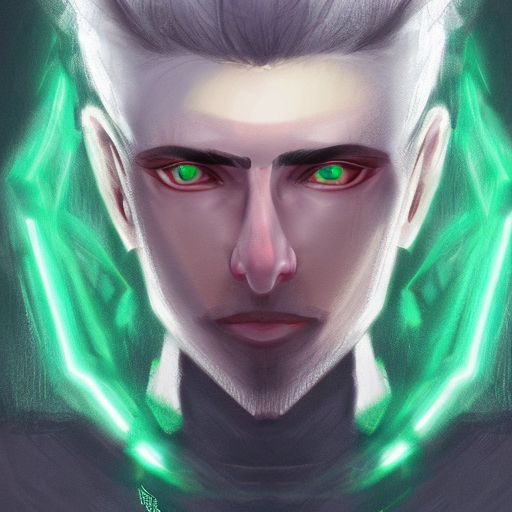}
      \put(3,6){\colorbox{black!60}{\scriptsize\color{white} (a) True/1.0000}}
    \end{overpic}
  \end{minipage}%
  \begin{minipage}[b]{.1245\textwidth}\centering
    \begin{overpic}[width=\linewidth]{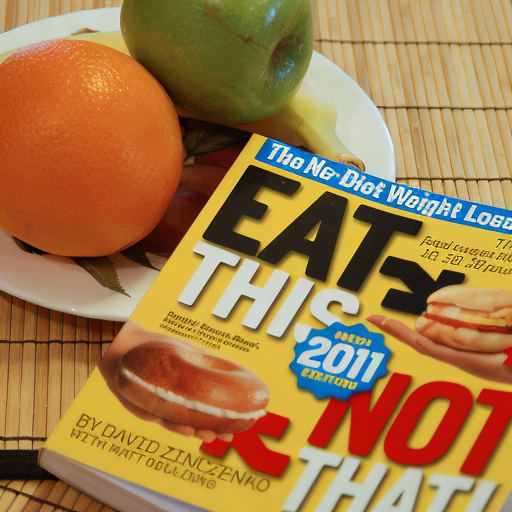}
      \put(3,6){\colorbox{black!60}{\scriptsize\color{white} (b) False/0.5022}}
    \end{overpic}
  \end{minipage}%
  \begin{minipage}[b]{.1245\textwidth}\centering
    \begin{overpic}[width=\linewidth]{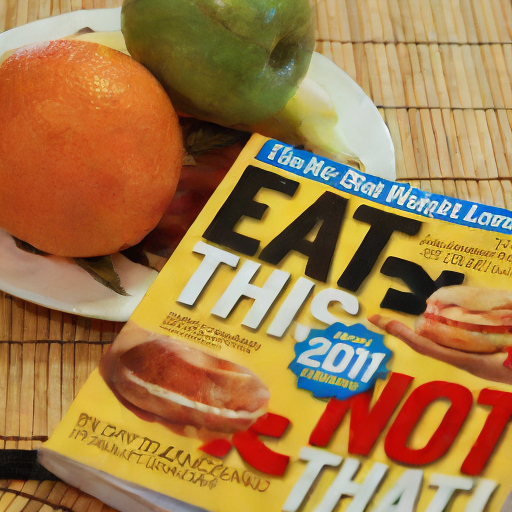}
      \put(3,6){\colorbox{black!60}{\scriptsize\color{white} (c) False/0.5273}}
    \end{overpic}
  \end{minipage}%
  \begin{minipage}[b]{.1245\textwidth}\centering
    \begin{overpic}[width=\linewidth]{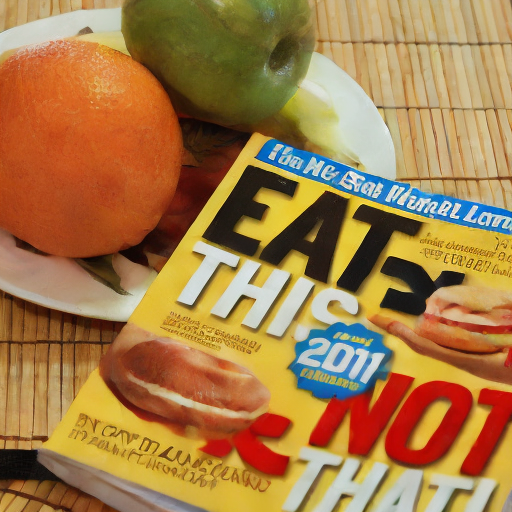}
      \put(3,6){\colorbox{black!60}{\scriptsize\color{white} (d) False/0.5078}}
    \end{overpic}
  \end{minipage}%
  \begin{minipage}[b]{.1245\textwidth}\centering
    \begin{overpic}[width=\linewidth]{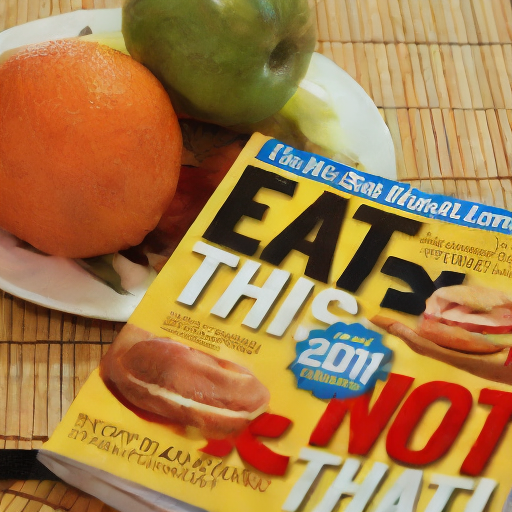}
      \put(3,6){\colorbox{black!60}{\scriptsize\color{white} (e) False/0.5820}}
    \end{overpic}
  \end{minipage}%
  \begin{minipage}[b]{.1245\textwidth}\centering
    \begin{overpic}[width=\linewidth]{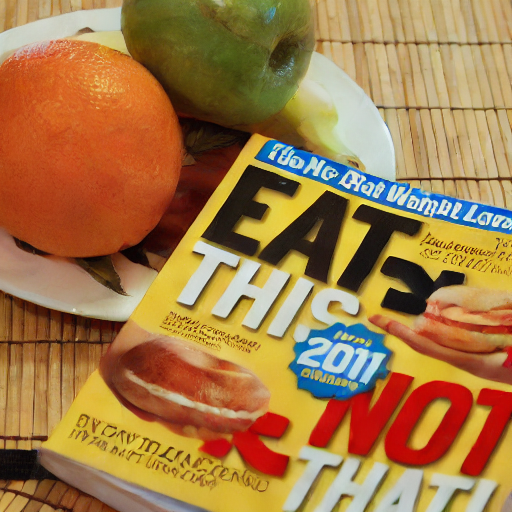}
      \put(3,6){\colorbox{black!60}{\scriptsize\color{white} (f) False/0.5000}}
    \end{overpic}
  \end{minipage}%
  \begin{minipage}[b]{.1245\textwidth}\centering
    \begin{overpic}[width=\linewidth]{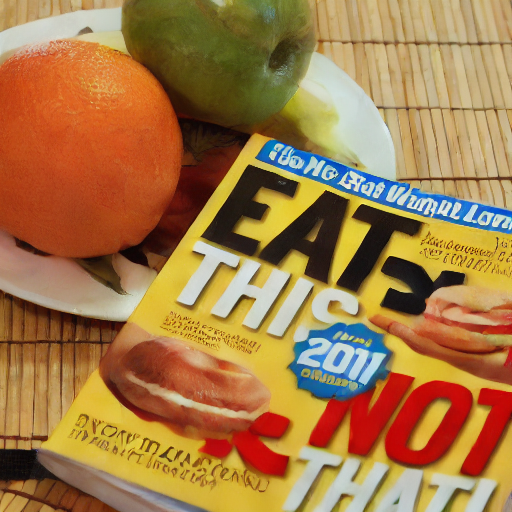}
      \put(3,6){\colorbox{black!60}{\scriptsize\color{white} (g) False/0.5078}}
    \end{overpic}
  \end{minipage}%
  \begin{minipage}[b]{.1245\textwidth}\centering
    \begin{overpic}[width=\linewidth]{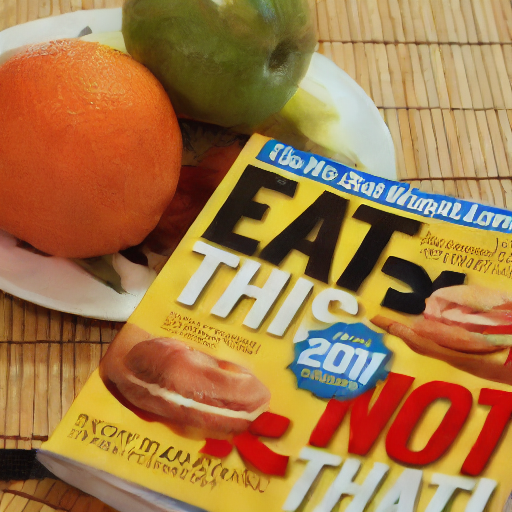}
      \put(3,6){\colorbox{black!60}{\scriptsize\color{white} (h) False/0.5312}}
    \end{overpic}
  \end{minipage}%

  \par

  \begin{minipage}[b]{.1245\textwidth}\centering
    \begin{overpic}[width=\linewidth]{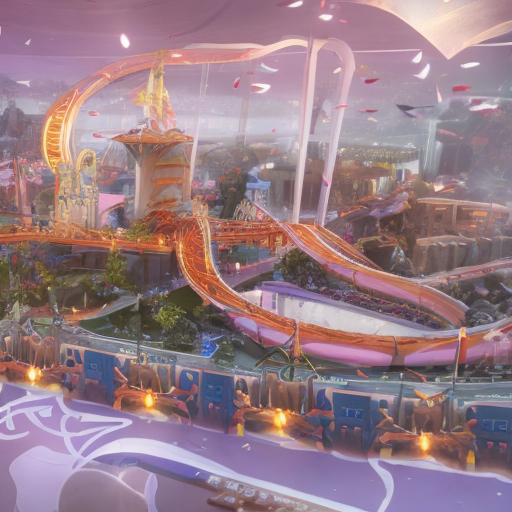}
      \put(3,6){\colorbox{black!60}{\scriptsize\color{white} (a) True/1.0000}}
    \end{overpic}
  \end{minipage}%
  \begin{minipage}[b]{.1245\textwidth}\centering
    \begin{overpic}[width=\linewidth]{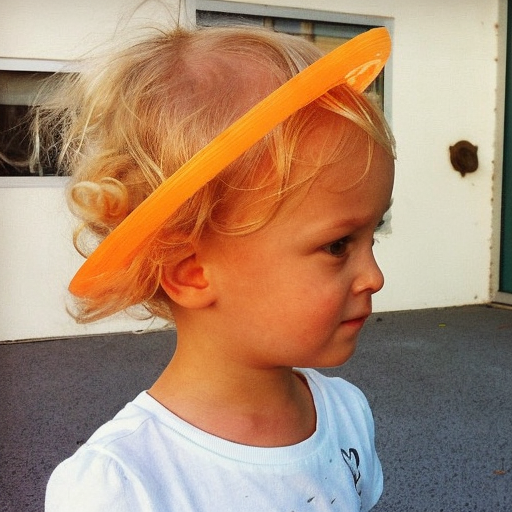}
      \put(3,6){\colorbox{black!60}{\scriptsize\color{white} (b) False/0.5274}}
    \end{overpic}
  \end{minipage}%
  \begin{minipage}[b]{.1245\textwidth}\centering
    \begin{overpic}[width=\linewidth]{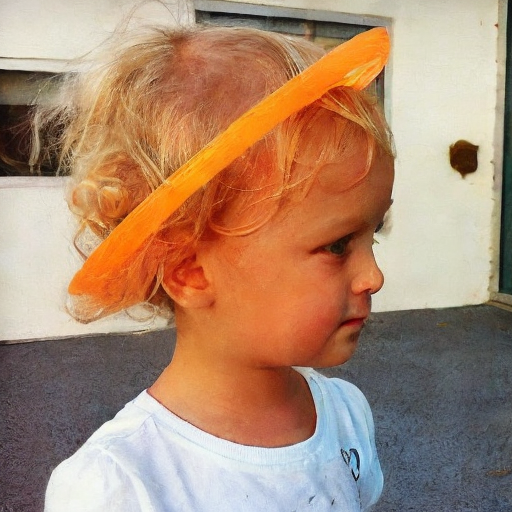}
      \put(3,6){\colorbox{black!60}{\scriptsize\color{white} (c) True/1.0000}}
    \end{overpic}
  \end{minipage}%
  \begin{minipage}[b]{.1245\textwidth}\centering
    \begin{overpic}[width=\linewidth]{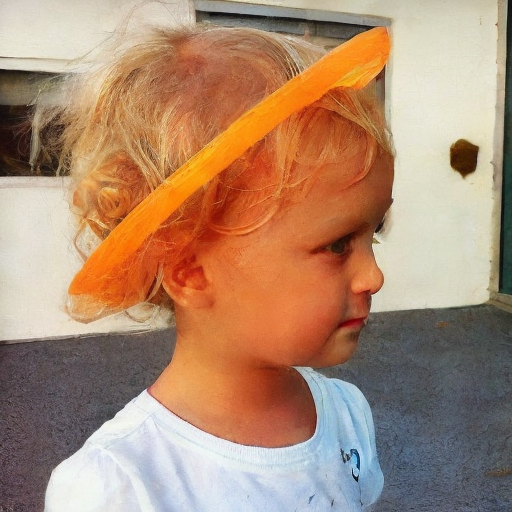}
      \put(3,6){\colorbox{black!60}{\scriptsize\color{white} (d) True/1.0000}}
    \end{overpic}
  \end{minipage}%
  \begin{minipage}[b]{.1245\textwidth}\centering
    \begin{overpic}[width=\linewidth]{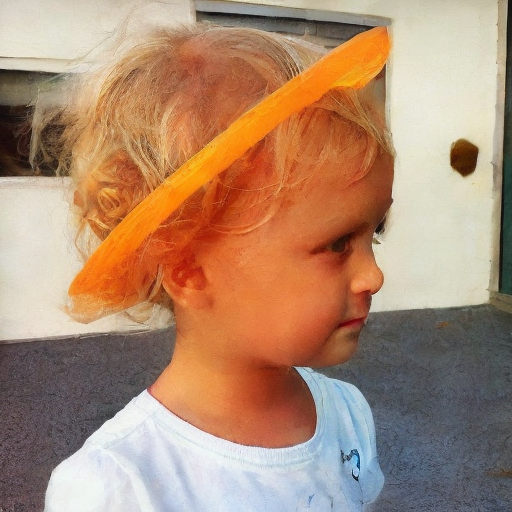}
      \put(3,6){\colorbox{black!60}{\scriptsize\color{white} (e) True/1.0000}}
    \end{overpic}
  \end{minipage}%
  \begin{minipage}[b]{.1245\textwidth}\centering
    \begin{overpic}[width=\linewidth]{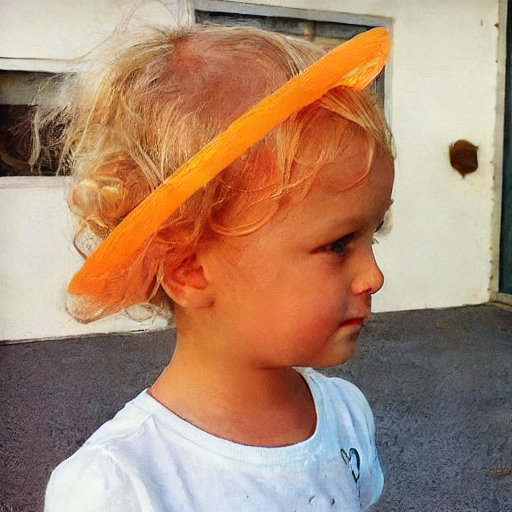}
      \put(3,6){\colorbox{black!60}{\scriptsize\color{white} (f) True/1.0000}}
    \end{overpic}
  \end{minipage}%
  \begin{minipage}[b]{.1245\textwidth}\centering
    \begin{overpic}[width=\linewidth]{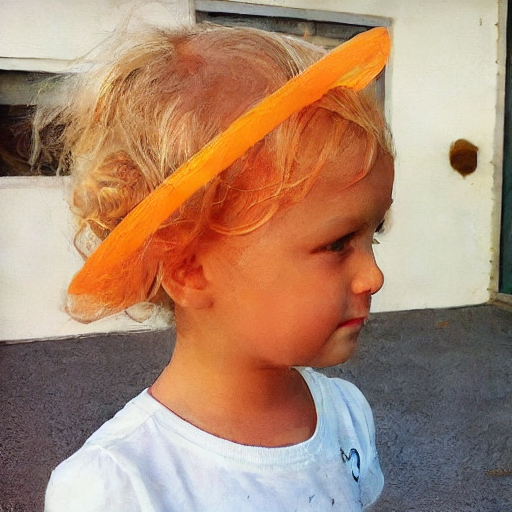}
      \put(3,6){\colorbox{black!60}{\scriptsize\color{white} (g) True/1.0000}}
    \end{overpic}
  \end{minipage}%
  \begin{minipage}[b]{.1245\textwidth}\centering
    \begin{overpic}[width=\linewidth]{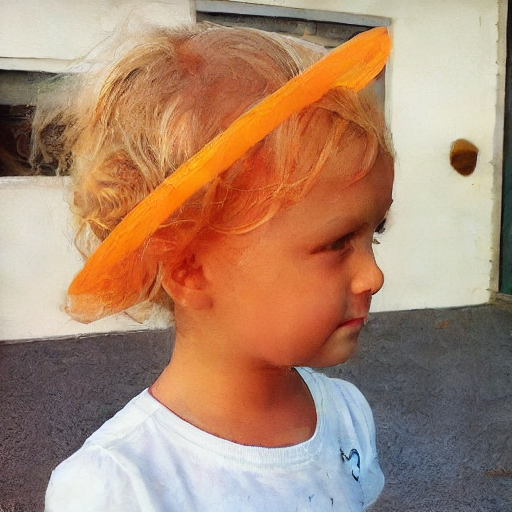}
      \put(3,6){\colorbox{black!60}{\scriptsize\color{white} (h) True/1.0000}}
    \end{overpic}
  \end{minipage}%

  \par

  \begin{minipage}[b]{.1245\textwidth}\centering
    \begin{overpic}[width=\linewidth]{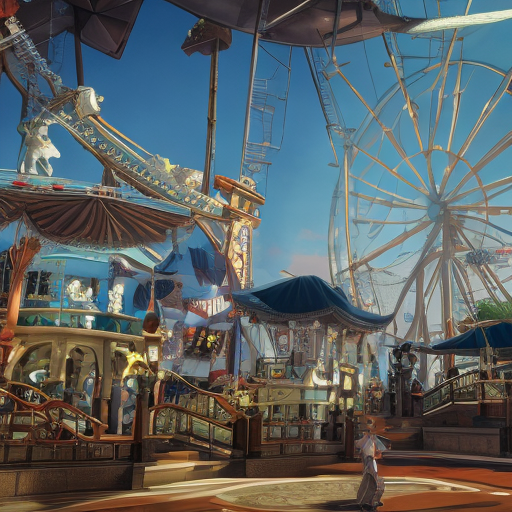}
      \put(3,6){\colorbox{black!60}{\scriptsize\color{white} (a) True/1.0000}}
    \end{overpic}
  \end{minipage}%
  \begin{minipage}[b]{.1245\textwidth}\centering
    \begin{overpic}[width=\linewidth]{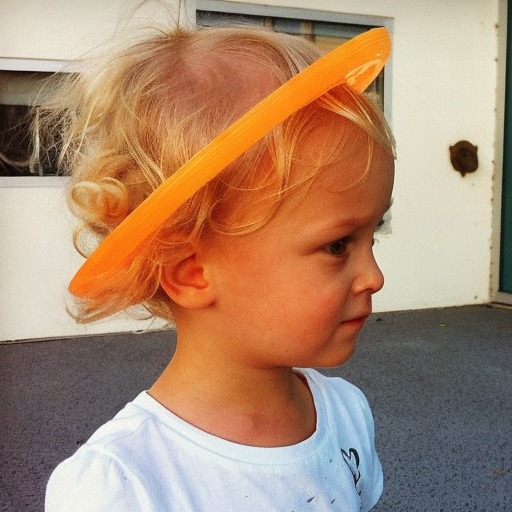}
      \put(3,6){\colorbox{black!60}{\scriptsize\color{white} (b) False/0.5252}}
    \end{overpic}
  \end{minipage}%
  \begin{minipage}[b]{.1245\textwidth}\centering
    \begin{overpic}[width=\linewidth]{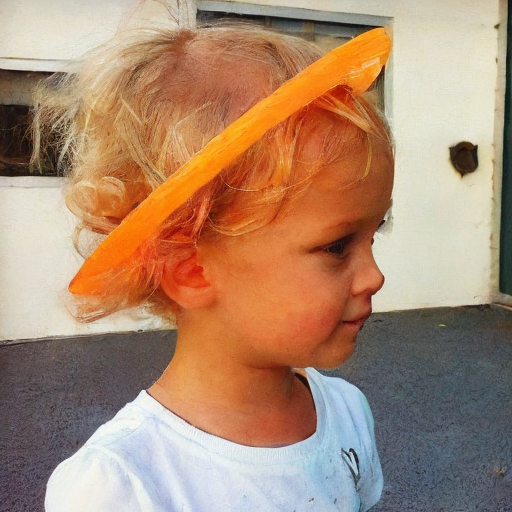}
      \put(3,6){\colorbox{black!60}{\scriptsize\color{white} (c) False/0.5078}}
    \end{overpic}
  \end{minipage}%
  \begin{minipage}[b]{.1245\textwidth}\centering
    \begin{overpic}[width=\linewidth]{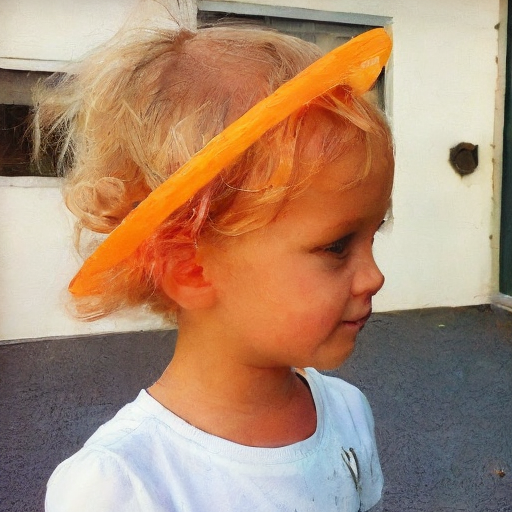}
      \put(3,6){\colorbox{black!60}{\scriptsize\color{white} (d) False/0.5078}}
    \end{overpic}
  \end{minipage}%
  \begin{minipage}[b]{.1245\textwidth}\centering
    \begin{overpic}[width=\linewidth]{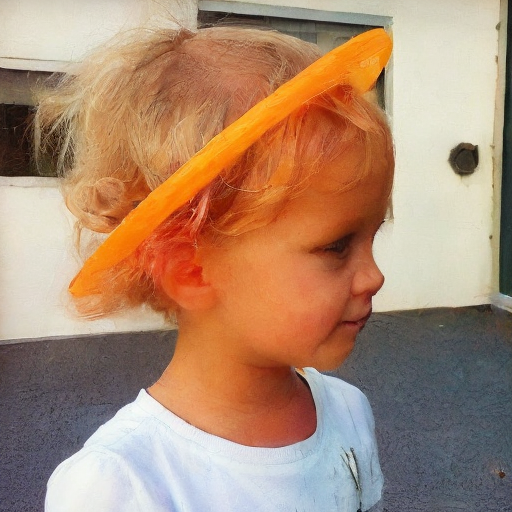}
      \put(3,6){\colorbox{black!60}{\scriptsize\color{white} (e) False/0.5078}}
    \end{overpic}
  \end{minipage}%
  \begin{minipage}[b]{.1245\textwidth}\centering
    \begin{overpic}[width=\linewidth]{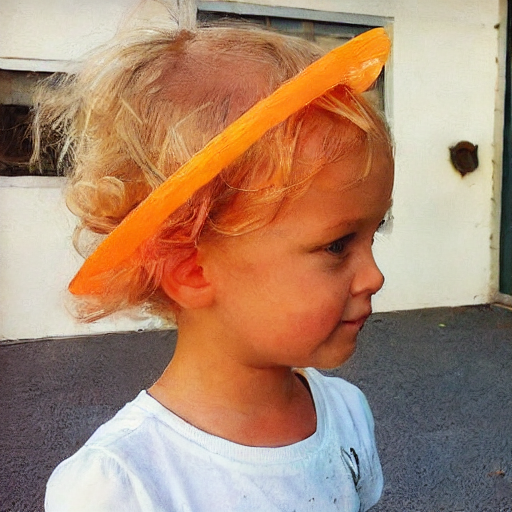}
      \put(3,6){\colorbox{black!60}{\scriptsize\color{white} (f) False/0.5117}}
    \end{overpic}
  \end{minipage}%
  \begin{minipage}[b]{.1245\textwidth}\centering
    \begin{overpic}[width=\linewidth]{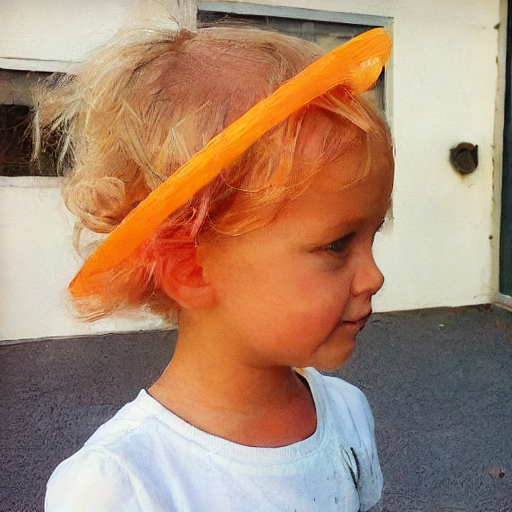}
      \put(3,6){\colorbox{black!60}{\scriptsize\color{white} (g) False/0.5117}}
    \end{overpic}
  \end{minipage}%
  \begin{minipage}[b]{.1245\textwidth}\centering
    \begin{overpic}[width=\linewidth]{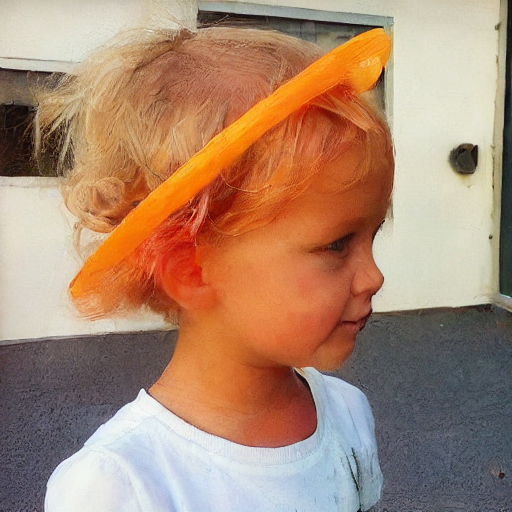}
      \put(3,6){\colorbox{black!60}{\scriptsize\color{white} (h) False/0.5117}}
    \end{overpic}
  \end{minipage}%

\par
  \begin{minipage}[b]{.1245\textwidth}\centering
    \begin{overpic}[width=\linewidth]{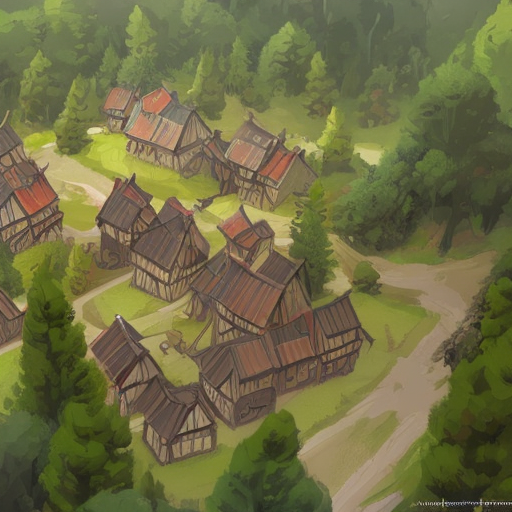}
      \put(3,6){\colorbox{black!60}{\scriptsize\color{white} (a) True/1.0000}}
    \end{overpic}
  \end{minipage}%
  \begin{minipage}[b]{.1245\textwidth}\centering
    \begin{overpic}[width=\linewidth]{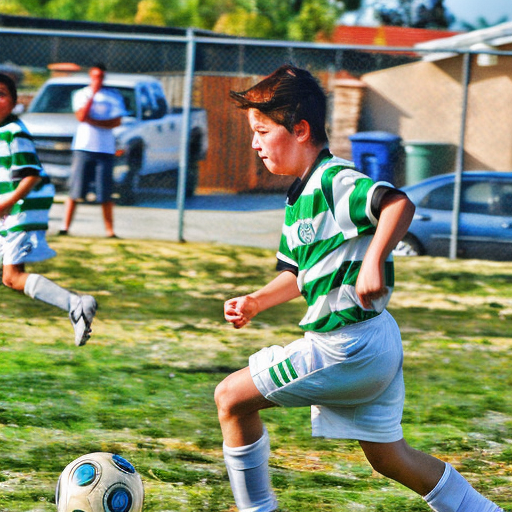}
      \put(3,6){\colorbox{black!60}{\scriptsize\color{white} (b) False/0.4980}}
    \end{overpic}
  \end{minipage}%
  \begin{minipage}[b]{.1245\textwidth}\centering
    \begin{overpic}[width=\linewidth]{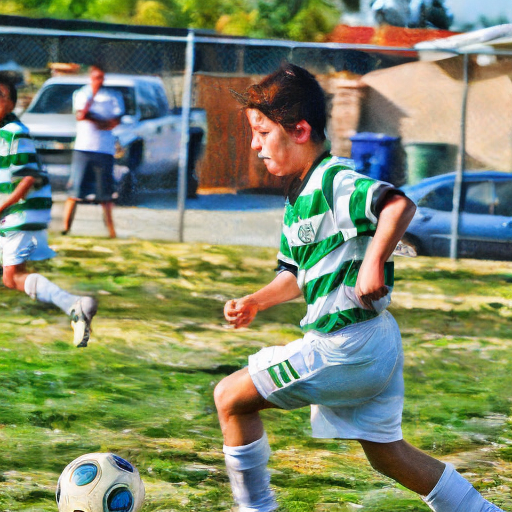}
      \put(3,6){\colorbox{black!60}{\scriptsize\color{white} (c) True/1.0000}}
    \end{overpic}
  \end{minipage}%
  \begin{minipage}[b]{.1245\textwidth}\centering
    \begin{overpic}[width=\linewidth]{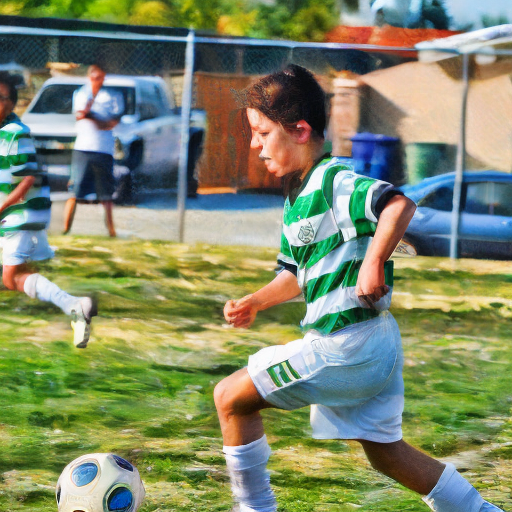}
      \put(3,6){\colorbox{black!60}{\scriptsize\color{white} (d) True/1.0000}}
    \end{overpic}
  \end{minipage}%
  \begin{minipage}[b]{.1245\textwidth}\centering
    \begin{overpic}[width=\linewidth]{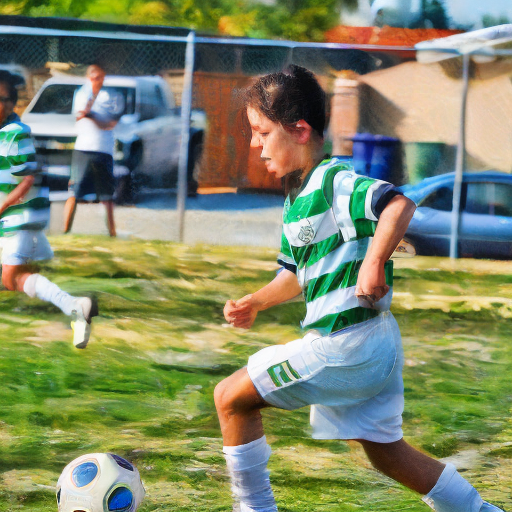}
      \put(3,6){\colorbox{black!60}{\scriptsize\color{white} (e) True/1.0000}}
    \end{overpic}
  \end{minipage}%
  \begin{minipage}[b]{.1245\textwidth}\centering
    \begin{overpic}[width=\linewidth]{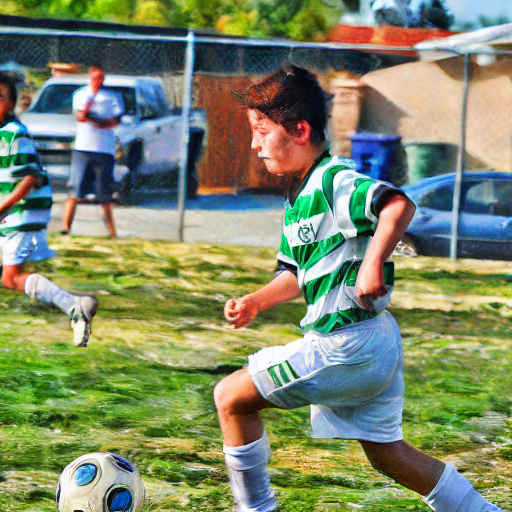}
      \put(3,6){\colorbox{black!60}{\scriptsize\color{white} (f) True/1.0000}}
    \end{overpic}
  \end{minipage}%
  \begin{minipage}[b]{.1245\textwidth}\centering
    \begin{overpic}[width=\linewidth]{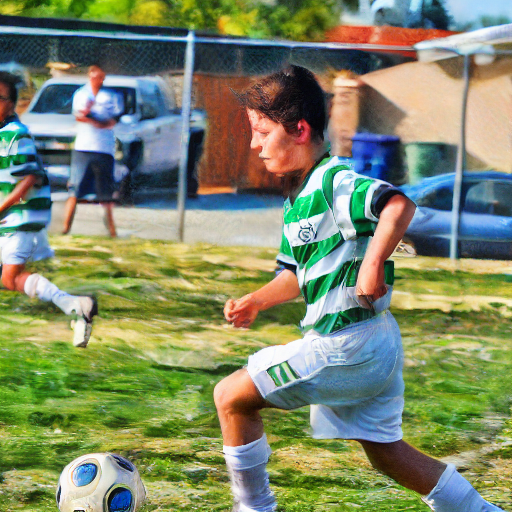}
      \put(3,6){\colorbox{black!60}{\scriptsize\color{white} (g) True/1.0000}}
    \end{overpic}
  \end{minipage}%
  \begin{minipage}[b]{.1245\textwidth}\centering
    \begin{overpic}[width=\linewidth]{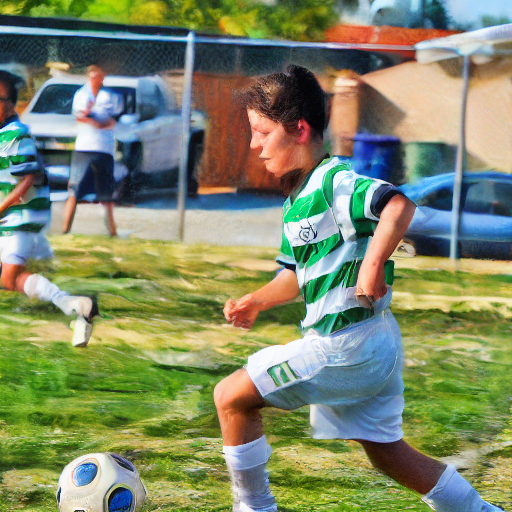}
      \put(3,6){\colorbox{black!60}{\scriptsize\color{white} (h) True/1.0000}}
    \end{overpic}
  \end{minipage}%

  \par

  \begin{minipage}[b]{.1245\textwidth}\centering
    \begin{overpic}[width=\linewidth]{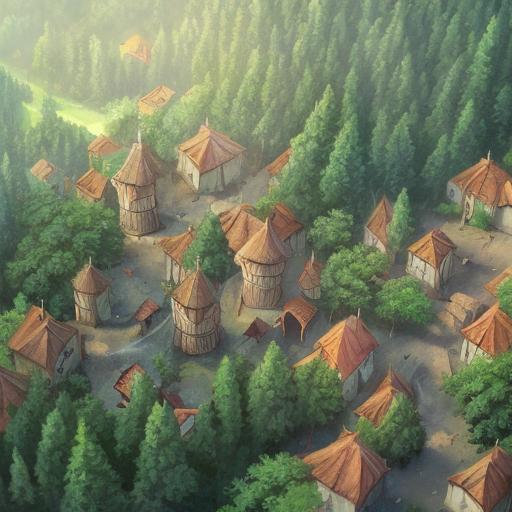}
      \put(3,6){\colorbox{black!60}{\scriptsize\color{white} (a) True/1.0000}}
    \end{overpic}
  \end{minipage}%
  \begin{minipage}[b]{.1245\textwidth}\centering
    \begin{overpic}[width=\linewidth]{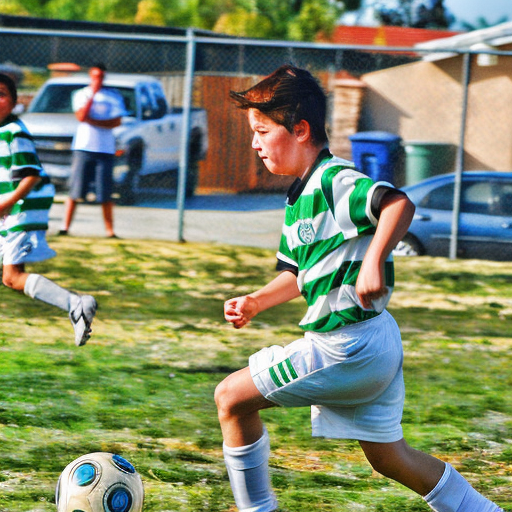}
      \put(3,6){\colorbox{black!60}{\scriptsize\color{white} (b) False/0.4980}}
    \end{overpic}
  \end{minipage}%
  \begin{minipage}[b]{.1245\textwidth}\centering
    \begin{overpic}[width=\linewidth]{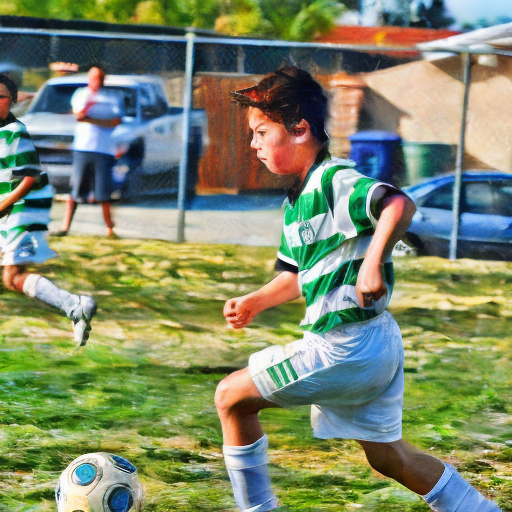}
      \put(3,6){\colorbox{black!60}{\scriptsize\color{white} (c) False/0.4570}}
    \end{overpic}
  \end{minipage}%
  \begin{minipage}[b]{.1245\textwidth}\centering
    \begin{overpic}[width=\linewidth]{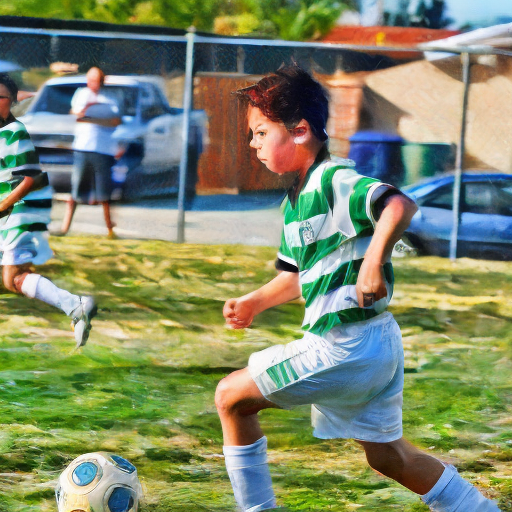}
      \put(3,6){\colorbox{black!60}{\scriptsize\color{white} (d) False/0.4922}}
    \end{overpic}
  \end{minipage}%
  \begin{minipage}[b]{.1245\textwidth}\centering
    \begin{overpic}[width=\linewidth]{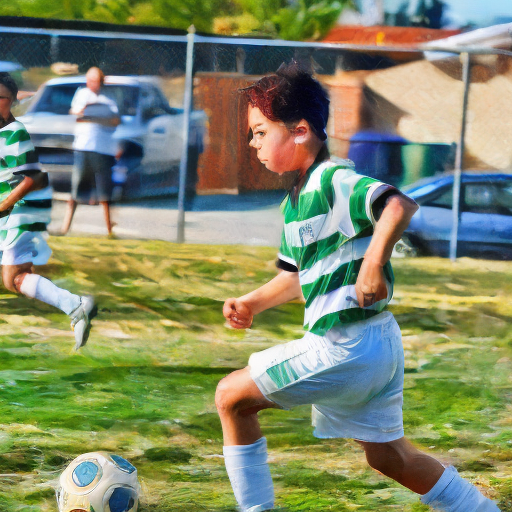}
      \put(3,6){\colorbox{black!60}{\scriptsize\color{white} (e) False/0.4453}}
    \end{overpic}
  \end{minipage}%
  \begin{minipage}[b]{.1245\textwidth}\centering
    \begin{overpic}[width=\linewidth]{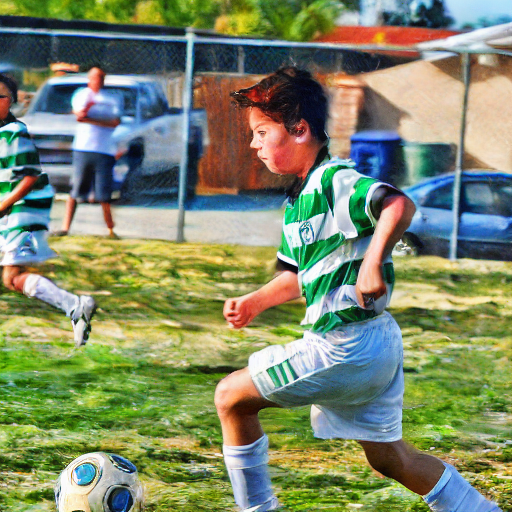}
      \put(3,6){\colorbox{black!60}{\scriptsize\color{white} (f) False/0.4922}}
    \end{overpic}
  \end{minipage}%
  \begin{minipage}[b]{.1245\textwidth}\centering
    \begin{overpic}[width=\linewidth]{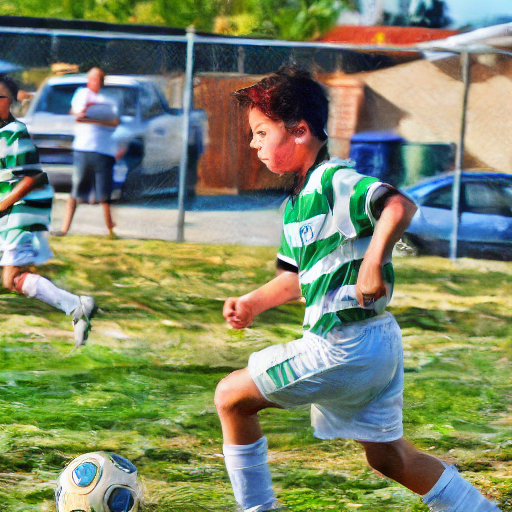}
      \put(3,6){\colorbox{black!60}{\scriptsize\color{white} (g) False/0.5039}}
    \end{overpic}
  \end{minipage}%
  \begin{minipage}[b]{.1245\textwidth}\centering
    \begin{overpic}[width=\linewidth]{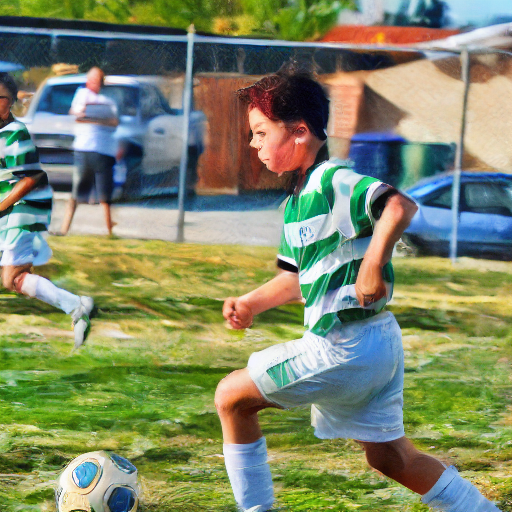}
      \put(3,6){\colorbox{black!60}{\scriptsize\color{white} (h) False/0.4883}}
    \end{overpic}
  \end{minipage}%
  \caption{Examples of imprinting attack results. Rows 1–8 correspond, respectively, to \textbf{Tree-Ring}, \textbf{Tree-Ring+SemBind}, \textbf{Gaussian Shading}, \textbf{Gaussian Shading+SemBind},  \textbf{PRC},  \textbf{PRC+SemBind}, \textbf{Gaussian-Shading++}, and \textbf{Gaussian-Shading++ + SemBind}.
  Panel labels show Detect/Decode outcomes. 
  (a) watermarked image; (b) target image; (c)/(d)/(e): attacker model SD~2.1 with 50/100/150 steps; 
  (f)/(g)/(h): attacker model SD~1.5 with 50/100/150 steps.}
  \label{fig:sup_gs++_grid_sdp}
\end{figure*}

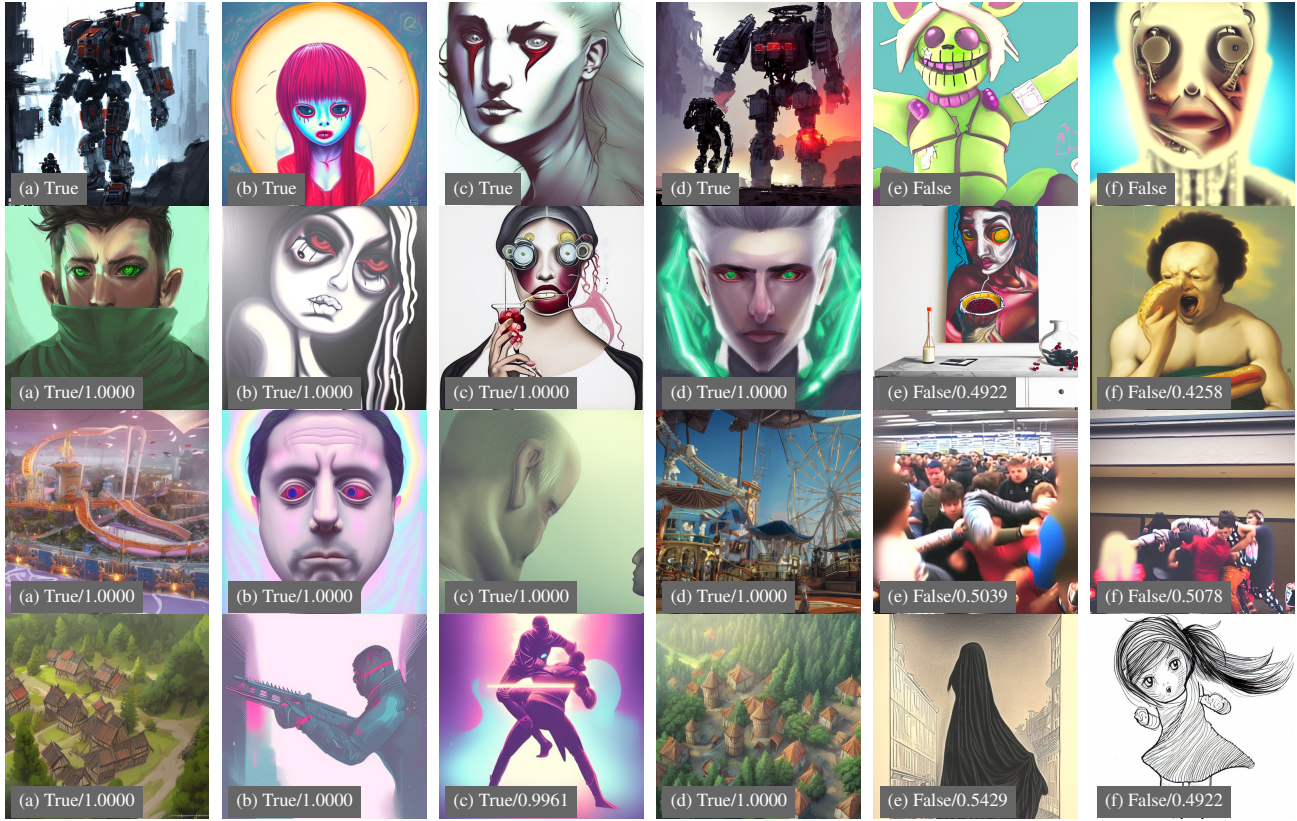
\begin{figure*}[t]
  \centering
  \begin{minipage}[b]{.158\textwidth}\centering
    \begin{overpic}[width=\linewidth]{figs/sup_figs_results_sdp/tr/tr_sdp.png}
      \put(3,6){\colorbox{black!60}{\scriptsize\color{white} (a) True}}
    \end{overpic}
  \end{minipage}\hfill
  \begin{minipage}[b]{.158\textwidth}\centering
    \begin{overpic}[width=\linewidth]{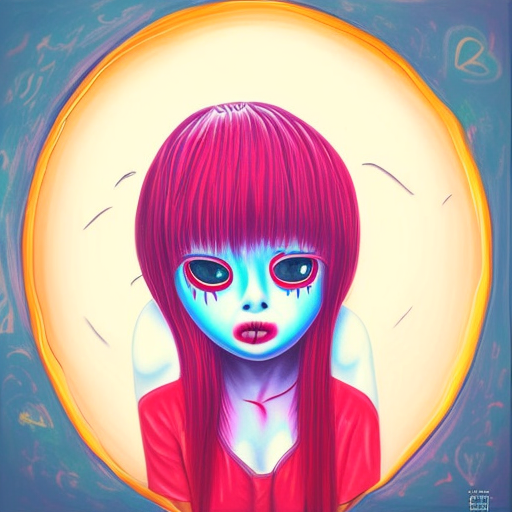}
      \put(3,6){\colorbox{black!60}{\scriptsize\color{white} (b) True}}
    \end{overpic}
  \end{minipage}\hfill
  \begin{minipage}[b]{.158\textwidth}\centering
    \begin{overpic}[width=\linewidth]{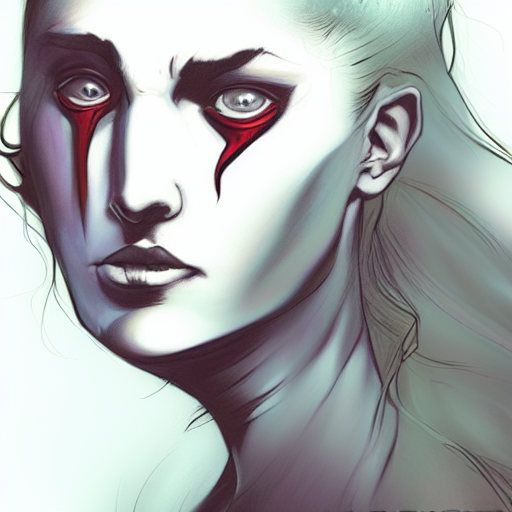}
      \put(3,6){\colorbox{black!60}{\scriptsize\color{white} (c) True}}
    \end{overpic}
  \end{minipage}\hfill
  \begin{minipage}[b]{.158\textwidth}\centering
    \begin{overpic}[width=\linewidth]{figs/sup_figs_results_sdp/tr_s/tr_s_sdp.png}
      \put(3,6){\colorbox{black!60}{\scriptsize\color{white} (d) True}}
    \end{overpic}
  \end{minipage}\hfill
  \begin{minipage}[b]{.158\textwidth}\centering
    \begin{overpic}[width=\linewidth]{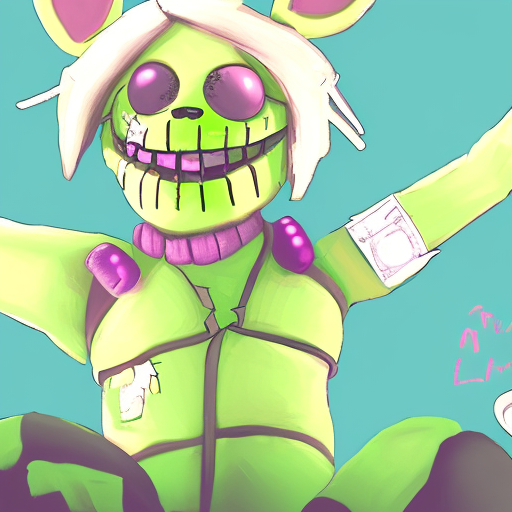}
      \put(3,6){\colorbox{black!60}{\scriptsize\color{white} (e) False}}
    \end{overpic}
  \end{minipage}\hfill
  \begin{minipage}[b]{.158\textwidth}\centering
    \begin{overpic}[width=\linewidth]{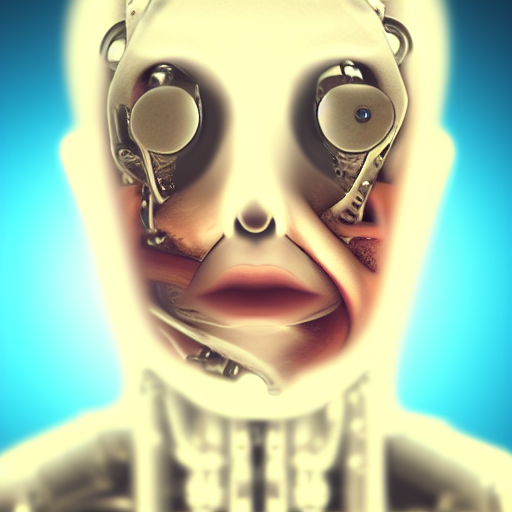}
      \put(3,6){\colorbox{black!60}{\scriptsize\color{white} (f) False}}
    \end{overpic}
  \end{minipage}

  \begin{minipage}[b]{.158\textwidth}\centering
    \begin{overpic}[width=\linewidth]{figs/sup_figs_results_sdp/gs/gs_sdp.png}
      \put(3,6){\colorbox{black!60}{\scriptsize\color{white} (a) True/1.0000}}
    \end{overpic}
  \end{minipage}\hfill
  \begin{minipage}[b]{.158\textwidth}\centering
    \begin{overpic}[width=\linewidth]{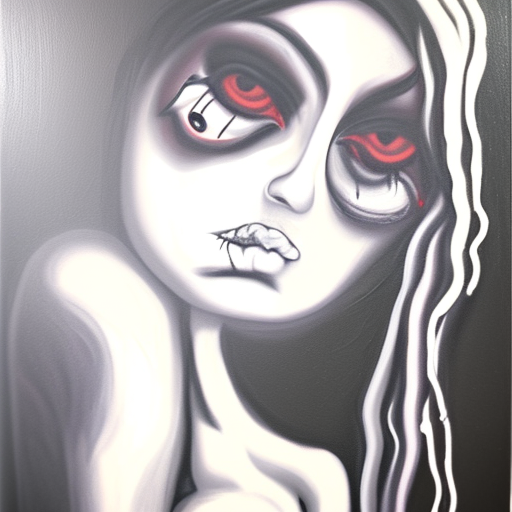}
      \put(3,6){\colorbox{black!60}{\scriptsize\color{white} (b) True/1.0000}}
    \end{overpic}
  \end{minipage}\hfill
  \begin{minipage}[b]{.158\textwidth}\centering
    \begin{overpic}[width=\linewidth]{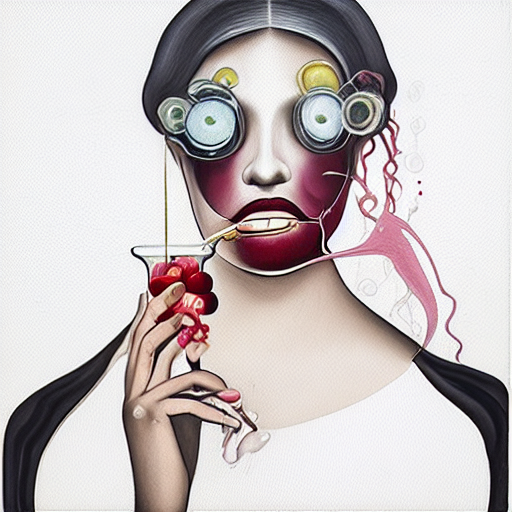}
      \put(3,6){\colorbox{black!60}{\scriptsize\color{white} (c) True/1.0000}}
    \end{overpic}
  \end{minipage}\hfill
  \begin{minipage}[b]{.158\textwidth}\centering
    \begin{overpic}[width=\linewidth]{figs/sup_figs_results_sdp/gs_s/gs_s_sdp.png}
      \put(3,6){\colorbox{black!60}{\scriptsize\color{white} (d) True/1.0000}}
    \end{overpic}
  \end{minipage}\hfill
  \begin{minipage}[b]{.158\textwidth}\centering
    \begin{overpic}[width=\linewidth]{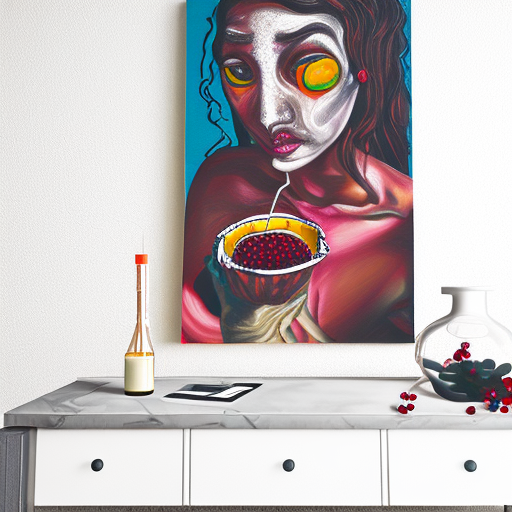}
      \put(3,6){\colorbox{black!60}{\scriptsize\color{white} (e) False/0.4922}}
    \end{overpic}
  \end{minipage}\hfill
  \begin{minipage}[b]{.158\textwidth}\centering
    \begin{overpic}[width=\linewidth]{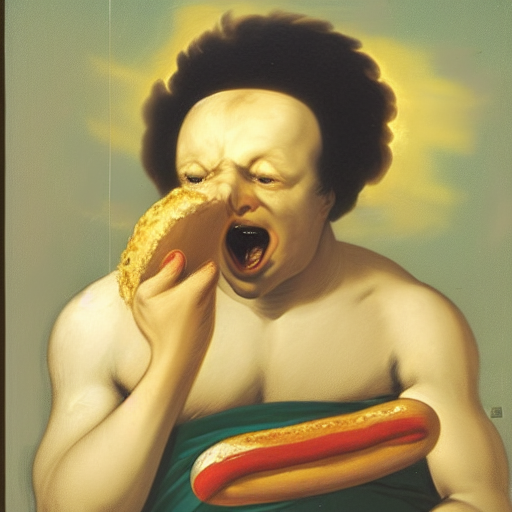}
      \put(3,6){\colorbox{black!60}{\scriptsize\color{white} (f) False/0.4258}}
    \end{overpic}
  \end{minipage}

  \begin{minipage}[b]{.158\textwidth}\centering
    \begin{overpic}[width=\linewidth]{figs/sup_figs_results_sdp/prc/prc_sdp.png}
      \put(3,6){\colorbox{black!60}{\scriptsize\color{white} (a) True/1.0000}}
    \end{overpic}
  \end{minipage}\hfill
  \begin{minipage}[b]{.158\textwidth}\centering
    \begin{overpic}[width=\linewidth]{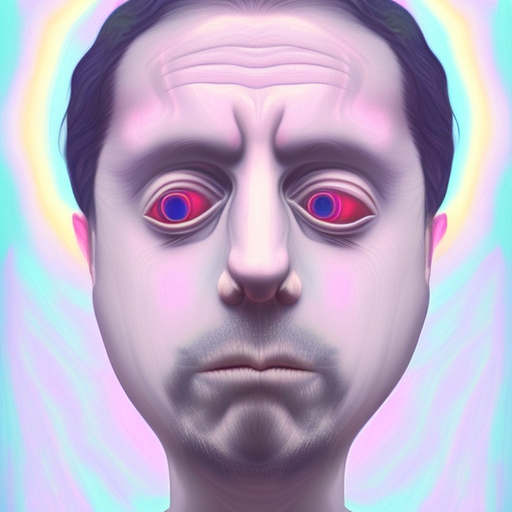}
      \put(3,6){\colorbox{black!60}{\scriptsize\color{white} (b) True/1.0000}}
    \end{overpic}
  \end{minipage}\hfill
  \begin{minipage}[b]{.158\textwidth}\centering
    \begin{overpic}[width=\linewidth]{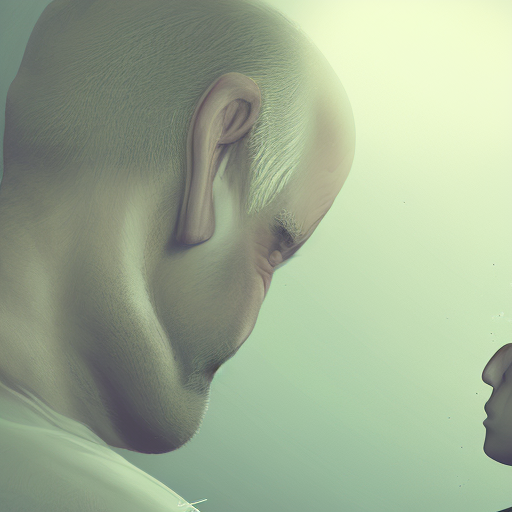}
      \put(3,6){\colorbox{black!60}{\scriptsize\color{white} (c) True/1.0000}}
    \end{overpic}
  \end{minipage}\hfill
  \begin{minipage}[b]{.158\textwidth}\centering
    \begin{overpic}[width=\linewidth]{figs/sup_figs_results_sdp/prc_s/prc_s_sdp.png}
      \put(3,6){\colorbox{black!60}{\scriptsize\color{white} (d) True/1.0000}}
    \end{overpic}
  \end{minipage}\hfill
  \begin{minipage}[b]{.158\textwidth}\centering
    \begin{overpic}[width=\linewidth]{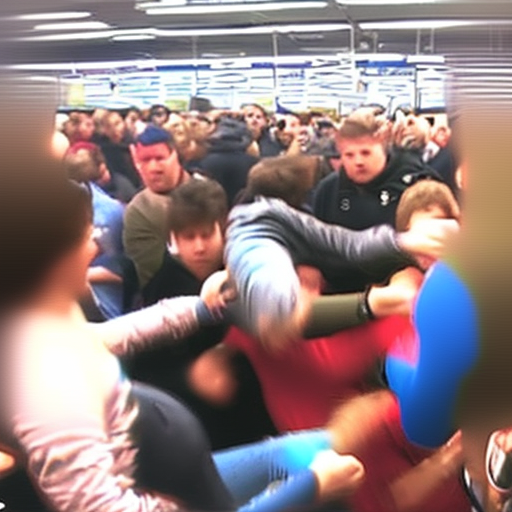}
      \put(3,6){\colorbox{black!60}{\scriptsize\color{white} (e) False/0.5039}}
    \end{overpic}
  \end{minipage}\hfill
  \begin{minipage}[b]{.158\textwidth}\centering
    \begin{overpic}[width=\linewidth]{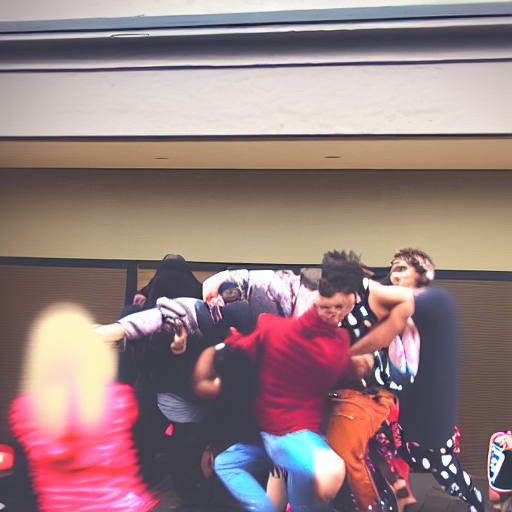}
      \put(3,6){\colorbox{black!60}{\scriptsize\color{white} (f) False/0.5078}}
    \end{overpic}
  \end{minipage}

  \begin{minipage}[b]{.158\textwidth}\centering
    \begin{overpic}[width=\linewidth]{figs/sup_figs_results_sdp/gsp/gsp_sdp.png}
      \put(3,6){\colorbox{black!60}{\scriptsize\color{white} (a) True/1.0000}}
    \end{overpic}
  \end{minipage}\hfill
  \begin{minipage}[b]{.158\textwidth}\centering
    \begin{overpic}[width=\linewidth]{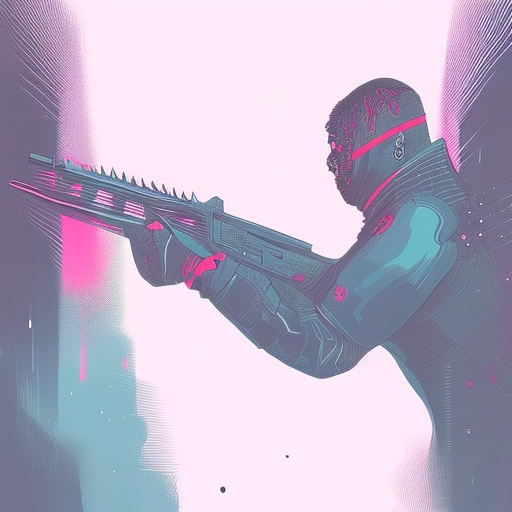}
      \put(3,6){\colorbox{black!60}{\scriptsize\color{white} (b) True/1.0000}}
    \end{overpic}
  \end{minipage}\hfill
  \begin{minipage}[b]{.158\textwidth}\centering
    \begin{overpic}[width=\linewidth]{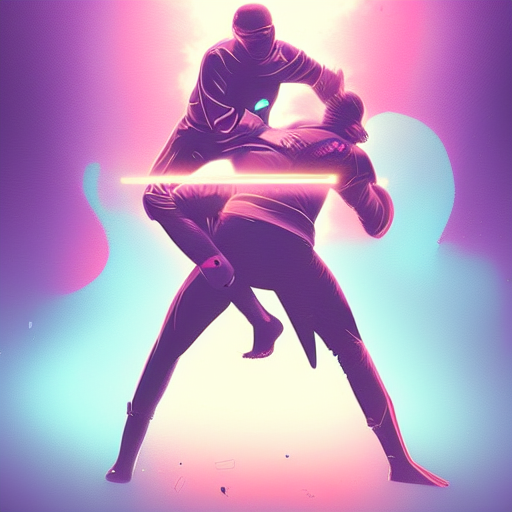}
      \put(3,6){\colorbox{black!60}{\scriptsize\color{white} (c) True/0.9961}}
    \end{overpic}
  \end{minipage}\hfill
  \begin{minipage}[b]{.158\textwidth}\centering
    \begin{overpic}[width=\linewidth]{figs/sup_figs_results_sdp/gsp_s/gsp_s_sdp.png}
      \put(3,6){\colorbox{black!60}{\scriptsize\color{white} (d) True/1.0000}}
    \end{overpic}
  \end{minipage}\hfill
  \begin{minipage}[b]{.158\textwidth}\centering
    \begin{overpic}[width=\linewidth]{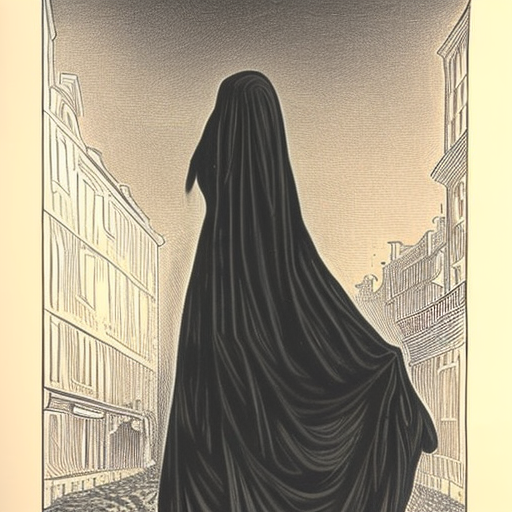}
      \put(3,6){\colorbox{black!60}{\scriptsize\color{white} (e) False/0.5429}}
    \end{overpic}
  \end{minipage}\hfill
  \begin{minipage}[b]{.158\textwidth}\centering
    \begin{overpic}[width=\linewidth]{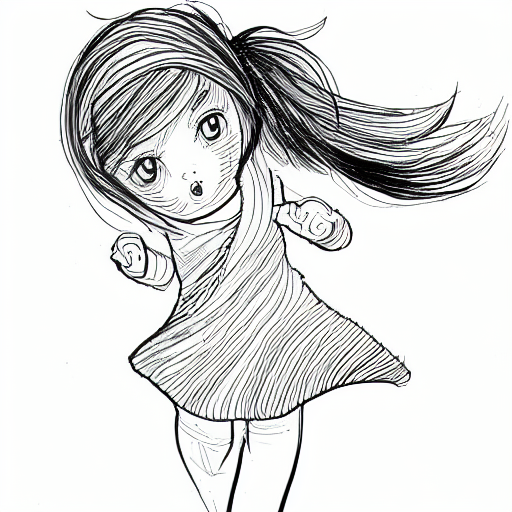}
      \put(3,6){\colorbox{black!60}{\scriptsize\color{white} (f) False/0.4922}}
    \end{overpic}
  \end{minipage}

  \caption{Examples of \textbf{reprompting} attack results. Rows 1–4 show, from the left three columns to the right three columns, respectively: \textbf{Tree-Ring} and \textbf{Tree-Ring+SemBind}; \textbf{Gaussian Shading} and \textbf{Gaussian Shading+SemBind}; \textbf{PRC} and \textbf{PRC+SemBind}; \textbf{Gaussian-Shading++} and \textbf{Gaussian-Shading++ + SemBind}.
  Panel labels show Detect/Decode outcomes.
  (a) watermarked image;
  (b) reprompting result with attacker model \textbf{SD~2.1};
  (c) reprompting result with attacker model \textbf{SD~1.5};
  (d) watermarked image with SemBind;
  (e) reprompting result with \textbf{SD~2.1};
  (f) reprompting result with \textbf{SD~1.5}.}
  \label{fig:reprompt-tr1}
\end{figure*}

%% file: example_paper.bib
@article{ho2020denoising,
  title={Denoising diffusion probabilistic models},
  author={Ho, Jonathan and Jain, Ajay and Abbeel, Pieter},
  journal={Advances in neural information processing systems},
  volume={33},
  pages={6840--6851},
  year={2020}
}

@inproceedings{sohl2015deep,
  title={Deep unsupervised learning using nonequilibrium thermodynamics},
  author={Sohl-Dickstein, Jascha and Weiss, Eric and Maheswaranathan, Niru and Ganguli, Surya},
  booktitle={International conference on machine learning},
  pages={2256--2265},
  year={2015},
  organization={pmlr}
}

@article{song2019generative,
  title={Generative modeling by estimating gradients of the data distribution},
  author={Song, Yang and Ermon, Stefano},
  journal={Advances in neural information processing systems},
  volume={32},
  year={2019}
}

@misc{europol2022facing,
  title={Facing reality? Law enforcement and the challenge of deepfakes},
  author={Europol, Z},
  year={2022},
  publisher={Europol Innovation Lab; Publications Office of the European Union Luxembourg}
}

@article{goldstein2021disinformation,
  title={How disinformation evolved in 2020},
  author={Goldstein, Josh A and Grossman, Shelby},
  journal={Brookings TechStream, January},
  volume={4},
  year={2021}
}

@article{biden2023executive,
  title={Executive order on the safe, secure, and trustworthy development and use of artificial intelligence},
  author={Biden, Joseph R},
  year={2023}
}

@article{act2024regulation,
  title={REGULATION (EU) 2024/2847 OF THE EUROPEAN PARLIAMENT AND OF THE COUNCIL},
  author={Act, Resilience},
  journal={Regulation (eu)},
  year={2024}
}

@article{bartz2023openai,
  title={OpenAI, Google, others pledge to watermark AI content for safety, White House says},
  author={Bartz, Diane and Hu, Krystal},
  journal={Reuters. 21iyulya 2023g.—URL: https://www. reuters. com/technology/openai-googleothers-pledge-watermark-ai-content-safety-whitehouse-2023-07-21 (data obrashcheniya: 18.08. 2023)},
  year={2023}
}

@article{clegg2024labeling,
  title={Labeling AI-generated images on Facebook, Instagram and Threads},
  author={Clegg, Nick},
  journal={Meta},
  volume={6},
  year={2024}
}

@article{cox2008digital,
  title={Digital watermarking},
  author={Cox, Ingemar J and Miller, Matthew L and Bloom, Jeffrey A and Fridrich, Jessica and Kalker, Ton},
  journal={Morgan Kaufmann Publishers},
  volume={54},
  pages={56--59},
  year={2008},
  publisher={Springer}
}

@article{zhang2019robust,
  title={Robust invisible video watermarking with attention},
  author={Zhang, Kevin Alex and Xu, Lei and Cuesta-Infante, Alfredo and Veeramachaneni, Kalyan},
  journal={arXiv preprint arXiv:1909.01285},
  year={2019}
}

@article{cui2023diffusionshield,
  title={Diffusionshield: A watermark for copyright protection against generative diffusion models},
  author={Cui, Yingqian and Ren, Jie and Xu, Han and He, Pengfei and Liu, Hui and Sun, Lichao and Xing, Yue and Tang, Jiliang},
  journal={arXiv preprint arXiv:2306.04642},
  year={2023}
}

@inproceedings{fernandez2023stable,
  title={The stable signature: Rooting watermarks in latent diffusion models},
  author={Fernandez, Pierre and Couairon, Guillaume and J{\'e}gou, Herv{\'e} and Douze, Matthijs and Furon, Teddy},
  booktitle={Proceedings of the IEEE/CVF International Conference on Computer Vision},
  pages={22466--22477},
  year={2023}
}

@inproceedings{xiong2023flexible,
  title={Flexible and secure watermarking for latent diffusion model},
  author={Xiong, Cheng and Qin, Chuan and Feng, Guorui and Zhang, Xinpeng},
  booktitle={Proceedings of the 31st ACM International Conference on Multimedia},
  pages={1668--1676},
  year={2023}
}

@article{zhao2023recipe,
  title={A recipe for watermarking diffusion models},
  author={Zhao, Yunqing and Pang, Tianyu and Du, Chao and Yang, Xiao and Cheung, Ngai-Man and Lin, Min},
  journal={arXiv preprint arXiv:2303.10137},
  year={2023}
}

@article{wen2023tree,
  title={Tree-rings watermarks: Invisible fingerprints for diffusion images},
  author={Wen, Yuxin and Kirchenbauer, John and Geiping, Jonas and Goldstein, Tom},
  journal={Advances in Neural Information Processing Systems},
  volume={36},
  pages={58047--58063},
  year={2023}
}

@article{labs2025flux,
  title={FLUX. 1 Kontext: Flow Matching for In-Context Image Generation and Editing in Latent Space},
  author={Labs, Black Forest and Batifol, Stephen and Blattmann, Andreas and Boesel, Frederic and Consul, Saksham and Diagne, Cyril and Dockhorn, Tim and English, Jack and English, Zion and Esser, Patrick and others},
  journal={arXiv preprint arXiv:2506.15742},
  year={2025}
}

@inproceedings{yang2024gaussian,
  title={Gaussian shading: Provable performance-lossless image watermarking for diffusion models},
  author={Yang, Zijin and Zeng, Kai and Chen, Kejiang and Fang, Han and Zhang, Weiming and Yu, Nenghai},
  booktitle={Proceedings of the IEEE/CVF Conference on Computer Vision and Pattern Recognition},
  pages={12162--12171},
  year={2024}
}

@inproceedings{gunnundetectable,
  title={An Undetectable Watermark for Generative Image Models},
  author={Gunn, Sam and Zhao, Xuandong and Song, Dawn},
  booktitle={The Thirteenth International Conference on Learning Representations}
}

@article{yang2025gaussian,
  title={Gaussian Shading++: Rethinking the Realistic Deployment Challenge of Performance-Lossless Image Watermark for Diffusion Models},
  author={Yang, Zijin and Zhang, Xin and Chen, Kejiang and Zeng, Kai and Yao, Qiyi and Fang, Han and Zhang, Weiming and Yu, Nenghai},
  journal={arXiv preprint arXiv:2504.15026},
  year={2025}
}

@inproceedings{christ2024undetectable,
  title={Undetectable watermarks for language models},
  author={Christ, Miranda and Gunn, Sam and Zamir, Or},
  booktitle={The Thirty Seventh Annual Conference on Learning Theory},
  pages={1125--1139},
  year={2024},
  organization={PMLR}
}

@inproceedings{muller2025black,
  title={Black-box forgery attacks on semantic watermarks for diffusion models},
  author={M{\"u}ller, Andreas and Lukovnikov, Denis and Thietke, Jonas and Fischer, Asja and Quiring, Erwin},
  booktitle={Proceedings of the Computer Vision and Pattern Recognition Conference},
  pages={20937--20946},
  year={2025}
}

@article{jain2025forging,
  title={Forging and Removing Latent-Noise Diffusion Watermarks Using a Single Image},
  author={Jain, Anubhav and Kobayashi, Yuya and Murata, Naoki and Takida, Yuhta and Shibuya, Takashi and Mitsufuji, Yuki and Cohen, Niv and Memon, Nasir and Togelius, Julian},
  journal={arXiv preprint arXiv:2504.20111},
  year={2025}
}

@article{heusel2017gans,
  title={Gans trained by a two time-scale update rule converge to a local nash equilibrium},
  author={Heusel, Martin and Ramsauer, Hubert and Unterthiner, Thomas and Nessler, Bernhard and Hochreiter, Sepp},
  journal={Advances in neural information processing systems},
  volume={30},
  year={2017}
}

@inproceedings{radford2021learning,
  title={Learning transferable visual models from natural language supervision},
  author={Radford, Alec and Kim, Jong Wook and Hallacy, Chris and Ramesh, Aditya and Goh, Gabriel and Agarwal, Sandhini and Sastry, Girish and Askell, Amanda and Mishkin, Pamela and Clark, Jack and others},
  booktitle={International conference on machine learning},
  pages={8748--8763},
  year={2021},
  organization={PmLR}
}

@inproceedings{rombach2022high,
  title={High-resolution image synthesis with latent diffusion models},
  author={Rombach, Robin and Blattmann, Andreas and Lorenz, Dominik and Esser, Patrick and Ommer, Bj{\"o}rn},
  booktitle={Proceedings of the IEEE/CVF conference on computer vision and pattern recognition},
  pages={10684--10695},
  year={2022}
}

@article{lu2022dpm,
  title={Dpm-solver: A fast ode solver for diffusion probabilistic model sampling in around 10 steps},
  author={Lu, Cheng and Zhou, Yuhao and Bao, Fan and Chen, Jianfei and Li, Chongxuan and Zhu, Jun},
  journal={Advances in neural information processing systems},
  volume={35},
  pages={5775--5787},
  year={2022}
}

@inproceedings{mokady2023null,
  title={Null-text inversion for editing real images using guided diffusion models},
  author={Mokady, Ron and Hertz, Amir and Aberman, Kfir and Pritch, Yael and Cohen-Or, Daniel},
  booktitle={Proceedings of the IEEE/CVF conference on computer vision and pattern recognition},
  pages={6038--6047},
  year={2023}
}

@article{yang2024can,
  title={Can simple averaging defeat modern watermarks?},
  author={Yang, Pei and Ci, Hai and Song, Yiren and Shou, Mike Zheng},
  journal={Advances in Neural Information Processing Systems},
  volume={37},
  pages={56644--56673},
  year={2024}
}

@article{oquab2023dinov2,
  title={Dinov2: Learning robust visual features without supervision},
  author={Oquab, Maxime and Darcet, Timoth{\'e}e and Moutakanni, Th{\'e}o and Vo, Huy and Szafraniec, Marc and Khalidov, Vasil and Fernandez, Pierre and Haziza, Daniel and Massa, Francisco and El-Nouby, Alaaeldin and others},
  journal={arXiv preprint arXiv:2304.07193},
  year={2023}
}

@article{khosla2020supervised,
  title={Supervised contrastive learning},
  author={Khosla, Prannay and Teterwak, Piotr and Wang, Chen and Sarna, Aaron and Tian, Yonglong and Isola, Phillip and Maschinot, Aaron and Liu, Ce and Krishnan, Dilip},
  journal={Advances in neural information processing systems},
  volume={33},
  pages={18661--18673},
  year={2020}
}

@inproceedings{hong2024exact,
  title={On exact inversion of dpm-solvers},
  author={Hong, Seongmin and Lee, Kyeonghyun and Jeon, Suh Yoon and Bae, Hyewon and Chun, Se Young},
  booktitle={Proceedings of the IEEE/CVF Conference on Computer Vision and Pattern Recognition},
  pages={7069--7078},
  year={2024}
}

@article{wei2024exploring,
  title={Exploring hierarchical information in hyperbolic space for self-supervised image hashing},
  author={Wei, Rukai and Liu, Yu and Song, Jingkuan and Xie, Yanzhao and Zhou, Ke},
  journal={IEEE Transactions on Image Processing},
  volume={33},
  pages={1768--1781},
  year={2024},
  publisher={IEEE}
}

@inproceedings{wang2023deep,
  title={Deep hashing with minimal-distance-separated hash centers},
  author={Wang, Liangdao and Pan, Yan and Liu, Cong and Lai, Hanjiang and Yin, Jian and Liu, Ye},
  booktitle={Proceedings of the IEEE/CVF conference on computer vision and pattern recognition},
  pages={23455--23464},
  year={2023}
}

@inproceedings{shen2024contrastive,
  title={Contrastive transformer masked image hashing for degraded image retrieval},
  author={Shen, Xiaobo and Cai, Haoyu and Gong, Xiuwen and Zheng, Yuhui},
  booktitle={Proceedings of the Thirty-ThirdInternational Joint Conference on Artificial Intelligence},
  year={2024},
  organization={International Joint Conferences on Artificial Intelligence}
}

@article{zhang2025both,
  title={Both Semantics and Reconstruction Matter: Making Representation Encoders Ready for Text-to-Image Generation and Editing},
  author={Zhang, Shilong and Zhang, He and Zhang, Zhifei and Ge, Chongjian and Xue, Shuchen and Liu, Shaoteng and Ren, Mengwei and Kim, Soo Ye and Zhou, Yuqian and Liu, Qing and others},
  journal={arXiv preprint arXiv:2512.17909},
  year={2025}
}

@inproceedings{lin2014microsoft,
  title={Microsoft coco: Common objects in context},
  author={Lin, Tsung-Yi and Maire, Michael and Belongie, Serge and Hays, James and Perona, Pietro and Ramanan, Deva and Doll{\'a}r, Piotr and Zitnick, C Lawrence},
  booktitle={European conference on computer vision},
  pages={740--755},
  year={2014},
  organization={Springer}
}

@article{goldwasser1984probabilistic,
  title={Probabilistic encryption},
  author={Goldwasser, Shafi and Micali, Silvio},
  journal={Journal of Computer and System Sciences},
  volume={28},
  number={2},
  pages={270--299},
  year={1984}
}

@book{goldreich2001foundations,
  title={Foundations of Cryptography, Volume 1: Basic Tools},
  author={Goldreich, Oded},
  year={2001},
  publisher={Cambridge University Press}
}

@inproceedings{sculley2010web,
  title={Web-scale k-means clustering},
  author={Sculley, David},
  booktitle={Proceedings of the 19th international conference on World wide web},
  pages={1177--1178},
  year={2010}
}
